\documentclass[a4paper,UKenglish,cleveref,autoref,thm-restate,numberwithinsect]{lipics-v2021}

\usepackage[T1]{fontenc}

\usepackage{amsfonts}
\usepackage{amsmath}
\usepackage{amssymb}
\usepackage{cite}

\usepackage{tikz}
\usetikzlibrary{quantikz2}

\usepackage{hyperref}
\crefname{algorithm}{Alg.}{Algs.}
\crefname{section}{Sec.}{Secs.}
\crefname{definition}{Def.}{Defs.}
\crefname{table}{Tab.}{Tabs.}
\crefname{example}{Ex.}{Exs.}
\crefname{proposition}{Prop.}{Props.}
\crefname{theorem}{Thm.}{Thms.}
\crefname{corollary}{Cor.}{Cors.}
\crefname{appendix}{Appx.}{Appxs.}

\usepackage{subcaption}

\newcommand{\interp}[1]{[\![#1]\!]}
\newcommand{\pinterp}[1]{[\![#1]\!]_{\textsf{Poly}}}
\newcommand{\SPoly}{\textsf{SPoly}}
\newcommand{\CPoly}{\textsf{CPoly}}

\newcommand{\Circ}{\textsf{Circ}}
\newcommand{\ZCirc}{\Circ_{\mathbb{Z}}}
\newcommand{\win}{\textsf{in}}
\newcommand{\wout}{\textsf{out}}
\newcommand{\Bnd}{\textsf{Bnd}}
\newcommand{\pBnd}{\textsf{Bnd}_{\textsf{Poly}}}
\newcommand{\lcm}{\operatorname{lcm}}
\newcommand{\clcm}{\textsf{circLcm}}
\newcommand{\denom}{\operatorname{denom}}
\newcommand{\FindPhase}{\textsf{FindPhase}}
\newcommand{\zpred}[0]{\textsf{Pred}}

\DeclarePairedDelimiter\ceil{\lceil}{\rceil}

\title{Cutoff Theorems for the Equivalence of Parameterized Quantum Circuits (Extended)}

\author{Neil J. Ross}{Dalhousie University, Canada}{neil.jr.ross@dal.ca}{0000-0003-0941-4333}{}
\author{Scott Wesley}{Dalhousie University, Canada}{scott.wesley@dal.ca}{0000-0002-6708-2122}{}

\authorrunning{N. J. Ross and S. Wesley}
\Copyright{Neil J. Ross and Scott Wesley}



\acknowledgements{We thank Mingkuan Xu for his feedback on the paper and our initial experiments, Linh Dinh for sharing their insights on the number-theoretic aspects of this paper, and Mingsheng Ying for his feedback on an early draft of the paper.}

\nolinenumbers

\EventEditors{Pawe\l{} Gawrychowski, Filip Mazowiecki, and Micha\l{} Skrzypczak}
\EventNoEds{3}
\EventLongTitle{50th International Symposium on Mathematical Foundations of Computer Science (MFCS 2025)}
\EventShortTitle{MFCS 2025}
\EventAcronym{MFCS}
\EventYear{2025}
\EventDate{August 25--29, 2025}
\EventLocation{Warsaw, Poland}
\EventLogo{}
\SeriesVolume{345}
\ArticleNo{84}
\begin{document}

\maketitle

\begin{abstract}
    Many promising quantum algorithms in economics, medical science, and material science rely on circuits that are parameterized by a large number of angles.
    To ensure that these algorithms are efficient, these parameterized circuits must be heavily optimized.
    However, most quantum circuit optimizers are not verified, so this procedure is known to be error-prone.
    For this reason, there is growing interest in the design of equivalence checking algorithms for parameterized quantum circuits.
    In this paper, we define a generalized class of parameterized circuits with arbitrary rotations and show that this problem is decidable for cyclotomic gate sets.
    We propose a cutoff-based procedure which reduces the problem of verifying the equivalence of parameterized quantum circuits to the problem of verifying the equivalence of finitely many parameter-free quantum circuits.
    Because the number of parameter-free circuits grows exponentially with the number of parameters, we also propose a probabilistic variant of the algorithm for cases when the number of parameters is intractably large.
    We show that our techniques extend to equivalence modulo global phase, and describe an efficient angle sampling procedure for cyclotomic gate sets.
\end{abstract}

\section{Introduction}

In quantum mechanics, unitary operators describe how the probability distributions of quantum systems evolve over time.
In quantum computing, primitive operators (known as \emph{quantum gates}) are composed in sequence and parallel, to create \emph{quantum circuits} which prepare quantum systems with desirable probability distributions.
By sampling from these distributions, answers can be obtained to many high-value problems, such as those from economics~\cite{Herman2023}, medical science~\cite{Santagati2024,DingHuang2024}, and material science~\cite{MaGovoni2020}.
In these algorithms, an initial guess is made for the correct probability distribution, and then each sample is used to further refine the distribution.
To make this search tractable, the probability distributions are sampled from a family of parameterized quantum circuits, known as \emph{ansatz circuits}.

In practice, the structure of the ansatz circuit is static, so that the parameters only vary the operators which appear within the circuits.
The parameterized operators within ansatz circuits can be understood geometrically as rotations by arbitrary angles.
As a result, the gate sets used to construct ansatz circuits are necessarily infinite.
In contrast, the gate sets implemented by real quantum computers are finite, due to limitations related to error-correction~\cite{EastinKnill2009}.
This means that for each parameter refinement, the ansatz circuit must be recompiled and optimized again.
However, the compilation and optimization of quantum circuits are known to be highly error-prone~\cite{ZhaoMiao2023,HietalaRand2023}, so it is desirable to verify both the equivalence of the optimized circuit to the original circuit, and more generally, the correctness of each optimization.
In both cases, it is necessary to reason equationally about parameterized relations between quantum circuits.

The problem of parameterized equivalence-checking has been well-studied in the context of distributed system.
Given a set of parameters $P$ and two programs parameterized by $P$, say $C_1$ and $C_2$, the parameterized-equivalence checking problem asks whether $C_1( \theta ) = C_2( \theta ), \forall \theta \in P$.
When $P$ is finite, this problem can be solved by simply testing the elements of $P$.
When $P$ is infinite, one approach to this problem is to find a cutoff $n$ for which checking the equivalence of $C_1$ and $C_2$ for $n$ distinct elements of $P$ implies the equivalence of $C_1$ and $C_2$ for all elements of $P$~\cite{EmersonNamjoshi1995}.
Formally, one tries to find an $n \in \mathbb{N}$ such that for all $D \subseteq P$, if $|D| \ge n$, then $\forall \theta \in D \cdot C_1( \theta ) = C_2( \theta )$ implies $\forall \theta \in P \cdot C_1( \theta ) = C_2( \theta )$.
Typically, the choice of $n$ (and sometimes even $D$) will depend on both $C_1$ and $C_2$, and therefore this technique requires domain-specific insights~(see,~e.g.,~\cite{IpDill1993,KaiserKroening2010,KhalimovJacobs2013,AbdullaHH13,NamjoshiTrefler2016,Wesley2022}).
When $n$ becomes intractably large, probabilistic techniques have also been employed~\cite{DemilloLipton1978}.

Cutoff-based techniques have yet to see wide application in the domain of parameterized quantum circuit equivalence-checking.
In 2020, Miller-Bakewell developed a framework which adapts cutoff-based techniques to quantum circuits~\cite{MillerBakewell2020}, though these techniques have yet to be applied in practice.
The key insight of this work was to note that parameterized quantum circuits are analytic for realistic gate sets, and (up to a change of variable) can often be expressed as matrices over complex Laurent polynomials.
The positive and negative degrees of these Laurent polynomials can be over-approximated in an inductive manner, and correspond to a cutoff for parameterized verification.
The main challenge in applying the Miller-Bakewell framework is to identify an appropriate change-of-variables such that all parameterized matrices become matrices over complex Laurent polynomials.
Once this change-of-variable has been identified, further steps may be taken, such as deriving a closed-form equation for the cutoff.
In Miller-Bakewell's paper, the framework was applied to ZX-, ZW-, and ZH-diagrams, though closed-form bounds were not derived.

In this paper, we propose a cutoff-based technique for quantum circuits with arbitrary rotations with linear arguments.
This technique can be understood as an instantiation of the Miller-Bakewell framework, insofar as each parameterized circuit is realized as a matrix over complex Laurent polynomials.
However, the circuits considered in this paper correspond to ZXW-diagrams (i.e., with matrix exponentiation)~\cite{ShaikhWang2023}, which are not addressed in Miller-Bakewell's original work.
We derive closed-form equations for these cutoffs, which depend only on the coefficients of the parameters in the circuits.
Furthermore, we provide an alternative proof for the correctness of the Miller-Bakewell framework, which depends on the distribution of zeros of Laurent polynomials as opposed to polynomial interpolation.
This change in perspective motivates a probabilistic variant of the Miller-Bakewell framework, which is applicable for circuits with intractably large cutoffs.

In \cref{Sect:Circuits}, we provide the syntax and semantics for our circuit language.
In \cref{Sect:Motivation}, we illustrate our technique on a simple real-world example.
In \cref{Sect:Equiv}, we prove a cutoff theorem, and propose a probabilistic variant.
In \cref{Sect:Implementation}, we identify and solve several challenges faced when implementing this technique.

\section{Background}
\label{Sect:Background}

We write $\mathbb{N}$ for the set of natural numbers (including zero), $\mathbb{Z}$ for the set of integers, $\mathbb{Q}$ for the set of rational numbers, $\mathbb{R}$ for the set of real numbers, and $\mathbb{C}$ for the set of complex numbers.
If $z \in \mathbb{C}$, then $\overline{z}$ denotes the complex conjugate of $z$.
If $n \in \mathbb{N}$, then $[n]$ denotes the set $\{ j \in \mathbb{N} : 1 \le j \le n \}$ so that $[0] = \varnothing$.
If $a \in \mathbb{R}$, then $a^{+} = \max( 0, a )$ and $a^{-} = \min( 0, a )$.
 
% -------------------------------------------------
\subsection{Linear Algebra}

We assume familiarity with the basics of linear algebra.
Otherwise, we refer the reader to an introductory text, such as~\cite{Axler2014}.
Let $M$ be a complex matrix.
We let $M_{j,k}$ denote the entry of $M$ in the $j$-th row and the $k$-th column.
We recall the following definitions.
The \emph{conjugate of $M$} is the matrix $\overline{M}$ such that $\overline{M}_{j,k} = \overline{M_{j,k}}$. The \emph{transpose of $M$} is the matrix $M^T$ such that $(M^T)_{j,k} = M_{k,j}$. The \emph{adjoint of $M$} is the matrix $\overline{M}^T$, and is denoted $M^{\dagger}$.
A matrix $H$ is called \emph{Hermitian} if $H = H^{\dagger}$.
A matrix $U$ is called \emph{unitary} if $U$ is invertible and $U^{-1} = U^{\dagger}$.

% -------------------------------------------------
\subsection{Algebraic Numbers and Computation}
\label{Sect:Background:Algebraic}

We assume the reader is familiar with field theory, as found in standard abstract algebra textbooks, such as~\cite{DummitFoote2003}.
Let $\mathbb{F}$ be a subfield of $\mathbb{K}$.
An element $\alpha \in \mathbb{K}$ is \emph{algebraic over $\mathbb{F}$} if there exists a polynomial $p \in \mathbb{F}[x]$ such that $p( \alpha ) = 0$.
We write $\mathbb{F}( \alpha )$ to denote the smallest subfield of $\mathbb{K}$ containing both $\mathbb{F}$ and $\alpha$.
If $\deg( p ) = n$, then it can be shown that the elements of $\mathbb{F}( \alpha )$ form  a finite-dimensional $\mathbb{F}$-vector space with basis vectors $\{ 1, \alpha, \alpha^2, \ldots, \alpha^{n-1} \}$.
Furthermore, this vector space forms an $\mathbb{F}$-algebra under the multiplication of $\mathbb{F}( \alpha )$.
In the case where $\mathbb{F} = \mathbb{Q}$ and $\mathbb{K} = \mathbb{C}$, we say that $\alpha$ is an \emph{algebraic number}.
The field of all algebraic numbers is denoted $\mathbb{Q}^{\mathrm{Alg}}$.
Algebraic numbers are ideal from a computational perspective, since elements from $n$-dimensional $\mathbb{Q}$-vector spaces can be represented exactly using only $2n$ integers (i.e., the numerators and denominators).
This is in contrast to floating-point arithemtic, which is inherently inexact.

A special class of algebraic numbers are the \emph{cyclotomic numbers}.
These are solutions to polynomial equations of the form $x^n - 1 = 0$.
In other words, each cyclotomic number is a root of unity.
We let $\zeta_n$ denote the \emph{primitive $n$-th root of unity}, which can be defined analytically as $\zeta_n = e^{i2\pi/n}$.
For example, $\zeta_2 = -1$ and $\zeta_4 = i$.
The smallest subfield of $\mathbb{C}$ containing $\mathbb{Q}$ and all cyclotomic numbers is referred to as the \emph{universal cyclotomic field}.
Many algorithms exist to work efficiently with elements of the universal cyclotomic field, such as~\cite{Bosma1990} and~\cite{Breuer1997}.
It is well-known that many quantum gate sets can be defined exactly using only finite-dimensional sub-fields of the universal cyclotomic field, such as the Clifford+$T$ gate set~\cite{GilesSelinger2013} and its generalizations~\cite{AmyGlaudell2024}.
For this reason, recent work in the verification of quantum programs has advocated for the use of cyclotomic numbers as an exact representation~\cite{Avanzini2024}.

\begin{figure}[t]
    \begin{subfigure}[t]{0.48\textwidth}
        \centering
        \includegraphics[scale=0.75]{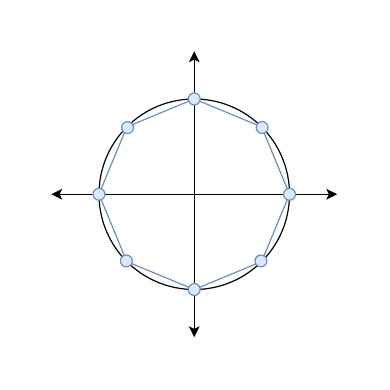}
        \vspace{-2em}
        \caption{The roots of unity in $\mathbb{Q}(\zeta_8)$.}
        \label{Fig:Zeta8}
    \end{subfigure}
    \begin{subfigure}[t]{0.48\textwidth}
        \centering
        \includegraphics[scale=0.75]{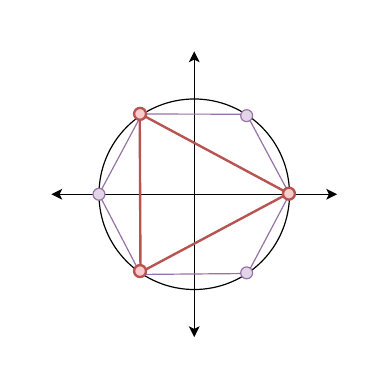}
        \vspace{-2em}
        \caption{$\mathbb{Q}(\zeta_6) = \mathbb{Q}(\zeta_3)$ since $\zeta_6 = -(\zeta_3)^2$.}
        \label{Fig:Zeta3Zeta6}
    \end{subfigure}
    \caption{Geometry of the cyclotomic numbers. The basis vectors of $\mathbb{Q}[\zeta_n]$ form the vertices of a regular $n$-gon on the complex unit circle, with one vertex at $(1, 0)$.}
    \vspace{-1em}
\end{figure}

In this paper, we also utilize analytic properties of cyclotomic numbers.
It follows from Euler's formula that $e^{i\theta} = \cos( \theta ) + i\sin( \theta )$.
We can then think of each cyclotomic number as a point of the complex unit circle (see \cref{Fig:Zeta8}).
It follows geometrically that $\mathbb{Q}(\zeta_n) = \mathbb{Q}(\zeta_{2n})$ whenever $n$ is odd (see \cref{Fig:Zeta3Zeta6}).
Moreover, it can be shown by simple algebraic manipulations that the following equations hold.
{\small\begin{align*}
    \cos( \theta ) &= \frac{e^{i\theta} + e^{-i\theta}}{2}
    &
    \sin( \theta ) &= \frac{e^{i\theta} - e^{-i\theta}}{2i}
\end{align*}\par}\noindent%
If $\theta$ is a rational multiple of $\pi$, say $(q / r) 2\pi$, this means that both $\cos( \theta )$ and $\sin( \theta )$ are elements of $\mathbb{Q}(i,\zeta_r)$.
However, identifying roots of unity can be challenging, since not all elements of norm $1$ in the universal cyclotomic field are roots of unity.
A well-known example is $(3+4i)/5$, which has norm $1$ but is not a root of unity.

% -------------------------------------------------
\subsection{Multivariate Laurent Polynomials}
\label{Sect:Background:Polynomials}

Let $R$ be a ring.
Then $R[ x_1, \ldots, x_k ]$ denotes the \emph{ring of multivariate polynomials} with coefficients in $R$ and indeterminates $x_1$ through $x_k$.
An arbitrary element $f \in R[ x_1, \ldots, x_k ]$ is of the form $f( x_1, \ldots, x_k ) = \sum_{t \in T} ( a_t \prod_{j=1}^{k} {x_j}^{t_j} )$ for some finite $T \subseteq \mathbb{N}^k \setminus \{ 0 \}^k$ with a non-zero sequence $\{ a_t \}_{t \in T}$ over $R$.
We write $\deg_{x_j}( f )$ for the degree of $f$ in variable $x_j$ and $\deg( f )$ for the total degree of $f$, where $\deg_{x_j}( f ) = \max \{ t_j : t \in T \}$ and $\deg( f ) = \max\{ \sum_{j=1}^{k} t_j : t \in T \}$.
When $R$ is an integral domain, the following hold for all $f, g \in R[ x_1, \ldots, x_k ]$ and $j \in [k]$.
{\small\begin{align*}
    \deg_{x_j}( fg ) &= \deg_{x_j}( f ) + \deg_{x_j}( g )
    &
    \deg( fg ) &= \deg( f ) + \deg( g )
    \\
    \deg_{x_j}( f + g ) &\le \max\{ \deg_{x_j}( f ), \deg_{x_j}( g ) \}
    &
    \deg( f + g ) &\le \max\{ \deg( f ), \deg( g ) \}
\end{align*}\par}\noindent%
It is well known that when $k = 1$ and $R$ is an integral domain, either $f = 0$ or $f$ has at most $\deg( f )$ zeros.
A consequence is that for any $S \subseteq R$, if $f \ne 0$ and $|S| > \deg_{x_1}( f )$, then there exists an $s \in S$ such $f( s ) \ne 0$.
Moreover, if $s$ is sampled uniformly from $S$, then $\Pr( f( x ) = 0 ) \le \deg( f ) / |S|$.
The latter two remarks generalize to multivariate polynomials.
Further generalization to Laurent polynomials are possible, by clearing the denominators.

\begin{theorem}[{Combinatorial Nullstellensatz~\cite{Alon1999}}]
    \label{Thm:Nullstellensatz}
    Let $\mathbb{F}$ be a field and $f$ a polynomial in $\mathbb{F}[ x_1, x_2, \ldots, x_k ]$ with total degree $d_1 + d_2 + \cdots + d_k$ such that the coefficient of $\prod_{j=1}^{k} x_j^{d_j}$ is nonzero in $f$.
    If $S_1, S_2, \ldots, S_k$ are subsets of $\mathbb{F}$ with $|S_j| > d_j$ for each $j$, then there exists $x \in S_1 \times S_2 \times \cdots \times S_k$ such that $f( x ) \ne 0$.
\end{theorem}

\begin{theorem}[{DeMillo–Lipton–Schwartz–Zippel Lemma~\cite{DemilloLipton1978,Zippel1979,Schwartz1980}}]
    \label{Thm:DLSZ}
    Let $R$ be an integral domain and $f \in R[ x_1, x_2, \ldots, x_k ]$ a polynomial with total degree $d$.
    For each finite subset $S$ of $R$, if $s_1, s_2, \ldots, s_k$ are sampled at random, both independently and uniformly from $S$, then $\Pr( f( s_1, s_2, \ldots, s_k ) = 0 ) \le d / |S|$.
\end{theorem}

We can further generalize multivariate polynomials to multivariate Laurent polynomials, denoted $R[ x_1, x_1^{-1}, \ldots, x_k, x_k^{-1} ]$.
In this setting, $T \subseteq \mathbb{Z}^k$, so that powers may be positive or negative.
For example, $f( x_1, x_2 ) = x_1x_2 - x_1^{-3} + 5$ is a Laurent polynomial from $\mathbb{Z}[ x_1, x_1^{-1}, x_2, x_2^{-1} ]$.
Since the exponents in a Laurent polynomial may be both positive and negative, each Laurent polynomial has both positive and negative degrees.
We write $\deg_{x_j}^{+}( f )$ for the positive degree of $f$ in variable $x_j$ and $\deg_{x_j}^{-}$ for the negative degree of $f$ in variable $x_j$, where $\deg_{x_j}^{+}( f ) = \max\{ t_j^{+} : t \in T \}$ and $\deg_{x_j}^{-}( f ) = \max\{ -t_j^{-} : t \in T \}$.
Similarly, the total positive degree of $f$ is $\deg^{+}( f ) = \max\{ \sum_{j=1}^{k} t_j^{+} : t \in T \}$.

\section{A Syntax and Semantics for Parameterized Circuits}
\label{Sect:Circuits}

This section begins by reviewing quantum states, quantum operators, and their composition, as in~\cite[Ch.~4]{NielsenChuang2011}.
This background material is then used to give syntax and parameterized semantics for quantum circuits with arbitrary gates, and rotations around arbitrary axes.

\subsection{Quantum States}

The primitive unit of information in quantum computing is the qubit.
As in classical computing, a qubit can be in the states zero and one, denoted $\ket{0}$ and $\ket{1}$.
However, a qubit may also be in a 
\emph{superposition} of the states $\ket{0}$ and $\ket{1}$.
Formally, this means that the state of a qubit $\ket{\psi}$ can be described as $\alpha \ket{0} + \beta \ket{1}$ for any $\alpha \in \mathbb{C}$ and $\beta \in \mathbb{C}$ satisfying $|\alpha|^2 + |\beta|^2 = 1$.
To simplify calculations, we think of $\ket{0}$ and $\ket{1}$ as the standard basis vectors for $\mathbb{C}^2$ to obtain the following vector equation: $\ket{\psi}
    =
    \alpha \ket{0} + \beta \ket{1}
    =
    \alpha \bigl[\begin{smallmatrix}
        1 \\ 0
    \end{smallmatrix}\bigr]
    +
    \beta \bigl[\begin{smallmatrix}
        0 \\ 1
    \end{smallmatrix}\bigr]
    =
    \bigl[\begin{smallmatrix}
        \alpha \\ \beta
    \end{smallmatrix}\bigr]
$.

Of course, the quantum algorithms described in the introduction of this paper require more than a single qubit of information.
Given an $n$-qubit quantum system, there are clearly $2^n$ possible basis states.
For example, when $n = 2$, these are $\ket{00}$, $\ket{01}$, $\ket{10}$, and $\ket{11}$.
As before, an $n$-qubit quantum system may also be in an arbitrary superposition of these basis states with the modulus-squared of the coefficients summing to $1$.
For example, an arbitrary $2$-qubit quantum system has state $\ket{\psi} = \alpha \ket{00} + \beta \ket{01} + \gamma \ket{10} + \rho \ket{11}$ for any $\alpha, \beta, \gamma, \rho \in \mathbb{C}$ satisfying $|\alpha|^2 + |\beta|^2 + |\gamma|^2 + |\rho|^2 = 1$.
This means that the states of an $n$-qubit quantum system correspond to the unit vectors in $\mathbb{C}^{2^n}$.

\subsection{Quantum Operations}

A quantum program evolves the state of a quantum system, after which all qubits are measured.
Given a quantum state $\ket{\psi} = \sum_{j=1}^{2^n} \alpha_j \ket{j}$, the probability of observing state $\ket{j}$ is $|\alpha_j|^2$.
Then the paradigm of quantum computing is to construct an $n$-qubit quantum system whose probability distribution assigns high probability to the correct output.

The evolution of a quantum system is described by a linear transformation of its state space.
Since the laws of physics are reversible, then this transformation must be invertible.
Moreover, the inverse of this transformation should be its conjugate transpose.
This means that operations on $n$-qubit systems correspond to unitary matrices.
Given an $n$-qubit state $\ket{\psi}$ and an $(2^n) \times (2^n)$ dimensional matrix $M$, the state obtained by applying $M$ to $\ket{\psi}$ is $M \ket{\psi}$.
For example, the following four matrices are unitary operations on a qubit.
{\small\begin{align*}
    I &= \begin{bmatrix}
        1 & 0 \\
        0 & 1
    \end{bmatrix}
    &
    X &= \begin{bmatrix}
        0 & 1 \\
        1 & 0
    \end{bmatrix}
    &
    Z &= \begin{bmatrix}
        1 & 0 \\
        0 & -1
    \end{bmatrix}
    &
    Y &= \begin{bmatrix}
        0 & -i \\
        i & 0
    \end{bmatrix}
\end{align*}\par}\noindent%
The matrix $I$ corresponds to a no-op and the matrix $X$ corresponds to a not gate.
The matrix $Z$ can be understood as adjusting the coefficient of $\ket{1}$ by a factor of $(-1)$.
This has no classical analogue.
The gate $Y$ is equal to $(-iZ)X$, and therefore, corresponds to a not gate followed by some non-classical operation.

An important construct in classical computing is the if-then statement.
This can be generalized to quantum computing as follows.
Let $M$ be a unitary transformation on an $n$-qubit quantum system.
Then there exists a unitary transformation $I_{2^n} \oplus M$ on an $(n+1)$-qubit quantum system, such that $I_{2^n} \oplus M$ applies $M$ to the last $n$ qubits of a basis state if and only if the first qubit of the basis is in state $\ket{1}$.
Formally, $I_{2^n}$ is the $(2^n) \times (2^n)$ identity matrix, and $I_{2^n} \oplus M$ is the direct sum of $I_{2^n}$ with $M$.
In terms of matrices, $I_{2^n} \oplus M$ is simply the block diagonal matrix with blocks $I_{2^n}$ and $M$, as shown below.
{\small\begin{align*}
    I_{2^n} \oplus M
    &=
    \begin{bmatrix}
        I_{2^n} & 0 \\
        0 & M
    \end{bmatrix}
    &
    I_2 \oplus X
    &=
    \begin{bmatrix}
        I_2 & 0 \\
        0 & X
    \end{bmatrix}
    =
    \begin{bmatrix}
        1 & 0 & 0 & 0 \\
        0 & 1 & 0 & 0 \\
        0 & 0 & 0 & 1 \\
        0 & 0 & 1 & 0
    \end{bmatrix}
\end{align*}\par}\noindent%
The matrix for $I_2 \oplus X$, known as a \emph{cnot gate}, is given above.
This generalizes the classical conditional statement: if the first bit is in state $\ket{1}$, then apply a not gate to the second bit.

So far, all of the operations discussed are parameter-free.
However, quantum algorithms also make use of rotation gates, which are parameterized by an angle of rotation.
As the name suggests, a rotation gate is defined by its axis-of-rotation.
Formally, each axis $M$ is a Hermitian unitary matrix.
Then one can define, as a generalization of Euler's formula, the rotation $R_M( \theta )$ as follows.
{\small\begin{equation*}
    R_M( \theta ) = e^{-iM\theta/2} = \sum_{n=0}^{\infty} \frac{(-iM\theta/2)^n}{n!} = \cos( -\theta/2 ) I + i \sin( -\theta/2 ) M
\end{equation*}\par}\noindent%
This definition can be extended to $k$ parameters by taking any transformation $f: \mathbb{R}^k \to \mathbb{R}$.
For example, given $f( \theta_1, \theta_2 ) = \theta_1 + \theta_2$, we can define a two parameter rotation $R_M( f )$ where $R_M( f )( \theta_1, \theta_2 ) = R_M( f( \theta_1, \theta_2 ) ) = R_M( \theta_1 + \theta_2 )$.
In this work, we consider the family $\mathcal{F}$ of $k$-variable rational-linear functions with affine translations by rational multiples of $\pi$.
That is, the set $\mathcal{F}$ is defined to be $\{ f( \theta ) = a_1 \theta_1 + a_2 \theta_2 + \cdots + a_k \theta_k + q \pi \mid a_1, a_2, \ldots, a_k, q \in \mathbb{Q} \}$.

The most common rotations in quantum circuits are the $I$-, $X$-, $Y$-, and $Z$-rotations.
However, there are many single qubit rotations not of this form.
For example, given any coefficients $\alpha, \beta, \gamma \in \mathbb{R}$, if $\alpha^2 + \beta^2 + \gamma^2 = 1$, the matrix $\alpha X + \beta Y + \gamma Z$ is also a Hermitian unitary matrix.
Note that the matrix $R_I( -2\theta )$ is typically referred to as a \emph{global phase gate}, rather than an $I$-rotation.

\begin{example}[Real Amplitude Ansatz Circuit]
    In quantum machine learning, convolutional layers are often implemented using the real amplitude ansatz circuit~\cite{AbbasSutter2021,ArthurDate2022,DekelFrankel2023,KannoNakamura2024,YoffeEntin2024}.
    This circuit is composed from one or more  layers of $Z$-rotations, each followed by a layer of controlled-not gates.
    Since $Z$-rotations do not commute with the targets of controlled-not gates, then these layers can interact in non-trivial ways.
    The choice of parameter to each $Z$-rotation is treated as a weight in the quantum machine learning model.
\end{example}

\subsection{Composing Quantum Operations}

Just like classical operations, quantum operations can also be composed in sequence and in parallel.
Of the two, sequential composition is the simplest to describe.
Assume that both $M$ and $N$ are operations on an $n$-qubit quantum system.
If $N$ is applied to an $n$-qubit system $\ket{\psi}$, then the state $N \ket{\psi}$ is obtained.
If $M$ is then applied to this intermediate state, then the state $M (N \ket{\psi})$ is obtained.
This is equivalent to applying $MN$ to $\ket{\psi}$.
In other words, the sequential composition of quantum operations corresponds to matrix multiplication.

Now let $M$ denote a quantum operation on an $m$-qubit quantum system and $N$ denote a quantum operation on an $n$-qubit quantum system.
Intuitively, the parallel composition of $M$ and $N$ should act on the first $m$-qubits by $M$, and the last $n$-qubits by $N$.
However, this composition must also respect superposition, through a property known an \emph{bilinearity}.
To compute this new operation, the \emph{Kronecker tensor product} is required, which is denoted $\otimes$ and defined as follows for matrices of any dimension.
{\small\begin{equation*}
    \begin{bmatrix}
        c_{1,1} & c_{1,2} & \cdots & c_{1,n} \\
        c_{2,1} & c_{2,2} & \cdots & c_{2,n} \\
        \vdots & \vdots & \ddots & \vdots \\
        c_{m,1} & c_{m,2} & \cdots & c_{m,n}
    \end{bmatrix}
    \otimes
    M
    =
    \begin{bmatrix}
        c_{1,1} M & c_{1,2} M & \cdots & c_{1,n} M \\
        c_{2,1} M & c_{2,2} M & \cdots & c_{2,n} M \\
        \vdots & \vdots & \ddots & \vdots \\
        c_{m,1} M & c_{m,2} M & \cdots & c_{m,n} M
    \end{bmatrix}
\end{equation*}\par}\noindent%
It follows that $(M \otimes N)(\ket{\psi} \otimes \ket{\varphi}) = (M\ket{\phi}) \otimes (N\ket{\varphi})$ as desired.

\subsection{Quantum Circuits}

Quantum circuits are constructed from primitive gates, under sequential and parallel composition.
In this section, we first define what we take to be primitive gates, and then define what it means to be a circuit over this gate set.
The distinction between syntax and semantics is emphasized.
In both cases, we introduce inductive principles which will be used later in this paper.
Formally, these circuits correspond to diagrams in a certain PROP category~\cite{BaezCoya2010}, with semantics given functorially~\cite{Lawvere1963}, though this is only used to prove the inductive principles used throughout the paper, and to establish that our semantics and circuit transformations are well-defined (see~\cref{Appendix:Syntax} for more details).

In what follows, $C( - )$ is a function symbol used to denote conditional control.
A gate set is a collection of basic gates, closed under conditional control.
A basic gate is a complex matrix (e.g. unitary operations, state preparation, post-selection) or  parameterized rotation. 
Formally, we take some set $\mathcal{G}$ of complex matrices and some set $\mathcal{H}$ of Hermitian unitary matrices.
The associated gate set, denoted $\Sigma( \mathcal{G}, \mathcal{H} )$ is defined inductively as follows.
\begin{itemize}
    \item If $G \in \mathcal{G}$, then $G \in \Sigma( \mathcal{G}, \mathcal{H} )$.
    \item If $M \in \mathcal{H}$, then $R_M( f ) \in \Sigma( \mathcal{G}, \mathcal{H} )$ for each parameterization $f \in \mathcal{F}$.
    \item If $G \in \Sigma( \mathcal{G}, \mathcal{H} )$ and $G$ is unitary, then $C( G ) \in \Sigma( \mathcal{G}, \mathcal{H} )$.
\end{itemize}
We let $\win( - )$ and $\wout( - )$ denote the input and output arities of these gates, which are defined as follows.
\begin{itemize}
    \item If $G \in \mathcal{G}$ is $(2^n) \times (2^m)$, then $\win( G ) = n$ and $\wout( G ) = m$.
    \item If $M \in \mathcal{H}$ is $(2^n) \times (2^n)$ and $f \in \mathcal{F}$, then $\win( R_M( f ) ) = \wout( R_M( f ) ) = n$.
    \item If $G \in \Sigma( \mathcal{G}, \mathcal{H} )$, then $\win( C( G ) ) = \win( G ) + 1$ and $\wout( C( G ) ) = \wout( G ) + 1$.
\end{itemize}
We let $\interp{-}$ denote the parameterized semantics of each gate, which are defined as expected.
\begin{itemize}
    \item If $G \in \mathcal{G}$, then $\interp{G}( \theta ) = G$.
    \item If $M \in \mathcal{H}$ and $f \in \mathcal{F}$, then $\interp{R_M(f)}( \theta ) = \cos( -f(\theta) / 2 ) I + i \sin( -f(\theta) / 2 ) M$.
    \item If $G \in \Sigma( \mathcal{G}, \mathcal{H} )$ with $G$ an $(2^n) \times (2^n)$ unitary, then $\interp{C(G)}(\theta) = I_{2^n} \oplus \interp{G}(\theta)$.
\end{itemize}
Since this gate set is defined inductively, then to prove that every gate satisfies a predicate $P$, it suffices to use well-founded induction~(see~\cref{Appendix:Syntax}).

\begin{restatable}{proposition}{gateprop}
    \label{Prop:GateInd}
    Assume that a predicate $P$ on $\Sigma( \mathcal{G}, \mathcal{H} )$ satisfies the following.
    \begin{itemize}
    \item \textbf{Base Case (1).}
          $\forall G \in \mathcal{G}, P( G )$.
    \item \textbf{Base Case (2).}
          $\forall M \in \mathcal{H}, \forall f \in \mathcal{F}, P( R_M( f ) )$.
    \item \textbf{Control Induction.}
          $\forall G \in \Sigma( \mathcal{G}, \mathcal{H} )$, $G$ unitary and $P( G )$ implies $P( C( G ) )$.
    \end{itemize}
    Then $P( G )$ holds for each $G \in \Sigma( \mathcal{G}, \mathcal{H} )$.
\end{restatable}

\begin{figure}[t]
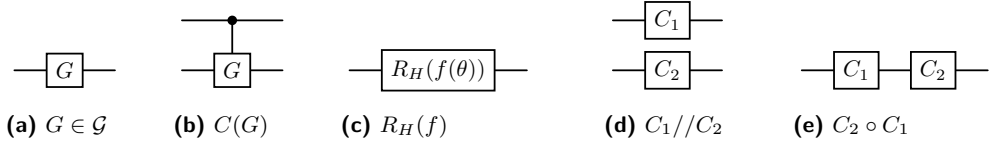

    \begin{subfigure}[b]{0.15\textwidth}
        \input{circuits/gate}
        \caption{$G \in \mathcal{G}$}
    \end{subfigure}
    \begin{subfigure}[b]{0.15\textwidth}
        \input{circuits/ctrl}
        \caption{$C( G )$}
    \end{subfigure}
    \begin{subfigure}[b]{0.24\textwidth}
        \input{circuits/rot}
        \caption{$R_H( f )$}
    \end{subfigure}
    \begin{subfigure}[b]{0.17\textwidth}
        \input{circuits/parallel}
        \caption{$C_1 // C_2$}
    \end{subfigure}
    \begin{subfigure}[b]{0.26\textwidth}
        \input{circuits/sequential}
        \caption{$C_2 \circ C_1$}
    \end{subfigure}
    \caption{The graphical language for circuits in $\Circ( \mathcal{G}, \mathcal{H} )$.}
    \label{Fig:CircLang}
    \vspace{-1em}
\end{figure}

Circuits are then constructed from the elements of $\Sigma( \mathcal{G}, \mathcal{H} )$ through sequential and parallel composition.
We let  $( \circ )$ denote sequential composition and $( // )$ denote parallel composition, to distinguish between syntactic compositions and their semantic counterparts.
Of course, sequential composition requires that the outputs of the first sub-circuit matches the inputs of the second sub-circuit.
To handle this, we extend $\win( - )$ and $\wout( - )$ as follows.
\begin{itemize}
    \item $\win( C_1 // C_2 ) = \win( C_1 ) + \win( C_2 )$ and $\wout( C_1 // C_2 ) = \wout( C_1 ) + \wout( C_2 )$.
    \item $\win( C_2 \circ C_1 ) = \win( C_1 )$ and $\wout( C_2 \circ C_1 ) = \wout( C_2 )$.
\end{itemize}
Then $\Circ( \mathcal{G}, \mathcal{H} )$, the family of circuits over the gate set $\Sigma( \mathcal{G}, \mathcal{H} )$, is defined inductively as follows where $\epsilon$ denotes the \emph{empty} wire with $\win( \epsilon ) = \wout( \epsilon ) = 1$.
\begin{itemize}
    \item If $C \in \Sigma( \mathcal{G}, \mathcal{H} )$, then $C \in \Circ( \mathcal{G}, \mathcal{H} )$.
    \item If $C_1, C_2 \in \Circ( \mathcal{G}, \mathcal{H} )$, then $C_1 // C_2 \in \Circ( \mathcal{G}, \mathcal{H} )$.
    \item If $C_1, C_2 \in \Circ( \mathcal{G}, \mathcal{H} )$ and $\win( C_2 ) = \wout( C_1 )$, then $C_2 \circ C_1 \in \Circ( \mathcal{G}, \mathcal{H} )$.
\end{itemize}
A graphical language for $\Circ( \mathcal{G}, \mathcal{H} )$ is given in \cref{Fig:CircLang}.
The semantic map $\interp{-}$ extends to these circuits as expected: $\interp{C_2 // C_1}( \theta ) = \interp{C_2}( \theta ) \otimes \interp{C_1}( \theta )$, $\interp{C_2 \circ C_1}( \theta ) = ( \interp{C_2}( \theta ) ) ( \interp{C_1}( \theta ) )$, and $\interp{\epsilon} = I_2$.
As with quantum gates, an inductive principle also holds for quantum circuits.

\begin{restatable}{proposition}{circprop}
    \label{Prop:CircInd}
    Assume that a predicate $P$ on $\Circ( \mathcal{G}, \mathcal{H} )$ satisfies the following.
    \begin{itemize}
    \item \textbf{Base Case (1).}
          $P( \epsilon )$. 
    \item \textbf{Base Case (2).}
          $\forall G \in \Sigma( \mathcal{G}, \mathcal{H} ), P( G )$.
    \item \textbf{Parallel Induction}.
          If $C_1, C_2 \in \Circ( \mathcal{G}, \mathcal{H} )$ such that $P( C_1 )$ and $P( C_2 )$, then $P( C_1 // C_2 )$.
    \item \textbf{Sequential Induction.}
          If $C_1, C_2 \in \Circ( \mathcal{G}, \mathcal{H} )$ such that $\win( C_2 ) = \wout( C_1 )$ with $P( C_1 )$ and $P( C_2 )$, then $P( C_2 \circ C_1 )$.
    \end{itemize}
    Then $P( C )$ holds for each $C \in \Circ( \mathcal{G}, \mathcal{H} )$.
\end{restatable}

\section{A Motivating Example: Circuit Compilation}
\label{Sect:Motivation}

\begin{figure}[t]
    \begin{align*}
        \input{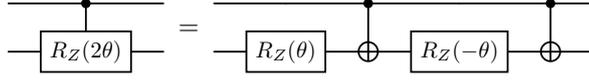}
    \end{align*}
    \vspace{-1.5em}
    \caption{A parameterized equality used to compile controlled rotations.}
    \label{Fig:MotivatingEx}
    \vspace{-1em}
\end{figure}

We now discuss the verification of a concrete circuit equation. The example is simple but illustrative of the techniques we will develop in the next section. Consider the equation in \cref{Fig:MotivatingEx}. A naive approach to establishing this equation is to evaluate the right-hand side to obtain the following operator.
{\small\begin{equation*}
    (I \oplus X)
    (I \otimes R_Z( -\theta ))
    (I \oplus X)
    (I \otimes R_Z( \theta ))
    =
    \begin{bmatrix}
        R_Z( -\theta ) R_Z( \theta ) & 0 \\
        0 & X R_Z( -\theta ) X R_Z\left( \theta \right)
    \end{bmatrix}
\end{equation*}\par}\noindent%
Then, by further simplification, we obtain the following equations.
{\small\begin{align*}
    R_Z(-\theta) R_Z(\theta) &=
    \begin{bmatrix}
        e^{-i\theta/2} e^{i\theta/2} & 0 \\
        0 & e^{i\theta/2} e^{-i\theta/2}
    \end{bmatrix}
    &
    X R_Z( -\theta ) X R_Z( \theta ) &=
    \begin{bmatrix}
        e^{-i\theta/2} e^{-i\theta/2} & 0 \\
        0 & e^{i\theta/2} e^{i\theta/2}
    \end{bmatrix}
\end{align*}\par}\noindent%
Using the identities $e^{a}e^{b} = e^{a+b}$ and $e^0=1$, it then follows that $X R_Z( -\theta ) X R_Z( \theta ) = R_Z( 2\theta )$ and $R_Z(-\theta) R_Z(\theta) = I$.
Consequently,
{\small\begin{equation*}
    (I \oplus X)
    (I \otimes R_Z( -\theta ))
    (I \oplus X)
    (I \otimes R_Z( \theta ))
    =
    \begin{bmatrix}
        I & 0 \\
        0 & R_Z( 2\theta )
    \end{bmatrix}
    =
    (I \oplus R_Z( 2\theta )).
\end{equation*}\par}\noindent%
This establishes the equation in \cref{Fig:MotivatingEx} for all choices of $\theta$.
However, this proof depends on the parameterized equations $e^{a+b} = e^{a} e^{b}$ and $e^0=1$. In general, it is challenging to find a complete set of parameterized relations for a parameterized gate set~\cite{MillerBakewellThesis}. Moreover, given an arbitrary set of complete relations, the problem of deciding if two expressions are equivalent is known to be undecidable~\cite{Novikov1955}. For these reasons, we adopt a different approach.

A perhaps surprising result is that all parameterized circuit equalities can be established by checking only a finite number of rotation angles.
In other words, if the equality in \cref{Fig:MotivatingEx} did not hold, then a counterexample could be found by checking only a fixed number of instances.
To do this, we first convert the equality into a family of polynomials, such that the equality holds if and only if all of the polynomials are identically zero.
We then find an integer $n$ such that each of the polynomials has degree at most $n$.
Since non-zero polynomials of degree $n$ have at most $n$ roots, then either the polynomial is zero and will evaluate to zero on $n + 1$ angles, or the polynomial is non-zero and at least one of the $n + 1$ angles yields a non-zero result.

To obtain the desired polynomials, we apply the change-of-variable $e^{-i\theta/2} \mapsto z$. 
Under this change of variable, the following equalities hold.
{\small\begin{align*}
    R_Z(-\theta) R_Z(\theta) &=
    \begin{bmatrix}
        z^{-1} z & 0 \\
        0 & z z^{-1}
    \end{bmatrix} =
    \begin{bmatrix}
        1 & 0 \\
        0 & 1
    \end{bmatrix} =
    z^{-2}
    \begin{bmatrix}
        z^2 & 0 \\
        0 & z^2
    \end{bmatrix}
    \\
    X R_Z( -\theta ) X R_Z( \theta ) &=
    \begin{bmatrix}
        z^{-1} z^{-1} & 0 \\
        0 & z z
    \end{bmatrix} =
    \begin{bmatrix}
        z^{-2} & 0 \\
        0 & z^2
    \end{bmatrix} =
    z^{-2}
    \begin{bmatrix}
        1 & 0 \\
        0 & z^4
    \end{bmatrix}
\end{align*}\par}\noindent%
Continuing in this fashion, we can find that each matrix entry on the left-hand side or the right-hand side of \cref{Fig:MotivatingEx} has degree at most four.
Then the difference between the left-hand side and the right-hand side also has degree at most four.
Note that the $z^{-2}$ terms correspond to a removable singularity at $z = 0$, which does not fall on the complex unit circle, and can be safely ignored.
Since degree four polynomials have at most four roots, then it suffices to check the equality in \cref{Fig:MotivatingEx} using only $5$ angles from $[0, 4\pi)$.
For example, consider the five angles $\theta_j = j\pi/2$ for $0\leq j\leq 4$. It is easily verified that $( I \oplus R_Z( 2\theta_j ) ) = (I \oplus X) (I \otimes R_Z( -\theta_j )) (I \oplus X) (I \otimes R_Z( \theta_j ))$ for all $0\leq j \leq 4$.
Then $f( \theta ) = (I \oplus R_Z( 2\theta )) - ( I \oplus X ) (I \otimes R_Z( -\theta )) (I \oplus X) (I \otimes R_Z( \theta ))$ has at least five roots.
Since each entry of $f( \theta )$ has degree at most four, then $f$ is identically zero and \cref{Fig:MotivatingEx} must hold.
Note that the angles were sampled from $[0,4\pi)$ since $e^{-i\theta_j/2}$ has period $4\pi$.

While this example was admittedly simplistic, we will see in the next section, that the technique generalizes to all parameterized circuits.
In particular, just as in this example, we will see that computing the polynomials is inconsequential.
Instead, it will suffice to find an efficient procedure which provides a reason bound on each degree.

\section{Equivalence Checking Techniques}
\label{Sect:Equiv}

In this section, we consider parameterized quantum circuits where all coefficients are from $\mathbb{Z}$, rather than $\mathbb{Q}$.
We denote these circuits $\ZCirc( \mathcal{G}, \mathcal{H} )$.
It is first shown that up to a change of variable, these circuits admit semantics as matrices over the ring of Laurent polynomials $\mathbb{C}[ z_1, z_1^{-1}, \ldots, z_k, z_k^{-1} ]$.
This is then combined with \cref{Thm:Nullstellensatz} to establish a cutoff-based equivalence checking theorem for these circuits.
Using \cref{Thm:DLSZ}, a probabilistic variant is also obtained.
In \cref{Sect:Implementation}, we show how these results generalize back to parameterized circuits with rational coefficients.

% -------------------------------------------------
\subsection{Polynomial Semantics}
\label{Sect:Equiv:Poly}

This section shows that, up to a change of variable, each circuit $\Circ( \mathcal{G}, \mathcal{H} )$ has semantics given by a matrix with entries corresponding to complex Laurent polynomials.
Moreover, these polynomials are shown to have degrees bounded by certain sums of the coefficients which appear in the circuit.
It follows that the techniques used in \cref{Sect:Motivation} can be generalized to all integral circuits in $\ZCirc( \mathcal{G}, \mathcal{H} )$.

As a first step, a new semantic interpretation $\pinterp{-}$ is provided for $\ZCirc( \mathcal{G}, \mathcal{H} )$, which interprets each circuit in $\ZCirc( \mathcal{G}, \mathcal{H} )$ as a polynomial over $\mathbb{C}[ z_1, z_1^{-1}, \ldots, z_k, z_k^{-1} ]$.
Since parameters only appear in trigonometric terms, then a first step is to give Laurent polynomials which abstract the trigonometric terms.
Let $\alpha \in \mathbb{Z}^k$, $q \in \mathbb{Q}$, and $f( \theta ) = \alpha_1 \theta_1 + \cdots \alpha_k \theta_k + q$.
{\small\begin{align*}
    \cos\left( -\frac{f( \theta )}{2} \right)
    &=
    \frac{e^{i(-f( \theta )/2)} + e^{-i(-f( \theta )/2)}}{2}
    =
    \frac{e^{-iq/2}}{2} \prod_{j=1}^{k} \left( e^{-i\theta_j/2} \right)^{a_j}
    +
    \frac{e^{iq/2}}{2} \prod_{j=1}^{k} \left( e^{i\theta_j/2} \right)^{a_j}
    \\
    \sin\left( - \frac{f( \theta )}{2} \right)
    &=
    \frac{e^{i(-f( \theta )/2)} - e^{-i(-f( \theta )/2)}}{2i}
    =
    \frac{e^{-iq/2}}{2i} \prod_{j=1}^{k} \left( e^{-i\theta_j/2} \right)^{a_j}
    -
    \frac{e^{iq/2}}{2i} \prod_{j=1}^{k} \left( e^{i\theta_j/2} \right)^{a_j}
\end{align*}\par}\noindent%
By substituting $z_j = e^{-i \theta_j / 2}$ for each $j \in [k]$ and letting $c = e^{-iq/2}$, the following Laurent polynomials are obtained.
{\small\begin{align*}
    \CPoly( f ) &= \frac{c}{2} \prod_{j=1}^{k} z_j^{\alpha_j} + \frac{1}{2c} \prod_{j=1}^{k} z_j^{-\alpha_j}
    &
    \SPoly( f )
    &= \frac{-ic}{2} \prod_{j=1}^{k} z_j^{\alpha_j} + \frac{i}{2c} \prod_{j=1}^{k} z_j^{-\alpha_j}
\end{align*}\par}\noindent%
Then the following equations hold by construction.
{\small\begin{align*}
    \CPoly( f )\left( e^{-i\theta_1/2}, \ldots, e^{-i\theta_k/2} \right)
    &=
    \frac{e^{-iq/2}}{2} \prod_{j=1}^{k} \left( e^{-i\theta_j/2} \right)^{a_j}
    +
    \frac{e^{iq/2}}{2} \prod_{j=1}^{k} \left( e^{i\theta_j/2} \right)^{a_j}
    =
    \cos\left( -\frac{f( \theta )}{2} \right)
    \\
    \SPoly( f )\left( e^{-i\theta_1/2}, \ldots, e^{-i\theta_k/2} \right)
    &=
    \frac{e^{-iq/2}}{2i} \prod_{j=1}^{k} \left( e^{-i\theta_j/2} \right)^{a_j}
    -
    \frac{e^{iq/2}}{2i} \prod_{j=1}^{k} \left( e^{i\theta_j/2} \right)^{a_j}
    =
    \sin\left( - \frac{f( \theta )}{2} \right)
\end{align*}\par}\noindent%
Given these polynomials, $\pinterp{-}$ is defined inductively on the gates as follows.
\begin{itemize}
    \item If $G \in \mathcal{G}$, then $\pinterp{G} = G$.
    \item If $M \in \mathcal{H}$ and $f \in \mathcal{F}$, then $\pinterp{R_M(f)} = \CPoly( f ) I + i \, \SPoly( f ) M$.
    \item If $G \in \Sigma( \mathcal{G}, \mathcal{H} )$ with $G$ an $(2^n) \times (2^n)$ unitary, then $\pinterp{C(G)} = I_{2^n} \oplus \pinterp{G}$.
\end{itemize}
The semantics extend as expected to sequential and parallel composition.
This makes precise the change of variable used in~\cref{Sect:Motivation}.

\begin{definition}[Polynomial Abstraction]
    A \emph{polynomial abstraction} is a function $\interp{-}_{*}$ from $\ZCirc( \mathcal{G}, \mathcal{H} )$ to collection of matrices over $\mathbb{C}[ z_1, z_1^{-1}, \ldots, z_k, z_k^{-1} ]$ such that $\interp{C}( \theta_1, \ldots, \theta_k ) = \interp{C}_{*}\left( e^{-i\theta_1/2}, \ldots, e^{-i\theta_k/2} \right)$ for all $C \in \ZCirc( \mathcal{G}, \mathcal{H} )$.
\end{definition}

\begin{restatable}{theorem}{polyabs}
    \label{Thm:PolyAbstract}
    $\pinterp{-}$ is a polynomial abstraction.
\end{restatable}

\begin{example}[Polynomial Semantics]
    \label{Ex:PolySem}
    The calculations from \cref{Sect:Motivation} can be revisited from the perspective of polynomial semantics.
    Of course, the circuit in \cref{Fig:MotivatingEx} is somewhat uninteresting, since the circuit has only one parameter.
    Instead, we will consider a new circuit with two parameters $\rho_1$ and $\rho_2$ obtained through the substitution $\theta = f( \rho_1, \rho_2 )$ where $f( \rho_1, \rho_2 ) = \rho_1 - 2 \rho_2$.
    The sine and cosine polynomials for $f$ are as follows.
    \begin{align*}
        \CPoly( f )
        &=
        \frac{1}{2} z_1 z_2^{-2} + \frac{1}{2} z_1^{-1} z_2^{2}
        &
        \SPoly( f )
        &=
        \frac{-i}{2} z_1 z_2^{-2} + \frac{i}{2} z_1^{-1} z_2^{2}
    \end{align*}
    Then $\CPoly(f) + i \; \SPoly( f ) = z_1 z_2^{-2}$ and $\CPoly(f) - i \; \SPoly( f ) = z_1^{-1} z_2^{2}$.
    Let $C_1$ denote the right-hand side of the equation in~\cref{Fig:MotivatingEx}.
    To compute $\pinterp{C_1}$, we start by evaluating each gate.
    Clearly $\pinterp{ C( X ) } = I_2 \oplus X$.
    Moreover,
    {\small\begin{align*}
        \pinterp{\epsilon // R_Z( f )}
        =
        I_2 \otimes \pinterp{R_Z( f )}
        &=
        I_2 \otimes \begin{bmatrix}
            z_1z_2^{-2} & 0 \\
            0 & z_1^{-1}z_2^2
        \end{bmatrix},
        \\
        \pinterp{\epsilon // R_Z( -f )}
        =
        I_2 \otimes \pinterp{R_Z( -f )}
        &=
        I_2 \otimes \begin{bmatrix}
            z_1^{-1}z_2^{2} & 0 \\
            0 & z_1z_2^{-2}
        \end{bmatrix}.
    \end{align*}\par}\noindent%
    It follows by calculations similar to those in \cref{Sect:Motivation} that,
    {\small\begin{equation*}
        \pinterp{C_1}
        =
        \pinterp{ C( X ) }
        \pinterp{\epsilon // R_Z( -f )}
        \pinterp{ C( X ) }
        \pinterp{\epsilon // R_Z( f )}
        =
        I_2 \oplus
        \begin{bmatrix}
            z_1^{-2}z_2^{4} & 0 \\
            0 & z_1^{2}z_2^{-4}
        \end{bmatrix}.
    \end{equation*}\par}\noindent%
    Then $\pinterp{C_1}( e^{-i\rho_1/2}, e^{-i\rho_2/2} ) = I \oplus R_Z( 2f( \rho_1, \rho_2 ) ) = \interp{C_1}( \rho_1, \rho_2 )$ as expected.
    \qed
\end{example}

To check that $\interp{C_1} = \interp{C_2}$, it suffices to check symbolically that $\pinterp{C_1} = \pinterp{C_2}$.
However, it is often too computationally expensive to compute the polynomials explicitly.
Instead, one could first upper-bound the degree of each polynomial, and then combine these degree bounds with the theorems of~\cref{Sect:Background:Polynomials}.
It is not hard to see that for each component of $\pinterp{R_H( f )}$, its degrees are all bounded by the coefficients of $f$.
This property extends to all circuits in $\ZCirc( \mathcal{G}, \mathcal{H} )$ by studying their coefficient sequences.
Intuitively, the coefficient sequence of a circuit $C$ is a sequence $A( C )$ over $\mathbb{Q}^k$ such that $A( C )_j$ is the list of coefficients for the $j$-th rotation in $C$.
More formally, let $( \mathbb{Q}^k )^*$ denote the set of all finite sequences over $\mathbb{Q}^k$ and $( \cdot )$ denote sequence concatenation.
Then $A( - )$ is defined inductively as follows.
\begin{itemize}
\item If $G \in \mathcal{G}$, then $A( G ) = \epsilon$.
\item If $M \in \mathcal{H}$ and $f( \theta ) = a_1 \theta_1 + \cdots + a_k \theta_k + q$, then $A( R_M( f ) ) = ( ( a_1, \ldots, a_k ) )$.
\item If $G \in \Sigma( \mathcal{G}, \mathcal{H} )$, then $A( C( G ) ) = A( G )$.
\item If $C_1, C_2 \in \Circ( \mathcal{G}, \mathcal{H} )$, then $A( C_1 // C_2 ) = A( C_1 ) \cdot A( C_2 )$.
\item If $C_1, C_2 \in \Circ( \mathcal{G}, \mathcal{H} )$ and $\win( C_2 ) = \wout( C_1 )$, then $A( C_2 \circ C_1 ) = A( C_2 ) \cdot A( C_1 )$.
\end{itemize}
Then $\ZCirc( \mathcal{G}, \mathcal{H} )$ is precisely the set of circuits in $\ZCirc( \mathcal{G}, \mathcal{H} )$ such that $A( C ) \in ( \mathbb{Z}^k )^*$.
We define $\Sigma_\mathbb{Z}( \mathcal{G}, \mathcal{H} )$ analogously.
The following definition generalizes the coefficient bound of the degree of a gate to a coefficient bound on the degree of all circuits.

\begin{definition}[Coefficient Bounded Semantics]
    Let $\interp{-}_{*}$ be a polynomial abstraction.
    A circuit $C \in \Circ( \mathcal{G}, \mathcal{H} )$ with $\win( C ) = n$ and $\wout( C ) = m$ is \emph{coefficient bounded with respect to $\interp{-}_{*}$}, denoted $\Bnd_{*}( C )$, if for each $s \in [2^n]$ and $t \in [2^m]$ with $f = ( \interp{C}_{*} )_{s,t}$,
    \begin{itemize}
        \item (B1). $\deg^{+}_{z_j}( f ) \le \sum_{a \in A( C )} |a_j|$ for each $j \in [k]$,
        \item (B2). $\deg^{-}_{z_j}( f ) \le \sum_{a \in A( C )} |a_j|$ for each $j \in [k]$,
        \item (B3). $\deg^{+}( f ) \le \sum_{a \in A( C )} \kappa( a )$ where $\kappa( a ) = \max\{ \sum_{j=1}^{k} a_j^{+}, \sum_{j=1}^{k} -a_j^{-} \}$.
    \end{itemize}
\end{definition}

\begin{example}[Coefficient Bounded Semantics]
    \label{Ex:PolyBnd}
    Recall $C_1$ from \cref{Ex:PolySem}.
    It will be shown that $\Bnd_{\textsf{Poly}}( C_1 )$ holds.
    First, the coefficient sequence of $C_1$ must be computed.
    As illustrated in the previous example, $C_1$ contains only the rotations: $R_1 = C( R_Z( -\rho_1 + 2 \rho_2 ) )$ and $R_2 = C( R_Z( \rho_1 - 2 \rho_2 ) )$.
    The coefficient sequences of these rotations are $\beta = ( -1, 2 )$ and $\gamma = ( 1, -2 )$ respectively.
    Then $A( C_2 ) = A( R_1 ) \cdot A( R_2 ) = ( \beta ) \cdot ( \gamma ) = ( \beta, \gamma )$.
    Moreover, $\kappa( \beta ) = \max\{ 0 + 2, 1 + 0 \} = 2$ and $\kappa( \gamma ) = \max\{ 1 + 0, 0 + 2 \} = 2$.
    By inspecting the matrices in \cref{Ex:PolySem}, it is clear that the following bounds hold for all $j \in [2]$ and $s, t \in [4]$.
    {\small\begin{align*}
        \deg_{z_j}^{+}( ( \pinterp{R_1} )_{s,t} ) &\le |\beta_j|
        &
        \deg_{z_j}^{-}( ( \pinterp{R_1} )_{s,t} ) &\le |\beta_j|
        &
        \deg^{+}( ( \pinterp{R_1} )_{s,t} ) &\le \kappa( \beta )
        \\
        \deg_{z_j}^{+}( ( \pinterp{R_2} )_{s,t} ) &\le |\gamma_j|
        &
        \deg_{z_j}^{-}( ( \pinterp{R_2} )_{s,t} ) &\le |\gamma_j|
        &
        \deg^{+}( ( \pinterp{R_2} )_{s,t} ) &\le \kappa( \gamma )
    \end{align*}\par}\noindent%
    The $\kappa$ terms can be thought of as adding together the maximum positive degrees of the two terms in each sine or cosine polynomial
    It turns out that these bounds compose additively under the composition of matrices, motivating properties (B1) through to (B3).
    In this example $\sum_{\alpha \in A(C_1)} |\alpha_1| = |-1| + |1| = 2$, $\sum_{\alpha \in A(C_1)} |\alpha_2| = |2| + |-2| = 4$, and $\sum_{\alpha \in A(C_1)} \kappa(\alpha) = 2 + 2 = 4$
    By inspecting the final matrix in \cref{Ex:PolySem}, it is clear that the following bounds hold for all $s, t \in [4]$  where $f = ( \pinterp{C_1} )_{s,t}$.
    {\small\begin{align*}
        \deg_{z_1}^{+}( f ) &\le 2
        &
        \deg_{z_1}^{-}( f ) &\le 2
        &
        \deg_{z_2}^{+}( f ) &\le 4
        &
        \deg_{z_2}^{-}( f ) &\le 4
        &
        \deg^{+}( f ) &\le 4
    \end{align*}\par}\noindent%
    Then $C_1$ satisfies (B1) through to (B3).
    Therefore, $\Bnd_{\textsf{Poly}}( C_1 )$ holds
    \qed
\end{example}

This rationale given in \cref{Ex:PolyBnd} extends to all circuits in $\ZCirc( \mathcal{G}, \mathcal{H} )$.
Since primitive gates map to constant matrices, then they trivially satisfy $\Bnd_{\textsf{Poly}}( - )$.
By construction of $\CPoly( f )$ and $\SPoly( f )$, then rotation matrices also satisfy $\Bnd_{\textsf{Poly}}( - )$.
It is then easy to show, using \cref{Prop:GateInd}, that every gate in $\Sigma_{\mathbb{Z}}( \mathcal{G}, \mathcal{H} )$ satisfies $\Bnd_{\textsf{Poly}}( - )$.
With a slightly more careful analysis, it can then be shown that this invariant is closed under sequential and parallel composition.
Intuitively, both matrix multiplication and the Kronecker tensor product yields sums of products of polynomials, in which each term can be shown to satisfy the degree bounds.
Then by \cref{Prop:CircInd}, every circuit in $\ZCirc( \mathcal{G}, \mathcal{H} )$ also satisfies $\Bnd_{\textsf{Poly}}( - )$.
Given these coefficient bounded semantics, the singularity factoring techniques of \cref{Sect:Motivation} can then be applied to obtain \cref{Cor:PolyDiff}.
All proof details can be found in \cref{Appendix:Poly}.

\begin{restatable}{theorem}{circbnd}
    \label{Thm:PolyMat}
    If $C \in \ZCirc( \mathcal{G}, \mathcal{H} )$, then $\Bnd_{\textsf{Poly}}( C )$.
\end{restatable}

\vspace{-1em}
\begin{restatable}{corollary}{diffpoly}
    \label{Cor:PolyDiff}
    If $C_1 \in \ZCirc( \mathcal{G}, \mathcal{H} )$ and $C_2 \in \ZCirc( \mathcal{G}, \mathcal{H} )$ with $\win( C_1 ) = \win( C_2 ) = n$ and $\wout( C_1 ) = \wout( C_2 ) = m$, then for each pair of indices $s \in [2^n]$ and $t \in [2^m]$, there exists a polynomial $f \in \mathbb{C}[ x_1, \ldots, x_k ]$ such that,
    \begin{itemize}
        \item (D1). $\deg_{x_j}( f ) \le 2 \lambda_j$ for each $j \in [k]$,
        \item (D2). $\deg( f ) \le \max\{ \sum_{a \in A( C )} \kappa( a ) : C \in \{ C_1, C_2 \} \} + \sum_{j=1}^{k} \lambda_j$,
        \item (D3). $\left( \interp{C_1} - \interp{C_2} \right)_{s,t}( \theta ) = 0$ if and only if $f( e^{-i\theta_1/2}, \ldots, e^{-i \theta_k/2} ) = 0$,
    \end{itemize}
    where $\lambda_j = \max\{ \sum_{a \in A( C )} |a_j| : C \in \{ C_1, C_2 \} \}$ for each $j \in [k]$.
\end{restatable}

An interesting observation is that the bounds obtained through \cref{Thm:PolyMat} were tight in \cref{Ex:PolyBnd}.
A natural question is whether these bounds are always tight, with respect to the granularity of the abstraction.
We answer this question in the positive, by showing that for each coefficient sequence $\alpha$, there exists a circuit $C$ with $A( C ) = \alpha$ such that the corresponding bound is tight.
Of course, it is not possible to reconstruct a circuit from its coefficient sequence, so some information must be lost.
To this end, we exhibit a family of circuits in \cref{Ex:RelBnd}, each of degree zero, for which arbitrarily large bounds can be obtained.
In this example, relations exist between the rotations that depend on the axes-of-rotation and the parameter-free gates in the circuit, both of which are not captured by the coefficient sequence.
In particular, both examples rely on the relations $(R_X( \beta )) (R_X( \gamma )) = R_X( \beta + \gamma )$ and $Z (R_X( \beta ) ) = ( R_X( -\beta ) ) Z$.

\begin{figure}[t]
    \begin{subfigure}[b]{0.48\textwidth}
        \input{circuits/tight}
        \caption{The circuit for $\alpha = ( ( 1, 2 ), ( 0, -1 ) )$.}
        \label{Fig:Tight}
    \end{subfigure}
    \begin{subfigure}[b]{0.48\textwidth}
        \input{circuits/relbnd}
        \caption{$(R_X( n\theta )) \circ Z \circ (R_X( n\theta ))$}
        \label{Fig:CexFamily}
    \end{subfigure}
    \caption{Circuits used in \cref{Ex:Tight} and \cref{Ex:RelBnd} to illustrate the precision of $\Bnd( - )$.}
    \vspace{-1em}
\end{figure}

\begin{example}[Necessary Bounds]
    \label{Ex:Tight}
    Let $\alpha$ be any sequence over $\mathbb{Z}^k$ with $|\alpha| = n$.
    For each $j \in [n]$, define a linear function $f_j( \theta ) = ( \alpha_j )_1 \theta_1 + \cdots + ( \alpha_j )_k \theta_k$ and a rotation gate $G_j = R_X( f_j )$.
    Now consider the circuit $C = G_1 // \cdots // G_n$ (see~\cref{Fig:Tight}).
    It follows that $A( C ) = \alpha$.
    Moreover, $(\interp{C}( \theta ))_{0,0} = \prod_{j=1}^{n} \cos( f_j( \theta ) /2 )$.
    With regard to the polynomial semantics, $\pinterp{C} = 2^{-n} \prod_{a \in \alpha} ( \prod_{j=1}^{k} z_k^{a_j} - \prod_{j=1}^{k} z_k^{-a_j} )$.
    Clearly $\deg_{x_j}^{+}( ( \pinterp{C} )_{0,0} ) = \sum_{a \in \alpha} |a_j|$ and $\deg_{x_j}^{-}( g ) = \sum_{a \in \alpha} |a_j|$ for each $j \in [k]$.
    Then $\Bnd_{\textsf{Poly}}( C )$ is tight.
    Since $\alpha$ was arbitrary, then every coefficient sequence is realizable with tight bounds.
    \qed
\end{example}

\begin{example}[Impact of Circuit Relations]
    \label{Ex:RelBnd}
    Fix $k = 1$ as the number of parameters and let $n \in \mathbb{N}$.
    Consider the circuit $C = R_X( n \theta ) \circ Z \circ R_X( n \theta )$, as illustrated in \cref{Fig:CexFamily}.
    It follows that $\interp{C}( \theta ) = ( R_X( n \theta ) ) Z ( R_X( n \theta ) ) = ( R_X( n \theta ) ) ( R_X( -n \theta ) ) Z = R_X( 0 ) Z = Z$.
    Since $\interp{C}( \theta )$ is constant, its associated polynomials have degree zero.
    However, $\Bnd_{\textsf{Poly}}( C )$ yields an upper bound of $\sum_{a \in A( C )} |a_1| = |n| + |n| = 2n$, which exceeds the true degree by $2n$.
    Since $n$ was arbitrary, this error can be made arbitrarily large.
    \qed
\end{example}

% -------------------------------------------------
\subsection{A Cutoff Theorem for Parameterized Equivalence}
\label{Sect:Equiv:Elim}

This section shows that parameterized equivalence checking reduces to parameter-free equivalence checking for quantum circuits (\cref{Thm:QuantElim}).
The proof proceeds as follows.
First,~\cref{Cor:PolyDiff} is used to characterize a family of polynomials which are identically zero if and only if the two circuits are equal.
Using \cref{Thm:Nullstellensatz}, a finite set of points $S \subseteq \mathbb{Q}^k$ can be constructed to determine if these polynomials are identically zero.
The points in $S$ are in bijection with a set of points on the complex unit circle under the transformation $x \mapsto e^{-i x/2}$.
It follows that each polynomial is identically zero if and only if $\interp{C_1}( s ) = \interp{C_2}( s )$ for all points $s \in S$.
Note that the polynomials are never explicitly constructed.
All proof details are in \cref{Appendix:Elim}.

\begin{restatable}{theorem}{exacteq}
    \label{Thm:QuantElim}
    Let $C_1 \in \ZCirc( \mathcal{G}, \mathcal{H} )$ and $C_2 \in \ZCirc( \mathcal{G}, \mathcal{H} )$ with $\win( C_1 ) = \win( C_2 )$ and $\wout( C_1 ) = \wout( C_2 )$.
    If $S_1, S_2, \ldots, S_k \subseteq [0, 4\pi)$ such that $|S_j| > 2 \lambda_j$ for each $j \in [k]$, then $\interp{C_1}( \theta ) = \interp{C_2}( \theta )$ for all $\theta \in \mathbb{R}^k$ if and only if $\interp{C_1}( v ) = \interp{C_2}( v )$ for all $v \in S_1 \times S_2 \times \cdots \times S_k$.
\end{restatable}

\begin{restatable}{corollary}{decidable}
    If $\mathcal{G}$ and $\mathcal{H}$ consist of matrices over the universal cyclotomic field, then the parameterized equivalence checking problem is decidable for $\ZCirc( \mathcal{G}, \mathcal{H} )$.
\end{restatable}

As $k$ grows large, the utility of \cref{Thm:QuantElim} decreases.
For example, if each $\lambda_j$ is $b$, then $|S_1 \times \cdots \times S_k| = ( 2b + 1 )^k$.
That is, the number of instances grows exponentially with $k$.
However, this exponential growth can be overcome by a probabilistic algorithm.
Fix a finite subset $S$ of $[0, 4\pi)^k$ and assume that $s$ is chosen at random from $S$.
If $\interp{C_1}( s ) = \interp{C_2}( s )$, then conclude that $\interp{C_1} = \interp{C_2}$, otherwise conclude that $\interp{C_2} \ne \interp{C_2}$.
Clearly, this algorithm has no false negatives, since $\interp{C_1}( s ) \ne \interp{C_2}( s )$ implies $\interp{C_2} \ne \interp{C_2}$.
A more interesting question is the false positive rate.
Note that a false positive occurs when $\interp{C_1}( s ) = \interp{C_2}( s )$ but $\interp{C_1} \ne \interp{C_2}$.
It is shown in the following theorem that the probability of a false positive decreases with order $O( 1 / |S| )$, as an application of \cref{Thm:DLSZ}.

\begin{restatable}{theorem}{probeq}
    \label{Thm:ProbEquiv}
    Let $C_1 \in \ZCirc( \mathcal{G}, \mathcal{H} )$ and $C_2 \in \ZCirc( \mathcal{G}, \mathcal{H} )$ with $\win( C_1 ) = \win( C_2 )$, $\wout( C_1 ) = \wout( C_2 )$, and $\interp{C_1} \ne \interp{C_2}$.
    For each finite subset $S \subseteq [0, 4\pi)$, if $s_1, \ldots, s_k$ are sampled at random both independently and uniformly from $S$, then
    {\small\begin{equation*}
        \Pr\left( \interp{C_1}( s_1, \ldots, s_k ) = \interp{C_2}( s_1, \ldots, s_k ) \right) \le d / |S|
    \end{equation*}\par}\noindent%
    where $d = \max\{ \sum_{\alpha \in A( C )} \kappa( \alpha ) : C \in \{ C_1, C_2 \} \} + \sum_{j=1}^{k} \lambda_j$.
\end{restatable}

% -------------------------------------------------
\section{Extending to Rational Coefficients and Global Phase}
\label{Sect:Implementation}

The methods presented in~\cref{Sect:Equiv} face several limitations.
In particular, both~\cref{Thm:QuantElim} and \cref{Thm:ProbEquiv} assume that the circuits are integral, and do not allow for equivalence up to global phase.
In this section, we show how to extend the techniques of \cref{Sect:Equiv} to handle rational circuits and global phase.
We also expand~\cref{Thm:ProbEquiv} into an algorithm, and consider the problem of angle sampling given a gate set over the universal cyclotomic field.

% -------------------------------------------------
%
\subsection{Verifying Circuits with Rational Coefficients}
\label{Sect:Implementation:Reparam}

Most parameterized quantum circuits have fractional coefficients.
For example, the equality in~\cref{Fig:MotivatingEx} is typically stated with a parameter $\theta$ on the left-hand side and the parameters $\pm \theta / 2$ on the right-hand side.
The circuits in~\cref{Fig:MotivatingEx} are related to these fractional circuits by the substitution $f( \theta ) = \theta / 2$.
Conceptually, $f: \mathbb{R}^k \to \mathbb{R}^k$ \emph{reparameterizes} the circuit, by inducing a bijection between the parameter space of the rational circuits and the parameter space of the integral circuits.
This generalizes to all examples (see~\cref{Append:Reparam} for proofs).

\begin{restatable}{lemma}{biject}
    \label{Lemma:Biject}
    Let $C_1, C_2 \in \Circ( \mathcal{G}, \mathcal{H} )$.
    If $f: \mathbb{R}^k \to \mathbb{R}^k$ is a bijective function, then $\interp{C_1} = \interp{C_2}$ if and only if $\interp{C_1} \circ f = \interp{C_2} \circ f$.
\end{restatable}

The goal of this section is to construct a syntactic transformation which eliminates all rational coefficients, which preserving the semantic interpretation via a bijective reparameterization.
A \emph{syntactic reparameterization} is a map $F: \Circ( \mathcal{G}, \mathcal{H} ) \to \Circ( \mathcal{G}, \mathcal{H} )$ with a bijective function $f: \mathbb{R}^k \to \mathbb{R}^k$ such that $\interp{F(C)} = \interp{C} \circ f$.
The simplest syntactic reparameterization is a linear rescaling of the parameters in the circuit by a non-zero rational vector.
For each vector $v \in ( \mathbb{Q} \setminus \{ 0 \} )^k$, define the map $F_v: \Circ( \mathcal{G}, \mathcal{H} ) \to \Circ( \mathcal{G}, \mathcal{H} )$ as follows.
\begin{itemize}
\item If $G \in \mathcal{G}$, then $F_v( G ) = G$.
\item If $M \in \mathcal{H}$ and $f( \theta ) = a_1 \theta_1 + a_2 \theta_2 + \cdots + a_k \theta_k + q$, then $F_v( R_M( f ) ) = R_M( g )$ where $g( \theta ) = (v_1 a_1) \theta_1 + (v_2 a_2) \theta_2 + \cdots + (v_k a_k) \theta_k + q$.
\item If $G \in \Sigma( \mathcal{G}, \mathcal{H} )$, then $F_v( C( G ) ) = C( F_v( G ) )$.
\item If $C_1, C_2 \in \Circ( \mathcal{G}, \mathcal{H} )$, then $F_v( C_1 // C_2 ) = F_v( C_1 ) // F_v( C_2 )$.
\item If $C_1, C_2 \in \Circ( \mathcal{G}, \mathcal{H} )$, then $F_v( C_2 \circ C_1 ) = F_v( C_2 ) \circ F_v( C_1 )$.
\end{itemize}

\vspace{-0.5em} % issue with restatable spacing.
\begin{restatable}{theorem}{synparam}
    For each $v \in ( \mathbb{Q} \setminus \{ 0 \} )^k$, $f: \mathbb{R}^k \to \mathbb{R}^k$ defined by $f( \theta ) = ( v_1 \theta_1, v_2 \theta_2, \ldots, v_k \theta_k )$ is bijective and $F_v$ is syntactic reparameterization with respect to $f$.
\end{restatable}

Now assume that $C_1$ and $C_2$ are circuits in $\Circ( \mathcal{G}, \mathcal{H} )$.
For the correct choice of $v$, both $F_v( C_1 )$ and $F_v( C_2 )$ are elements of $\ZCirc( \mathcal{G}, \mathcal{H} )$.
Intuitively, each $v_j$ must be chosen such that it clears the denominators of all coefficients tied to $\theta_k$ in both $C_1$ and $C_2$.
Formally, let $\denom( q )$ denote the denominator of $q \in \mathbb{Q}$ and $\lcm\{ x_1, x_2, \ldots, x_n \}$ denote the least common multiple of $x_1, x_2, \ldots, x_n \in \mathbb{Z}$.
Then for each $j \in [k]$, $X_j = \{ \denom( \alpha_j ) : \alpha \in A( C_1 ) \cdot A( C_2 ) \}$ is the set of all denominators of coefficients tied to $\theta_k$ in both $C_1$ and $C_2$.
Then $v_j = \lcm( X_j )$ for each $j \in [k]$.
Let $\clcm( C_1, C_2 )$ denote this vector.

\begin{restatable}{theorem}{fracexact}
    \label{Thm:FracExact}
    If $C_1, C_2 \in \Circ( \mathcal{G}, \mathcal{H} )$ and $v = \clcm( C_1, C_2 )$, then $F_v( C_1 ) \in \ZCirc( \mathcal{G}, \mathcal{H} )$ and $F_v( C_2 ) \in \ZCirc( \mathcal{G}, \mathcal{H} )$.
    Moreover, $\interp{C_1} = \interp{C_2}$ if and only if $\interp{F_v( C_1 )} = \interp{F_v( C_2 )}$.
\end{restatable}

\begin{corollary}
    If $\mathcal{G}$ and $\mathcal{H}$ consist of matrices over the universal cyclotomic field, then the parameterized equivalence checking problem is decidable for $\Circ( \mathcal{G}, \mathcal{H} )$.
\end{corollary}

% -------------------------------------------------
%
\subsection{Verifying Circuits Modulo Global Phase}
\label{Sect:Implementation:Param}

In~\cref{Sect:Equiv} the circuits $C_1$ and $C_2$ where defined to be equivalence when $\interp{C_1}( \theta ) = \interp{C_2}( \theta )$ for all $\theta \in \mathbb{R}^k$.
For many applications, this notion of equivalence is far too strict.
This is because $C_1$ and $C_2$ will prepare the same probability distribution provided there exists some function $p: \mathbb{R}^k \to \mathbb{R}$ such that $\interp{C_1}( \theta ) = e^{ip(\theta)} \interp{C_2}( \theta )$ for all $\theta \in \mathbb{R}^k$.
When such a function exists, we say that $C_1$ and $C_2$ are equivalent modulo global phase.
Of course, verifying the existence of an arbitrary $p$ is infeasible.
Prior work has assumed $p$ to be affine linear~\cite{HongHuang2024,PehamBurgholzer2023,XuLi2022}.
That is, $p( \theta ) = \alpha_1 \theta_1 + \cdots \alpha_k \theta_k + \beta$.
We further assume that $\alpha_1$ through to $\alpha_k$ are rational.
In this section we show how to verify the equivalence of $C_1$ and $C_2$ modulo affine rational linear global phase, under the following assumptions.
\begin{enumerate}
\item All matrices in $\mathcal{H}$ are defined over the universal cyclotomic field.
\item All matrices in $\mathcal{G}$ are injective and defined over the universal cyclotomic field.
\end{enumerate}
In practice, the second assumption restricts $\mathcal{G}$ to unitary operations and state preparation.

Since the universal cyclotomic field is closed under addition and multiplication, then every global phase will be cyclotomic when evaluated at rational multiples of $\pi$.
In general, $\alpha$ need not be rational, since there exists cyclotomic numbers of norm $1$ which are not roots of unity.
However, the periodicity of $\interp{C_1}$ an $\interp{C_2}$ ensure that $\alpha \in \mathbb{Q}^k$.
Using properties of cyclotomic numbers, such as the fact that $\mathbb{Q}( \zeta_{2n} ) = \mathbb{Q}( \zeta_n )$ for odd $n$, it is then possible to solve for $\alpha$ (if it exists).
In~\cref{Appendix:Phase}, an algorithm $\FindPhase( C_1, C_2 )$ is described to compute these coefficients.
The injectivity of $\mathcal{G}$ ensures that all coefficients can be isolated (this condition is sufficient but not necessary).
In the case where $C_1$ and $C_2$ are not equivalent up to global phase, then arbitrary coefficients are returned.
Otherwise, the function $\FindPhase( C_1, C_2 )$ returns a tuple $( z, f )$ such that $z = e^{i\beta}$ and $f( \theta ) = (-2\alpha_1) \theta_1 + \cdots + (-2\alpha_k) \theta_k$.
Then the global phase can be offset by introducing a unitary gate $z I$ and a global phase gate $R_I( f )$.
Then equivalence modulo global phase reduces to exact equivalence as follows.

\begin{restatable}{theorem}{gphase}
    \label{Thm:GPhase}
    Assume $\mathcal{G}$ and $\mathcal{H}$ consist of matrices over the universal cyclotomic field, with all gates in $\mathcal{G}$ injective.
    If $C_1, C_2 \in \ZCirc( \mathcal{G}, \mathcal{H} )$ and $( z, f ) = \FindPhase( C_1, C_2 )$, then $C_1$ is equivalent to $C_2$ modulo affine rational linear global phase if and only if the equation $\interp{C_1} = \interp{z I \circ R_I( f ) \circ C_2}$ holds.
\end{restatable}

\begin{corollary}
    If $\mathcal{G}$ and $\mathcal{H}$ satisfy assumptions (1--2), then the parameterized equivalence checking problem is decidable modulo affine rational linear global phase for $\Circ( \mathcal{G}, \mathcal{H} )$.
\end{corollary}

% -------------------------------------------------
%
\subsection{A Probabilistic Equivalence Checking Procedure}

Imagine applying \cref{Thm:ProbEquiv} to a pair of quantum circuits $C_1$ and $C_2$.
In practice, an end-user would have some desired upper bound $p \in ( 0, 1 ]$ on the false positive rate.
A simply way to bound the false positive rate is to require that $d / |S| \le p$, meaning that $d / p \le |S|$.
Since $d / p$ is positive and $|S|$ is a natural number, then the minimum value of $|S|$ which satisfies this inequality is $N = \ceil{d / p}$.
Using this optimal solution, the following algorithm is obtained.
\begin{enumerate}
\item Compute $d = \max\{ \sum_{\alpha \in A( C )} \kappa( \alpha ) : C \in \{ C_1, C_2 \} \} + \sum_{j=1}^{k} \lambda_j$.
\item Select a set $S \subseteq [ 0, 4\pi )$ such that $|S| = \ceil{d / p}$.
\item Sample $s_1, \ldots, s_k$ at random both independently and uniformly from $S$.
\item Determine if $\interp{C_1}( s_1, \ldots, s_k ) = \interp{C_2}( s_1, \ldots, s_k )$.
\end{enumerate}
The most crucial step of this algorithm is the second step.
First, the choice of $S$ must ensure that the values of $\sin( - )$ and $\cos( - )$ are exact.
As outlined in \cref{Sect:Background:Algebraic}, the simplest way to do this is to sample $S$ from $[ 0, 4\pi ) \cap \mathbb{Q}\pi$ for with for which $\sin( - )$ and $\cos( - )$ must evaluate to cyclotomic numbers.
This method is particularly effective when $\mathcal{G}$ and $\mathcal{H}$ consists purely of matrices over the universal cyclotomic field, in which case all computation can be carried out over the universal cyclotomic field.

Now, consider the elements of $\sin( S )$ and $\cos( S )$.
For each $( j / n ) \pi$ in $S$, the elements $\sin( j / n )$ and $\cos( j / n )$ will be elements of $\mathbb{Q}[ \zeta_n ]$.
Likewise, if $\ell$ is the least common denominator of all fractions in $S$, then $S \subseteq \mathbb{Q}[ \zeta_\ell ]$.
In the worst case, $\mathbb{Q}[ \zeta_\ell ]$ will be an $\ell$-dimensional vector space.
This means that the cost of addition will grow at least linearly with $\ell$, and the cost of multiplication will grow at least quadratically with $\ell$.

\begin{restatable}{theorem}{mindenom}
    \label{Thm:MinDenom}
    If $k \in \mathbb{N}$, $S \subseteq [0,k) \cap \mathbb{Q}$ and $b = |S|$, then $\lcm\{ \denom( s ) : s \in S \} \ge \ceil{b / k}$.
\end{restatable}

Let $M$ be the smallest multiple of $4$ which is greater than or equal to $N$.
It follows from \cref{Thm:MinDenom} that $S = \{ 0, ( 1/M ) 4\pi, ( 2/M ) 4\pi, \ldots, ( ( M - 1 ) / M ) 4\pi \}$ minimizes $\ell$.
This set is also easy to compute, and is therefore taken to be the definition of $S$.

\section{Related Work}
\label{Sect:Related}

In the introduction, we discussed the cutoff-based techniques~\cite{MillerBakewell2020}, which subsumes prior work such as~\cite{JeandelPerdrix2018}.
In this section, we compare to other approaches.

\textbf{Circuit Rewriting.}
It was highlighted in~\cref{Ex:RelBnd} that circuit rewriting intersects with parameterized equivalence checking.
In~\cite{PehamBurgholzer2023}, an incomplete equational theory is given for a family of parameterized circuits, which is shown to be effective for equivalence checking.
In~\cite{WateringYeung2024}, a complete set of relations are derived, under the assumption that each parameter appears exactly once in the circuit.
Relations which hold for abstract gate sets, such as $\Sigma( \mathcal{G}, \mathcal{H} )$, have yet to be explored.

\textbf{Symbolic Techniques.}
In~\cite{XuLi2022}, symbolic techniques are used to determine parameterized equivalence.
Particularly, trigonometric relations, together with the Pythagorean relation $\cos( \theta )^2 + \sin( \theta )^2 = 1$, are used to reduce equivalence checking to a family of equations over the theory of non-linear real arithmetic.
This is then solved using the Z3~\cite{MouraBjorner2008} solver as a black box.
However, the decision problem for non-linear real arithmetic is known to be double-exponential in the number of variables~\cite{JovanovicDeMoura2012,BjørnerDeMoura2019}, whereas our approach is exponential in the number of variables.

\textbf{Probabilistic Techniques.}
In~\cite{XuMolavi2023},~\cref{Thm:DLSZ} was used to determine the equivalence of parameterized quantum circuits.
However, our technique yields Laurent polynomials rather than ordinary polynomials, which we do not compute explicitly.
In~\cite{PehamBurgholzer2023}, Peham et al. show that if $v$ is sampled uniformly at random from $[0, 4\pi)^k$, then $\Pr( \interp{C_1}( v ) = \interp{C_2}( v ) ) = 0$ given $\interp{C_1} \ne \interp{C_2}$.
However, sampling $v$ from a uniform continuous distribution is impossible on a digital computer, which can only represent a countable and non-enumerable subset of real numbers~\cite{Turing1937}.
In Peham et al., floating-point is used, and presumably, the error is assumed to be uniform as well.
In our work, all computation is exact, and therefore, such assumptions do not apply.
Since there does not exist a uniform distribution for countable sets, we instead sample uniformly from a finite subset of $[0, 4\pi)$, in which case \cref{Thm:DLSZ} applies, rather than the analytic results of Peham et al.

\section{Conclusion and Future Work}
\label{Sect:Conclusion}

In this paper, we considered the problem of parameterized equivalence checking for quantum circuits.
We show that the parameterized problem can be reduced to finitely many instances of the parameter-free problem, regardless of the gate set or axes of rotation.
Consequently, the problem is decidable in the case of gate sets defined over the universal cyclotomic field.
Moreover, we show that when the number of instances becomes intractable large, there exists a probabilistic variation of the algorithm where the probability of being incorrect can be made arbitrarily small.
We have outlined how the techniques can be implemented in practice, taking into account rational coefficients, global phase, and angle sampling.

In future work, we would like to explore how these decision procedures can be implemented efficiently using circuit rewriting and sparse matrix representations.
In particular, we would like to explore angle sampling and circuit evaluation using ZX-diagrams~\cite{PehamBurgholzer2022}, tensor decision-diagrams~\cite{ZhangSaligane2024}, and model-counting~\cite{MeiBonsangue2024}, which have all proven effective in parameter-free equivalence checking.
We would also like to explore how rewriting-based techniques and symmetry reductions might help to tighten the cutoffs obtained from $\Bnd( - )$.
For example, the bound obtained in \cref{Ex:RelBnd} could be reduced to zero by viewing each relation as a rewriting rule, and then searching for a derivation which reduces the bound.

% ---------- insert bibliography -----------
\bibliography{references}{}

\begin{thebibliography}{10}

\bibitem{AbbasSutter2021}
Amira Abbas, David Sutter, Christa Zoufal, Aurelien Lucchi, Alessio Figalli, and Stefan Woerner.
\newblock The power of quantum neural networks.
\newblock {\em Nature Computational Science}, 1(6):403--409, 2021.
\newblock \href {https://doi.org/10.1038/s43588-021-00084-1} {\path{doi:10.1038/s43588-021-00084-1}}.

\bibitem{AbdullaHH13}
Parosh~Aziz Abdulla, Fr{\'{e}}d{\'{e}}ric Haziza, and Luk{\'{a}}s Hol{\'{\i}}k.
\newblock All for the price of few.
\newblock In {\em VMCAI}, volume 7737 of {\em LNCS}, pages 476--495. Springer, 2013.
\newblock \href {https://doi.org/10.1007/978-3-642-35873-9\_28} {\path{doi:10.1007/978-3-642-35873-9\_28}}.

\bibitem{Alon1999}
Noga Alon.
\newblock Combinatorial {N}ullstellensatz.
\newblock {\em Combinatorics, Probability and Computing}, 8(1–2):7--29, 1999.
\newblock \href {https://doi.org/10.1017/S0963548398003411} {\path{doi:10.1017/S0963548398003411}}.

\bibitem{AmyGlaudell2024}
Matthew Amy, Andrew~N. Glaudell, Shaun Kelso, William Maxwell, Samuel~S. Mendelson, and Neil~J. Ross.
\newblock Exact synthesis of multiqubit {Clifford}-cyclotomic circuits.
\newblock In {\em RC}, volume 14680 of {\em LNCS}, pages 238--245. Springer, 2024.
\newblock \href {https://doi.org/10.1007/978-3-031-62076-8\_15} {\path{doi:10.1007/978-3-031-62076-8\_15}}.

\bibitem{ArthurDate2022}
Davis Arthur and Prasanna Date.
\newblock A hybrid quantum-classical neural network architecture for binary classification, 2022.
\newblock URL: \url{https://arxiv.org/abs/2201.01820}, \href {https://arxiv.org/abs/2201.01820} {\path{arXiv:2201.01820}}.

\bibitem{Avanzini2024}
Martin Avanzini, Georg Moser, Romain P{\'e}choux, and Simon Perdrix.
\newblock On the hardness of analyzing quantum programs quantitatively.
\newblock In {\em Programming Languages and Systems}, volume 14577 of {\em LNCS}, pages 31--58. Springer, 2024.
\newblock \href {https://doi.org/10.1007/978-3-031-57267-8\_2} {\path{doi:10.1007/978-3-031-57267-8\_2}}.

\bibitem{Axler2014}
Sheldon Axler.
\newblock {\em Linear Algebra Done Right}.
\newblock Springer, 3rd edition, 2014.
\newblock \href {https://doi.org/10.1007/978-3-031-41026-0} {\path{doi:10.1007/978-3-031-41026-0}}.

\bibitem{BaaderNipkow1998}
Franz Baader and Tobias Nipkow.
\newblock {\em Term Rewriting and All That}.
\newblock Cambridge University Press, 1998.
\newblock \href {https://doi.org/10.1017/CBO9781139172752} {\path{doi:10.1017/CBO9781139172752}}.

\bibitem{BaezCoya2010}
John~C. Baez, Brandon Coya, and Franciscus Rebro.
\newblock Props in network theory.
\newblock {\em Theory and Applications of Categories}, 33(25):727--783, 2010.

\bibitem{BjørnerDeMoura2019}
Nikolaj Bj{\o}rne, Leonardo de~Moura, Lev Nachmanson, and Christoph~M. Wintersteiger.
\newblock {\em Programming Z3}, volume 11430 of {\em LNPSE}, pages 148--201.
\newblock Springer, 2019.
\newblock \href {https://doi.org/10.1007/978-3-030-17601-3\_4} {\path{doi:10.1007/978-3-030-17601-3\_4}}.

\bibitem{Bosma1990}
Wieb Bosma.
\newblock Canonical bases for cyclotomic fields.
\newblock {\em Applicable Algebra in Engineering, Communication and Computing}, 1:125--134, 1990.
\newblock \href {https://doi.org/10.1007/BF01810296} {\path{doi:10.1007/BF01810296}}.

\bibitem{Breuer1997}
Thomas Breuer.
\newblock Integral bases for subfields of cyclotomic fields.
\newblock {\em Applicable Algebra in Engineering, Communication and Computing}, 8:279--289, 1997.
\newblock \href {https://doi.org/10.1007/s002000050065} {\path{doi:10.1007/s002000050065}}.

\bibitem{CurienMimram2017}
Pierre-Louis Curien and Samuel Mimram.
\newblock Coherent presentations of monoidal categories.
\newblock {\em LMCS}, 13, 2017.
\newblock \href {https://doi.org/10.23638/LMCS-13(3:31)2017} {\path{doi:10.23638/LMCS-13(3:31)2017}}.

\bibitem{MouraBjorner2008}
Leonardo de~Moura and Nikolaj Bj{\o}rner.
\newblock Z3: An efficient {SMT} solver.
\newblock In {\em TACAS}, volume 4963 of {\em LNCS}, pages 337--340. Springer, 2008.

\bibitem{DemilloLipton1978}
Richard~A. Demillo and Richard~J. Lipton.
\newblock A probabilistic remark on algebraic program testing.
\newblock {\em Info. Proc. Letters}, 7(4):193--195, 1978.
\newblock \href {https://doi.org/10.1016/0020-0190(78)90067-4} {\path{doi:10.1016/0020-0190(78)90067-4}}.

\bibitem{DingHuang2024}
Qi-Ming Ding, Yi-Ming Huang, and Xiao Yuan.
\newblock Molecular docking via quantum approximate optimization algorithm.
\newblock {\em Phys. Rev. Appl.}, 21:034036, 2024.
\newblock \href {https://doi.org/10.1103/PhysRevApplied.21.034036} {\path{doi:10.1103/PhysRevApplied.21.034036}}.

\bibitem{EastinKnill2009}
Bryan Eastin and Emanuel Knill.
\newblock Restrictions on transversal encoded quantum gate sets.
\newblock {\em Phys. Rev. Let.}, 102(11):110502, 2009.
\newblock \href {https://doi.org/10.1103/physrevlett.102.110502} {\path{doi:10.1103/physrevlett.102.110502}}.

\bibitem{EmersonNamjoshi1995}
E.~Allen Emerson and Kedar~S. Namjoshi.
\newblock On reasoning about rings.
\newblock {\em Int. J. Found. Comput. Sci.}, 14(4):527--550, 2003.
\newblock \href {https://doi.org/10.1142/S0129054103001881} {\path{doi:10.1142/S0129054103001881}}.

\bibitem{DummitFoote2003}
Richard~M. Foote and David~S. Dummit.
\newblock {\em Abstract Algebra}.
\newblock Wiley, 3rd edition, 2003.

\bibitem{GilesSelinger2013}
Brett Giles and Peter Selinger.
\newblock Exact synthesis of multiqubit {Clifford+$T$} circuits.
\newblock {\em Phys. Rev. A}, 87:032332, 2013.
\newblock \href {https://doi.org/10.1103/PhysRevA.87.032332} {\path{doi:10.1103/PhysRevA.87.032332}}.

\bibitem{Herman2023}
Dylan Herman, Cody Googin, Xiaoyuan Liu, Yue Sun, Alexey Galda, Ilya Safro, Marco Pistoia, and Yuri Alexeev.
\newblock Quantum computing for finance.
\newblock {\em Nature Rev. Phys.}, 5(8):450--465, 2023.
\newblock \href {https://doi.org/10.1038/s42254-023-00603-1} {\path{doi:10.1038/s42254-023-00603-1}}.

\bibitem{HietalaRand2023}
Kesha Hietala, Robert Rand, Liyi Li, Shih-Han Hung, Xiaodi Wu, and Michael Hicks.
\newblock A verified optimizer for quantum circuits.
\newblock {\em ACM Trans. Program. Lang. Syst.}, 45(3), 2023.
\newblock \href {https://doi.org/10.1145/3604630} {\path{doi:10.1145/3604630}}.

\bibitem{HongHuang2024}
Xin Hong, Wei-Jia Huang, Wei-Chen Chien, Yuan Feng, Min-Hsiu Hsieh, Sanjiang Li, and Mingsheng Ying.
\newblock Equivalence checking of parameterised quantum circuits, 2024.
\newblock URL: \url{https://arxiv.org/abs/2404.18456}, \href {https://arxiv.org/abs/2404.18456} {\path{arXiv:2404.18456}}.

\bibitem{IpDill1993}
C.~Norris Ip and David~L. Dill.
\newblock Better verification through symmetry.
\newblock In {\em CHDL}, volume {A-32} of {\em {IFIP} Transactions}, pages 97--111. North-Holland, 1993.
\newblock \href {https://doi.org/10.5555/648251.752211} {\path{doi:10.5555/648251.752211}}.

\bibitem{JeandelPerdrix2018}
Emmanuel Jeandel, Simon Perdrix, and Renaud Vilmart.
\newblock Diagrammatic reasoning beyond clifford+t quantum mechanics.
\newblock In {\em LICS}. ACM, 2018.
\newblock \href {https://doi.org/10.1145/3209108.3209139} {\path{doi:10.1145/3209108.3209139}}.

\bibitem{JovanovicDeMoura2012}
Dejan Jovanovi{\'{c}} and Leonardo de~Moura.
\newblock Solving non-linear arithmetic.
\newblock In {\em AR}, volume 7364 of {\em LNAI}, pages 339--354. Springer, 2012.
\newblock \href {https://doi.org/10.1007/978-3-642-31365-3\_27} {\path{doi:10.1007/978-3-642-31365-3\_27}}.

\bibitem{KaiserKroening2010}
Alexander Kaiser, Daniel Kroening, and Thomas Wahl.
\newblock Dynamic cutoff detection in parameterized concurrent programs.
\newblock In {\em CAV}, volume 6174 of {\em LNCS}, pages 645--659. Springer, 2010.
\newblock \href {https://doi.org/10.1007/978-3-642-14295-6\_55} {\path{doi:10.1007/978-3-642-14295-6\_55}}.

\bibitem{KannoNakamura2024}
Shu Kanno, Hajime Nakamura, Takao Kobayashi, Shigeki Gocho, Miho Hatanaka, Naoki Yamamoto, and Qi~Gao.
\newblock Quantum computing quantum {Monte} {Carlo} with hybrid tensor network for electronic structure calculations.
\newblock {\em npj Quantum Information}, 10(1), 2024.
\newblock \href {https://doi.org/10.1038/s41534-024-00851-8} {\path{doi:10.1038/s41534-024-00851-8}}.

\bibitem{KhalimovJacobs2013}
Ayrat Khalimov, Swen Jacobs, and Roderick Bloem.
\newblock Towards efficient parameterized synthesis.
\newblock In {\em VMCAI}, volume 7737 of {\em LNCS}, pages 108--127. Springer, 2013.
\newblock \href {https://doi.org/10.1007/978-3-642-35873-9\_9} {\path{doi:10.1007/978-3-642-35873-9\_9}}.

\bibitem{MacLane2010}
Saunders~Mac Lane.
\newblock {\em Categories for the Working Mathematician}.
\newblock Springer, 2010.
\newblock \href {https://doi.org/10.1007/978-1-4757-4721-8} {\path{doi:10.1007/978-1-4757-4721-8}}.

\bibitem{Lawvere1963}
F.~William Lawvere.
\newblock Functorial semantics of algebraic theories.
\newblock {\em Proc. Natl. Acad. Sci. U.S.A.}, 50(5):869--872, 1963.
\newblock \href {https://doi.org/10.1073/pnas.50.5.869} {\path{doi:10.1073/pnas.50.5.869}}.

\bibitem{MaGovoni2020}
He~Ma, Govoni Marco, and Giulia Galli.
\newblock Quantum simulations of materials on near-term quantum computers.
\newblock {\em npj Computational Materials}, 6:85, 2020.
\newblock \href {https://doi.org/10.1038/s41524-020-00353-z} {\path{doi:10.1038/s41524-020-00353-z}}.

\bibitem{MeiBonsangue2024}
Jingyi Mei, Marcello Bonsangue, and Alfons Laarman.
\newblock Simulating quantum circuits by model counting.
\newblock In {\em CAV}, volume 14683 of {\em LNCS}, pages 555--578. Springer, 2024.

\bibitem{DekelFrankel2023}
Dekel Meirom and Steven~H. Frankel.
\newblock {PANSATZ}: pulse-based ansatz for variational quantum algorithms.
\newblock {\em Frontiers in Quantum Science and Technology}, 2, 2023.
\newblock \href {https://doi.org/10.3389/frqst.2023.1273581} {\path{doi:10.3389/frqst.2023.1273581}}.

\bibitem{MillerBakewell2020}
Hector Miller-Bakewell.
\newblock Finite verification of infinite families of diagram equations.
\newblock {\em EPTCS}, 318:27--52, 2020.
\newblock \href {https://doi.org/10.4204/eptcs.318.3} {\path{doi:10.4204/eptcs.318.3}}.

\bibitem{MillerBakewellThesis}
Hector Miller-Bakewell.
\newblock {\em Graphical Calculi and their Conjecture Synthesis}.
\newblock PhD thesis, University of Oxford, 2020.

\bibitem{NamjoshiTrefler2016}
Kedar~S. Namjoshi and Richard~J. Trefler.
\newblock Parameterized compositional model checking.
\newblock In {\em TACAS}, volume 9636 of {\em LNCS}, pages 589--606. Springer, 2016.
\newblock \href {https://doi.org/10.1007/978-3-662-49674-9\_39} {\path{doi:10.1007/978-3-662-49674-9\_39}}.

\bibitem{NielsenChuang2011}
Michael~A. Nielsen and Isaac~L. Chuang.
\newblock {\em Quantum Computation and Quantum Information}.
\newblock Cambridge University Press, 2011.

\bibitem{Novikov1955}
Pyotr Novikov.
\newblock On the algorithmic unsolvability of the word problem in group theory.
\newblock {\em Trudy Matematicheskogo Instituta imeni V.A. Steklova}, 44:3--143, 1955.

\bibitem{PehamBurgholzer2022}
Tom Peham, Lukas Burgholzer, and Robert Wille.
\newblock Equivalence checking of quantum circuits with the {ZX}-calculus.
\newblock {\em IEEE Journal on Emerging and Selected Topics in Circuits and Systems}, 12(3):662--675, 2022.
\newblock \href {https://doi.org/10.1109/jetcas.2022.3202204} {\path{doi:10.1109/jetcas.2022.3202204}}.

\bibitem{PehamBurgholzer2023}
Tom Peham, Lukas Burgholzer, and Robert Wille.
\newblock Equivalence checking of parameterized quantum circuits: Verifying the compilation of variational quantum algorithms.
\newblock In {\em ASPDAC}, pages 702--708. ACM, 2023.
\newblock \href {https://doi.org/10.1145/3566097.3567932} {\path{doi:10.1145/3566097.3567932}}.

\bibitem{PowerRobinson1997}
John Power and Edmund Robinson.
\newblock Premonoidal categories and notions of computation.
\newblock {\em Mathematical. Structures in Comp. Sci.}, 7(5):453--468, 1997.
\newblock \href {https://doi.org/10.1017/S0960129597002375} {\path{doi:10.1017/S0960129597002375}}.

\bibitem{Santagati2024}
Raffaele Santagati, Alan Aspuru-Guzik, Ryan Babbush, Matthias Degroote, Leticia González, Elica Kyoseva, Nikolaj Moll, Markus Oppel, Robert~M. Parrish, Nicholas~C. Rubin, Michael Streif, Christofer~S. Tautermann, Horst Weiss, Nathan Wiebe, and Clemens Utschig-Utschig.
\newblock Drug design on quantum computers.
\newblock {\em Nature Phys.}, 20:549--557, 2024.
\newblock \href {https://doi.org/10.1038/s41567-024-02411-5} {\path{doi:10.1038/s41567-024-02411-5}}.

\bibitem{Schwartz1980}
J.~T. Schwartz.
\newblock Fast probabilistic algorithms for verification of polynomial identities.
\newblock {\em J. ACM}, 27(4):701--717, oct 1980.
\newblock \href {https://doi.org/10.1145/322217.322225} {\path{doi:10.1145/322217.322225}}.

\bibitem{ShaikhWang2023}
Razin~A. Shaikh, Quanlong Wang, and Richie Yeung.
\newblock How to sum and exponentiate {Hamiltonians} in {ZXW} calculus.
\newblock {\em EPTCS}, 394:236--261, 2023.
\newblock \href {https://doi.org/10.4204/eptcs.394.14} {\path{doi:10.4204/eptcs.394.14}}.

\bibitem{Turing1937}
A.~M. Turing.
\newblock On computable numbers, with an application to the entscheidungsproblem.
\newblock {\em Proc. of the London Math. Soc.}, s2-42(1):230--265, 1937.
\newblock \href {https://doi.org/10.1112/plms/s2-42.1.230} {\path{doi:10.1112/plms/s2-42.1.230}}.

\bibitem{WateringYeung2024}
John van~de Wetering, Richie Yeung, Tuomas Laakkonen, and Aleks Kissinger.
\newblock Optimal compilation of parametrised quantum circuits, 2024.
\newblock \href {https://arxiv.org/abs/2401.12877} {\path{arXiv:2401.12877}}, \href {https://doi.org/10.48550/arXiv.2401.12877} {\path{doi:10.48550/arXiv.2401.12877}}.

\bibitem{Wesley2025}
Scott Wesley.
\newblock Enriched categories for parameterized circuit semantics, 2025.
\newblock \href {https://arxiv.org/abs/2501.12481} {\path{arXiv:2501.12481}}.

\bibitem{Wesley2022}
Scott Wesley, Maria Christakis, Jorge~A. Navas, Richard Trefler, Valentin W{\"u}stholz, and Arie Gurfinkel.
\newblock Verifying {S}olidity smart contracts via communication abstraction in {SmartACE}.
\newblock In {\em VMCAI}, pages 425--449. Springer, 2022.
\newblock \href {https://doi.org/10.1007/978-3-030-94583-1\_21} {\path{doi:10.1007/978-3-030-94583-1\_21}}.

\bibitem{XuMolavi2023}
Amanda Xu, Abtin Molavi, Lauren Pick, Swamit Tannu, and Aws Albarghouthi.
\newblock Synthesizing quantum-circuit optimizers.
\newblock {\em Proc. ACM Program. Lang.}, 7(PLDI), 2023.
\newblock \href {https://doi.org/10.1145/3591254} {\path{doi:10.1145/3591254}}.

\bibitem{XuLi2022}
Mingkuan Xu, Zikun Li, Oded Padon, Sina Lin, Jessica Pointing, Auguste Hirth, Henry Ma, Jens Palsberg, Alex Aiken, Umut~A. Acar, and Zhihao Jia.
\newblock Quartz: superoptimization of quantum circuits.
\newblock In {\em PLDI}, pages 625--640. ACM, 2022.
\newblock \href {https://doi.org/10.1145/3519939.3523433} {\path{doi:10.1145/3519939.3523433}}.

\bibitem{YoffeEntin2024}
Daniel Yoffe, Noga Entin, Amir Natan, and Adi Makmal.
\newblock A qubit-efficient variational selected configuration-interaction method.
\newblock {\em Quantum Science and Technology}, 10(1):015020, 2024.
\newblock \href {https://doi.org/10.1088/2058-9565/ad7d32} {\path{doi:10.1088/2058-9565/ad7d32}}.

\bibitem{ZhangSaligane2024}
Qirui Zhang, Mehdi Saligane, Hun-Seok Kim, David Blaauw, Georgios Tzimpragos, and Dennis Sylvester.
\newblock Quantum circuit simulation with fast tensor decision diagram.
\newblock In {\em ISQED}, pages 1--8. IEEE, 2024.
\newblock \href {https://doi.org/10.1109/isqed60706.2024.10528748} {\path{doi:10.1109/isqed60706.2024.10528748}}.

\bibitem{ZhaoMiao2023}
Pengzhan Zhao, Zhongtao Miao, Shuhan Lan, and Jianjun Zhao.
\newblock {Bugs4Q}: A benchmark of existing bugs to enable controlled testing and debugging studies for quantum programs.
\newblock {\em J. of Systems and Software}, 205:111805, 2023.
\newblock \href {https://doi.org/10.1016/j.jss.2023.111805} {\path{doi:10.1016/j.jss.2023.111805}}.

\bibitem{Zippel1979}
Richard Zippel.
\newblock Probabilistic algorithms for sparse polynomials.
\newblock In {\em EUROSAM}, volume~72 of {\em LNCS}, pages 216--226. Springer, 1979.
\newblock \href {https://doi.org/10.1007/3-540-09519-5\_73} {\path{doi:10.1007/3-540-09519-5\_73}}.

\end{thebibliography}

% ----------- begin appendices -------------
\newpage
\appendix
\section{Categorical Foundations for Syntax and Semantics}
\label{Appendix:Syntax}

This section introduces the category theory necessary to understand why the syntax and semantics of~\cref{Sect:Circuits}, the abstractions of~\cref{Sect:Equiv}, and the syntactic transformations of~\cref{Sect:Implementation:Reparam} are all well-defined.
First, the syntax for $\Circ( \mathcal{G}, \mathcal{H} )$ is formally defined and the inductive principles are established.
Second,  categories are introduced as a mathematical framework for modeling the semantics of typed operations with sequential composition.
Third, premonoidal categories are introduced to model the composition of concurrent processes.
It is shown that all of the structures studied in this paper are premonoidal categories, and that all of the transformations studied are free premonoidal functors.
Finally, monoidal categories are introduced to model parallel composition in circuits.
It is shown that the semantics and syntactic transformations respect the monoidal structure as well, whereas the coefficient sequences do not.

\subsection{Syntax and Structural Induction}

Universal algebra is the field of mathematics which studies mathematical objects constructed from free variables, function symbols, and constant symbols.
We use the theory of universal algebra to formally define our syntactic structures and their various interpretations.
We begin with a review of many-sorted universal algebras as described in~\cite{BaaderNipkow1998}.
We then use this framework to define the set of all gates and the set of well-formed valid circuits.

Let $S$ be a finite set whose elements are referred to as \emph{sorts}.
An \emph{$S$-signature} is defined by the following data.
\begin{itemize}
\item A set $\Sigma$ whose elements are called \emph{function symbols}.
\item A function $\textsf{dom}: \Sigma \to S^*$ called the \emph{domain function}.
\item A function $\textsf{cod}: \Sigma \to S$ called the \emph{codomain function}.
\end{itemize}
Each $S$-signature $\Sigma$ defines a language of \emph{(ground) terms}, denoted $T( \Sigma )$.
Since $\Sigma$ is many-sorted, it is necessary to first define the ground terms of each sort $s \in S$, denoted $T( \Sigma, s )$.
Formally, for each function symbol $f \in \Sigma$ with domain $d = \textsf{dom}( f )$ and arity $n = |d|$, and for all terms $t_1 \in T( \Sigma, d_1 ), \ldots, t_n \in T( \Sigma, d_n )$, there is a term $f( t_1, \ldots, t_n ) \in T( \Sigma, \textsf{cod}( f ) )$.
The base cases for this definition are the constant terms, that is, the function symbols $f \in \Sigma$ such that $|\textsf{dom}( f )| = 0$.
It can be shown that the ground terms of sort $s \in S$ are always well-defined (and can be obtained by computing a least fixed point).
The set of all ground terms is defined to be $T( \Sigma ) = \bigcup_{s \in S} T( \Sigma, s )$.

An interpretation of an $S$-signature is an assignment of sets to each sort, an assignment of values to each constant term, an an assignment of functions to each function symbol.
Formally, let $\Sigma$ be an $S$-signature.
An \emph{$S$-interpretation of $\Sigma$} consists of the following data.
\begin{itemize}
\item For each $s \in S$, a set $X_s$ whose elements are called \emph{values of sort $s$}.
\item For each $f \in \Sigma$ with $|\textsf{dom}( f )| = 0$ and $s = \textsf{cod}( f )$, a choice of $v( f ) \in X_s$.
\item For each $f \in \Sigma$ with $d = \textsf{dom}( f )$, $s = \textsf{cod}( f )$, and $n = |d|$, a choice of function $v( f ): X_{d_1} \times \cdots \times X_{d_n} \to X_s$.
\end{itemize}
Each $S$-interpretation of $\Sigma$ defines a unique function $v: T( \Sigma ) \to \bigcup_{s \in S} X_s$ which satisfies the equation $v( f( t_1, \ldots, t_n ) ) = v( f )( v( t_1 ), \ldots, v( t_n ) )$.

The syntax used in this paper is easily expressible through universal algebra.
This construction is desirable, since the various functions defined on the gate set are merely interpretations.
For the signature of gate terms, there are four types, denoted $\{ \textsf{UMat}, \textsf{NMat}, \textsf{Herm}, \textsf{Poly} \}$.
The sorts $\textsf{UMat}$ and $\textsf{NMat}$ are used to distinguish the unitary operators from the non-unitary operators.
To this end, we partition $\mathcal{G}$ into $\mathcal{G}_U \cup \mathcal{G}_N$ where, $\mathcal{G}_U = \{ M \in \mathcal{G} : MM^{\dagger} = M^{\dagger}M = I \}$.
The sorts $\textsf{Herm}$ and $\textsf{Poly}$ distinguish the elements of $\mathcal{H}$ and $\mathcal{F}$ when constructing a rotation.
Moreover, the function symbols in the gate signature are $\Sigma_G = \{ C, \textsf{Rot} \} \cup \mathcal{G} \cup \mathcal{H} \cup \mathcal{F}$ with domains and codomains as follows.
\begin{itemize}
\item $C: \textsf{UMat} \to \textsf{UMat}$ and $\textsf{Rot}: \textsf{Herm} \times \textsf{Poly} \to \textsf{UMat}$.
\item If $G \in \mathcal{G}$, then $\textsf{dom}( G ) = ()$ and $\textsf{cod}( G ) = \textsf{UMat}$.
\item If $M \in \mathcal{H}$, then $\textsf{dom}( M ) = ()$ and $\textsf{cod}( M ) = \textsf{Herm}$.
\item If $p \in \mathcal{F}$, then $\textsf{dom}( p ) = ()$ and $\textsf{cod}( p ) = \textsf{Poly}$.
\end{itemize}
Then $\Sigma( \mathcal{G}, \mathcal{H} ) = T( \Sigma_G, \textsf{UMat} ) \cup T( \Sigma_G, \textsf{NMat} )$.
The functions $\win( - )$ and $\wout( - )$ are then interpretations of $\Sigma_G$.
This is illustrated for $\win( - )$.
\begin{itemize}
\item For each sort $s \in S$, the values of $s$ are $\mathbb{N}$.
\item $v( C ) = ( n \mapsto n + 1 )$ and $v( \textsf{Rot} ) = ( ( n, m ) \mapsto n )$.
\item If $G \in \mathcal{G}$ is a $( 2^n ) \times ( 2^m )$ matrix, then $v( G ) = n$.
\item If $M \in \mathcal{G}$ is a $( 2^n ) \times ( 2^n )$ matrix, then $v( M ) = n$.
\item If $p \in \mathcal{F}$, then $v( p ) = 0$.
\end{itemize}
This coincides with the definition given in~\cref{Sect:Circuits}.

In the circuit signature, there will be a single sort, denoted $\textsf{Circ}$.
Note that this sorting ignores the number of input wires and output wires on each gate.
We will think of the number of input and output wires as a type associated with each term, once the circuit terms has been constructed.
In the circuit signature, the function symbols are $\Sigma_C = \{ ( \circ ), (//), \epsilon \} \cup \Sigma( \mathcal{G}, \mathcal{H} )$ with the following arities.
\begin{itemize}
\item $( \circ ): \textsf{Circ} \times \textsf{Circ} \to \textsf{Circ}$ and $( // ): \textsf{Circ} \times \textsf{Circ} \to \textsf{Circ}$.
\item $\textsf{dom}( \epsilon ) = ()$ and $\textsf{cod}( \epsilon ) = \textsf{Circ}$ where $\epsilon$ represents an empty wire.
\item If $G \in \Sigma( \mathcal{G}, \mathcal{H} )$, then $\textsf{dom}( G ) = ()$ and $\textsf{cod}( G ) = \textsf{Circ}$.
\end{itemize}
Note that not every term in $T( \Sigma )$ is a well-formed circuit.
This is because the sorting does not account for the number of input and output wires.
While it is possible to define an (infinite) sorting which captures well-formed circuits, this would require infinitely many $( \circ )$ and $( // )$ associated with each valid typing.
Instead, we associate a type $( n, m )$ to each valid circuit $C \in T( \Sigma )$ indicating $\win( C ) = n$ and $\wout( C ) = m$, and the symbol $\bot$ to each invalid circuit.
Indeed, this is also an interpretation of $T( \Sigma )$.
The interpretation is defined as follows.
\begin{itemize}
\item For each $s \in S$, the values of $s$ are $( \mathbb{N} \times \mathbb{N} ) \cup \{ \bot \}$.
\item $v( \circ )( x, y ) =
    \begin{cases}
        \bot & \text{if } x = \bot \text{ or } y = \bot \\
        \bot & \text{if } y_2 \ne x_1 \\
        ( y_1, x_2 ) & \text{otherwise}
    \end{cases}$
\item $v( // )( x, y ) =
    \begin{cases}
    \bot & \text{if } x = \bot \text{ or } y = \bot \\
    ( x_1 + y_1, x_2 + y_2 ) & \text{otherwise}
    \end{cases}$
\item $v( \epsilon ) = ( 1, 1 )$.
\item If $G \in \Sigma( \mathcal{G}, \mathcal{H} )$, then $v( G ) = ( \win( G ), \wout( G ) )$.
\end{itemize}
This defines a unique interpretation $\textsf{type}: T( \Sigma ) \to ( \mathbb{N} \times \mathbb{N} ) \cup \{ \bot \}$.
Then the well-formed circuits are $\Circ( \mathcal{G}, \mathcal{H} ) = \{ C \in T( \Sigma ) : \textsf{type}( C ) \ne \bot \}$.
Moreover, $\win( C ) = \textsf{type}( C )_1$ and $\wout( C ) = \textsf{type}( C )_2$.
It is straight-forward to check that $\Circ( \mathcal{G}, \mathcal{H} )$ is closed under parallel and well-formed sequential composition.

It is now possible to establish the inductive theorems for the gate algebra and the circuit algebra.
Both theorems follow from the principle of well-founded induction~\cite{BaaderNipkow1998}.
First, a \emph{well quasi-ordering} on a set $X$ is a relation $( \succeq ) \subseteq X \times X$ subject to the following conditions.
\begin{itemize}
\item \textbf{Reflexivity}.
      If $x \in X$, then $x \succeq x$.
\item \textbf{Transitivity}.
      If $x \succeq y$ and $y \succeq z$, then $x \succeq z$.
\item \textbf{Well-Founded}.
      There does not exist an infinite chain $x_1 \succ x_2 \succ x_3 \succ \cdots$ where $x \succ y$ denotes $x \succeq y$ and $x \ne y$.
\end{itemize}
If $( \succeq )$ is well-founded, then the principle of well-founded induction states that a predicate $P( - )$ holds for all elements of $X$ if and only if,
{\small\begin{equation*}
    \forall x \in X,
    \left( \forall y \in X, x \succ y \implies P( y ) \right)
    \implies P( x ).
\end{equation*}\par}\noindent%
It can be shown that the inductive nature of $T( \Sigma_G )$ and $T( \Sigma_C )$ yields quasi well-orderings on their terms.
In the case of $T( \Sigma_G )$, the proof of~\cref{Prop:GateInd} follows almost immediately.
In the case of $T( \Sigma_C )$, some care is needed to show that this well quasi-ordering interacts well with $\textsf{type}( - )$.
Once this is established, the proof of~\cref{Prop:CircInd} also follows immediately.

Given an $S$-signature $\Sigma$, the quasi well-ordering $\succeq_{\Sigma}$ is constructed as follows.
First, define $R = \{ ( t, t_j ) \in T( \Sigma ) \times T( \Sigma ) : t = f( t_1, \ldots, t_n ) \land j \in [n] \}$.
This relation associates each term in $T( \Sigma )$ with its top-level sub-terms.
However, this relation is neither transitive nor reflextive.
This can be fixed by taking the transitive symmetric closure of $R$,
Formally, $( \succeq_{\Sigma} ) = \bigcup_{n=0}^{\infty} R^n$, where $R^0$ is the identity relation.
Then $( \succeq_{\Sigma} )$ is manifestly reflextive and transitive.
It remains to be shown that $( \succeq_{\Sigma} )$ is well-founded.
Recall that the least fixed point for $T( \Sigma )$ can be defined as $\bigcup_{n=1}^{\infty} T_n( \Sigma )$, where $T_n( \Sigma )$ is the set of terms obtained after $n$ iterations.
Then there exists a function $\textsf{lv}: T( \Sigma ) \to \mathbb{N}$ such that for each $\textsf{lv}( t )$ is the least $n$ such that $t \in T_n( \Sigma )$.
Clearly, if $t \succ_\Sigma t'$, then $\textsf{lv}( t ) > \textsf{lv}( t' )$.
If there did exist an infinite descending chain $x_1 \succ_\Sigma x_2 \succ_\Sigma x_3 \succ_\Sigma \cdots$ in $T( \Sigma )$, then there would also exist an infinite descending chain $\textsf{lv}( x_1 ) > \textsf{lv}( x_2 ) > \textsf{lv}( x_3 ) > \cdots$ in $\mathbb{N}$.
However, $( > )$ is well-founded, so this is a contradiction.
This means that $( \succ_\Sigma )$ is well-founded.

\gateprop*
\begin{proof}
    Since $\Sigma( \mathcal{G}, \mathcal{H} ) \subseteq T( \Sigma_G )$, then $( \succeq )$ restricts to $\Sigma( \mathcal{G}, \mathcal{H} )$.
    Clearly, this preserves reflexivity, transitivity, and well-foundedness.
    Then the proof proceeds by well-founded induction.
    Let $G \in \Sigma( \mathcal{G}, \mathcal{H} )$.
    Assume that for all $H \in \Sigma( \mathcal{H}, \mathcal{H} )$, if $G \succeq H$, then $P( H )$ holds.
    There are three cases to consider.
    \begin{enumerate}
    \item Assume that $G \in \mathcal{G}$.
          Then $P( G )$ holds by \textbf{Base Case (1)}.
    \item Assume that $G = \textsf{Rot}( M, f )$ for some $M \in \mathcal{H}$ and $f \in \mathcal{F}$.
          Then $P( G )$ holds by \textbf{Base Case (2)}.
    \item Assume that $G = C( H )$ for some $H \in T( \Sigma, \textsf{UMat} )$.
          Then $G \succ H$ by the definition of $( \succ )$.
          Since $T( \Sigma, \textsf{UMat} ) \subseteq \Sigma( \mathcal{G}, \mathcal{H} )$, then $P( H )$ holds by the inductive hypothesis.
          Then $P( G )$ holds by \textbf{Control Induction}.
    \end{enumerate}
    In each case, $P( G )$ holds.
    These cases exhaust all function symbols in $\Sigma_G$ except for those of type \textsf{Herm} and \textsf{Poly}.
    However, the terms of type \textsf{Herm} and \textsf{Poly} are omitted in $\Sigma( \mathcal{G}, \mathcal{H} )$.
    Then the cases are exhaustive, and $P( G )$ holds.
    Since $G$ was arbitrary, then by well-founded induction, $P( G )$ holds for each $G \in \Sigma( \mathcal{G}, \mathcal{H} )$.
\end{proof}

\begin{lemma}
    \label{Lemma:WireComp}
    Let $C_1 \in T( \Sigma_C )$ and $C_2 \in T( \Sigma_C )$.
    If either $C_2 \circ C_1 \in \Circ( \mathcal{G}, \mathcal{H} )$ or $C_1 // C_2 \in \Circ( \mathcal{G}, \mathcal{H} )$, then $C_1 \in \Circ( \mathcal{G}, \mathcal{H} )$ and $C_2 \in \Circ( \mathcal{G}, \mathcal{H} )$.
\end{lemma}

\begin{proof}
    Let $C_1 \in T( \Sigma_C )$ and $C_2 \in T( \Sigma_C )$.
    There are two cases to consider.
    \begin{enumerate}
    \item Assume that $C_2 \circ C_1 \in \Circ( \mathcal{G}, \mathcal{H} )$.
          Then $\textsf{type}( C_2 \circ C_1 ) \in \mathbb{N} \times \mathbb{N}$ by definition.
          Then there exists $x \in \mathbb{N} \times \mathbb{N}$ and $y \in \mathbb{N} \times \mathbb{N}$ such that $\textsf{type}( C_2 ) = x$, $\textsf{type}( C_1 ) = y$, and $x_1 = y_2$.
          Then $\textsf{type}( C_1 ) \in \mathbb{N} \times \mathbb{N}$ and $\textsf{type}( C_2) \in \mathbb{N} \times \mathbb{N}$.
    \item Assume that $C_1 // C_2 \in \Circ( \mathcal{G}, \mathcal{H} )$.
          Then $\textsf{type}( C_2 // C_1 ) \in \mathbb{N} \times \mathbb{N}$ by definition.
          Then $\textsf{type}( C_1 ) \in \mathbb{N} \times \mathbb{N}$ and $\textsf{type}( C_2) \in \mathbb{N} \times \mathbb{N}$ by definition.
    \end{enumerate}
    In either case, $\textsf{type}( C_1 ) \in \mathbb{N} \times \mathbb{N}$ and $\textsf{type}( C_2 ) \in \mathbb{N} \times \mathbb{N}$.
    It follows by definition that  $C_1 \in \Circ( \mathcal{G}, \mathcal{H} )$ and $C_2 \in \Circ( \mathcal{G}, \mathcal{H} )$.
\end{proof}

\circprop*
\begin{proof}
    Since $\Circ( \mathcal{G}, \mathcal{H} ) \subseteq T( \Sigma_C )$, then $( \succeq )$ restricts to $\Circ( \mathcal{G}, \mathcal{H} )$.
    Clearly, this preserves reflexivity, transitivity, and well-foundedness.
    Then the proof proceeds by well-founded induction.
    Let $C \in \Circ( \mathcal{G}, \mathcal{H} )$.
    Assume that for each circuit $C' \in \Sigma( \mathcal{H}, \mathcal{H} )$, if $C \succeq C'$, then $P( C' )$ holds.
    There are four cases to consider.
    \begin{enumerate}
    \item Assume $C = \epsilon$.
          Then $P( G )$ holds by \textbf{Base Case (1)}.
    \item Assume $C \in \Sigma( \mathcal{G}, \mathcal{H} )$.
          Then $P( G )$ holds by \textbf{Base Case (2)}.
    \item Assume $C = C_2 \circ C_1$ for some $C_1 \in T( \Sigma_C )$ and $C_2 \in T( \Sigma_C )$.
          It follows that $C_1 \in \Circ( \mathcal{G}, \mathcal{H} )$ and $C_2 \in \Circ( \mathcal{G}, \mathcal{H} )$ by \cref{Lemma:WireComp}.
          Moreover, $C \succ C_1$ and $C \succ C_2$ by the definition of $( \succ )$.
          Then $P( C_1 )$ and $P( C_2 )$ hold by the inductive hypothesis.
          Then $P( C )$ holds by \textbf{Sequential Induction}.
    \item Assume $C = C_1 // C_2$ for some $C_1 \in T( \Sigma_C )$ and $C_2 \in T( \Sigma_C )$.
          It follows by a symmetric argument that $P( C )$ holds, in which \textbf{Sequential Induction} is replaced by \textbf{Parallel Induction}.
    \end{enumerate}
    These cases exhaust all of the function symbols in $\Sigma_C$.
    Then $P( C )$ holds.
    Since $C$ was arbitrary, then by well-founded induction, $P( C )$ holds for each $C \in \Circ( \mathcal{G}, \mathcal{H} )$.
\end{proof}

The remaining subsections will address the semantics of $\Circ( \mathcal{G}, \mathcal{H} )$.
In particular, premonoidal semantics and monoidal semantics will be given for $\Circ( \mathcal{G}, \mathcal{H} )$.
These should be understood as interpretations of $T( \Sigma_C )$ restricted to $\Circ( \mathcal{G}, \mathcal{H} )$.

\subsection{Categories and Sequential Composition}

A (small) category $\mathcal{C}$ describes a set of typed operations under sequential composition.
Formally, a \emph{category} is defined by the following data~\cite{MacLane2010}.
\begin{itemize}
\item A set $\mathcal{C}_0$ of types.
\item For each pair of types $( X, Y ) \in \mathcal{C}_0 \times \mathcal{C}_0$, a collection of operations $\mathcal{C}( X, Y )$.
      For each operation $f \in \mathcal{C}( X, Y )$, we write $X \xrightarrow{f} Y$.
\item For each triple of types $( X, Y, Z ) \in \mathcal{C}_0 \times \mathcal{C}_0 \times \mathcal{C}_0$, a \emph{sequential composition} function $\circ: \mathcal{C}( Y, Z ) \times \mathcal{C}( X, Y ) \to \mathcal{C}( X, Z )$.
\item For each type $X \in \mathcal{C}_0$, a trivial operation $1_X \in \mathcal{C}( X, X )$.
\end{itemize}
As in a monoid, composition should be associative and the trivial operations should be identity elements.
Then $\mathcal{C}$ is subject to the following conditions~\cite{MacLane2010}.
\begin{itemize}
\item If $X \xrightarrow{f} Y \xrightarrow{g} Z \xrightarrow{h} W$, then $h \circ (g \circ f) = (h \circ g) \circ f$.
\item If $X \xrightarrow{f} Y$, then $1_Y \circ f = f = f \circ 1_X$.
\end{itemize}

\begin{example}[Monoids Are Categories]
    \label{Ex:Monoid}
    This example shows that every monoid defines a one-type category.
    Let $M$ be a monoid with identity $e$..
    Then define a category $BM$ such that $(BM)_0 = \{ \star \}$ and $(BM)( \star, \star ) = M$.
    In this category, if $\star \xrightarrow{x} \star \xrightarrow{y} \star$, then $y \circ x := yx$.
    Clearly, $( \circ )$ is associative and has a trivial operation given by $1_\star := e$.
    In particular, $B( \mathbb{Q}^k )^*$ is a category.
    \qed
\end{example}

\begin{example}[Matrices Form Categories]
    \label{Ex:FHilb}
    Complex matrices form a category $\textbf{FHilb}$.
    The types in this category are natural numbers, corresponding to the dimensions of complex vector spaces.
    That is, $\textbf{FHilb}_0 = \mathbb{N}$.
    The operations in this category are complex matrices.
    More concretely, if $( n, m ) \in \textbf{FHilb}_0 \times \textbf{FHilb}_0$, then $\textbf{FHilb}( n, m )$ corresponds to the set of $m \times n$ matrices.
    In this category, if $x \xrightarrow{M} y \xrightarrow{N} z$, then $N \circ M := NM$.
    Clearly, $( \circ )$ is associative, with trivial operation for $n \in \textbf{FHilb}_0$ given by the $n \times n$ identity matrix.
    \qed
\end{example}

\begin{example}[Circuits Form Categories]
    \label{Ex:CircCat}
    Circuits over a gate set form a category $\mathcal{C}$.
    Let $\Sigma_0$ be a set of types and $\Sigma_1$ be a set of gates, such that each gate $G \in \Sigma_1$ has input type $\win( G )$ and output type $\wout( G )$.
    The types in the category correspond to the possible wire types.
    That is, $\mathcal{C}_0 = \Sigma_0$.
    The identities in this category are given by circuits without any gates.
    That is, for each type $X \in \mathcal{C}_0$, the identity operation $X \xrightarrow{1_X} X$ is a wire of type $X$ without any gates.
    For each gate $G \in \Sigma_1$, $G$ is a singleton circuit $G \in \mathcal{C}( \win( G ), \wout( G ) )$.
    Composition in $\mathcal{C}$ corresponds to sequential circuit composition.
    This is clearly unital and associative.
    In the case of $\Circ( \mathcal{G}, \mathcal{H} )$, $\Sigma_0 = \mathbb{N}$ and $\Sigma_1 = \Sigma( \mathcal{G}, \mathcal{H} )$.
    \qed
\end{example}

A functor is a structure-preserving mapping between categories.
Formally, if $\mathcal{C}$ and $\mathcal{D}$ are categories, then a \emph{functor} $F: \mathcal{C} \to \mathcal{D}$ from a category $\mathcal{C}$ to a category $\mathcal{D}$ consists of the following data~\cite{MacLane2010}.
\begin{itemize}
\item A translation of types $F_0: \mathcal{C}_0 \to \mathcal{D}_0$.
\item For each pair of types $( X, Y ) \in \mathcal{C}_0 \times \mathcal{C}_0$, a translation from the type $X \rightarrow Y$ to the type $F_0( X ) \rightarrow F_0( Y )$ via a family of maps $F_{X,Y}: \mathcal{C}( X, Y ) \to \mathcal{D}( F_0( X ), F_0( Y ) )$.
\end{itemize}
This data is subject to the following conditions~\cite{MacLane2010}.
\begin{itemize}
\item If $X \in \mathcal{C}_0$, then $F_{X,X}( 1_X ) = 1_{F_0( X )}$.
\item If $X \xrightarrow{f} Y \xrightarrow{g} Z$, then $F_{X,Z}( g \circ f ) = F_{Y,Z}( g ) \circ F_{X,Y}( f )$.
\end{itemize}
The functors of interest in this paper preserve both sequential composition and parallel composition.
This is explored in the next subsection.

\subsection{Premonoidal Categories and Parallel Composition}

In parallel computation, it is possible to run operations both sequentially and in parallel.
Let $X \xrightarrow{f} Y$ and $X' \xrightarrow{g} Y'$ be two processes running in parallel.
Then $( Y // g ) \circ ( f // X' )$ would denote a serialization of the trace where $f$ executes before $g$, and $( f // Y' ) \circ ( X // g )$ would denote a serialization of the trace where $g$ executes before $f$.
If $f$ and $g$ share memory, for example, then it may be the case that $( Y // g ) \circ ( f // X' ) \ne ( f // Y' ) \circ ( X // g )$.
Such operations are described by premonoidal categories~\cite{PowerRobinson1997}.
Formally, a \emph{premonoidal category} $\mathcal{C}$ is a category $\mathcal{C}$ with the following data~\cite{CurienMimram2017}.
\begin{itemize}
\item A trivial type $\mathbb{I} \in \mathcal{C}_0$.
\item For each type $X \in \mathcal{C}_0$, a functor $X // ( - )$ which executes operations on the right process.
\item For each type $Y \in \mathcal{C}_0$, a functor $( - ) // Y$ which executes operations on the left process.
\end{itemize}
This data is subject to the following conditions~\cite{CurienMimram2017}.
\begin{enumerate}
\item If $( X, Y ) \in \mathcal{C}_0 \times \mathcal{C}_0$, then $X // ( Y ) = ( X ) // Y$, which we denote $X // Y$.
\item If $( X, Y, Z ) \in \mathcal{C}_0 \times \mathcal{C}_0 \times \mathcal{C}_0$, then $( X // Y ) // Z = X // ( Y // Z )$.
\item If $X \in \mathcal{C}_0$, then $\mathbb{I} // X = X = X // \mathbb{I}$.
\item If $X \xrightarrow{f} Y$, then $\mathbb{I} // f = f = f // \mathbb{I}$.
\item If $X \xrightarrow{f} Y$ and $( Z, W ) \in \mathcal{C}_0 \times \mathcal{C}_0$, then $X // ( Y // f ) = ( X // Y ) // f$.
\item If $X \xrightarrow{f} Y$ and $( Z, W ) \in \mathcal{C}_0 \times \mathcal{C}_0$, then $( f // X ) // Y = f // ( X // Y )$.
\item If $X \xrightarrow{f} Y$ and $( Z, W ) \in \mathcal{C}_0 \times \mathcal{C}_0$, then $( X // f ) // Y = X // ( f // Y )$.
\end{enumerate}
Properties (1--3) ensure that $( // )$ defines a monoid on the types in $\mathcal{C}_0$, with $\mathbb{I}$ the unit.
This means that even if $( Y // g ) \circ ( f // X' ) \ne ( f // Y' ) \circ ( X // g )$, the input and output types will be the same regardless of the serialization of the trace.
Properties (4--7) state that there exists a unique way to execute an operation $f$ between an idle process of type $Z$ and an idle process of type $Y$.

\begin{example}[Monoids Are Premonoidal]
    \label{Ex:MonoidPMon}
    Recall the category $BM$ from~\cref{Ex:Monoid}.
    This category is trivially premonoidal.
    Since $( BM )_0 = \{ \star \}$, then a premonoidal structure on $BM$ is defined by the following data: a type $\mathbb{I} \in \mathcal{C}_0$; a functor $\star // ( - )$; a functor $( - ) // \star$.
    Clearly $\mathbb{I} = \star$ since $( BM )_0 = \{ \star \}$.
    Then $\star // ( - )$ and $( - ) // \star$ must act trivially by (3) and (4).
    These trivial acts trivially satisfy (1--2) and (5--7).
    Then $BM$ is a premonoidal category with only trivial composition in the parallel direction.
    \qed
\end{example}

\begin{example}[{$\mathbf{FHilb}$} is a Premonoidal Category]
    \label{Ex:MatPMon}
    Recall the category $\textbf{FHilb}$ form \cref{Ex:FHilb}.
    For each $n \in \textbf{FHilb}_0$, define $n \otimes ( - )$ to be the functor which map each type $x$ to $nx$ and each matrix $x \xrightarrow{M} y$ to $nx \xrightarrow{I_n \otimes M} ny$.
    Likewise, for each $n \in \textbf{FHilb}_0$, define $( - ) \otimes n$ to be the functor which map each type $x$ to $xn$ and each matrix $x \xrightarrow{M} y$ to $xn \xrightarrow{M \otimes I_n} yn$.
    Then $n \otimes ( - )$ and $( - ) \otimes m$ define a premonoidal structure on $\textbf{FHilb}$ with respect to the trivial type $\mathbb{I} = 1$.
    It is shown in~\cite{Wesley2025} that the functions which map parameters in $\mathbb{R}^k$ to operations in $\textbf{FHilb}$ also form a premonoidal category.
    We denote this category $\textbf{Param}( \mathbb{R}^k, \textbf{FHilb} )$.
    \qed
\end{example}

\begin{example}[Circuits Form Premonoidal Categories]
    \label{Ex:CircPMon}
    Recall from~\cref{Ex:CircCat} that circuits over a gate set form a category $\mathcal{C}$.
    For the purpose of this discussion, $( \otimes )$ will be used to denote the premnoidal product, and $( // )$ will be used to denote parallel wire composition.
    For each $X \in \mathcal{C}_0$, define $X \otimes ( - )$ to be the functor which maps each type $Y$ to type $X // Y$ and each circuit $Y \xrightarrow{C} Z$ to $1_X // C$.
    Likewise, for each $X \in \mathcal{C}_0$ define $( - ) \otimes X$ to be the functor which maps each type $Y$ to $Y // X$ and each circuit $Y \xrightarrow{C} Z$ to $C // 1_X$.
    That is, $X \otimes ( - )$ acts on circuits by introducing empty wires of type $X$ above the circuit, and $( - ) \otimes X$ acts on circuits by introducing empty wires of type $X$ below the circuit.
    Since the order the wires are introduced is inconsequential, then properties (1--2) and (5--7) are satisfied.
    To satisfy (3--4), a trivial type $\mathbb{I}$ must be selected such that parallel composition with $\mathbb{I}$ is the same as doing nothing.
    The only circuit with this property is the empty circuit, so $1_{\mathbb{I}}$ must be the empty circuit, and $\mathbb{I}$ must be the corresponding type.
    For example, $\mathbb{I} = 0$ in $\Circ( \mathcal{G}, \mathcal{H} )$.
    Note that $\Circ( \mathcal{G}, \mathcal{H} )$ also allows for terms of the form $C_1 // C_2$, where neither $C_1$ nor $C_2$ is the identity.
    By convention, we associate the term $C_1 // C_2$ with the semantic value $( C_1 // m ) \circ ( n // C_2 )$, where $\win( C_1 ) = n$ and $\wout( C_2 ) = m$.
    It will be shown in the next subsection why this convention is reasonable.
    \qed
\end{example}

A premonoidal functor is a structure-preserving map between premonoidal category.
Formally, if $\mathcal{C}$ and $\mathcal{D}$ are premonoidal categories, the a \emph{premonoidal functor $F: \mathcal{C} \to \mathcal{D}$} is a functor from $\mathcal{C}$ to $\mathcal{D}$ satisfying the following properties.
\begin{itemize}
\item $F_0( \mathbb{I}_{\mathcal{C}} ) = \mathbb{I}_{\mathcal{D}}$.
\item If $X \in \mathcal{C}_0$ and $Y \xrightarrow{f} Z$, then $F_{X // Y, X // Z}( X // f ) = F_0( X ) // F_{Y,Z}( f )$.
\item If $X \xrightarrow{f} Y$ and $Z \in \mathcal{C}_0$, then $F_{X // Z, Y // Z}( f // Z ) = F_{X,Y}( f ) // F_0( Z )$.
\end{itemize}
It will be shown in the next subsection that $\interp{-}$, $A(-)$, and $F_v(-)$ are all premonoidal functors.
It will follow that these maps are well-defined.

\subsection{Constructing Premonoidal Functors with Monoidal Signatures}

An important construction in category theory is the free premonoidal category.
These are the categories out of which it is easy to define premonoidal functors.
It must first be shown that every premonoidal category has an underlying monoidal signature $\Sigma$.
It will then be shown that every monoidal signature $\Sigma$ generates a premonoidal category $\Sigma^*$ such that every structure-preserving map out of $\Sigma$ defines a unique premonoidal functor out of $\Sigma^*$.
This subsection follows~\cite{CurienMimram2017}.

A \emph{monoidal signature} is a pair of sets $\Sigma_0$ and $\Sigma_1$ equipped with a pair of functions $s: \Sigma_0 \to ( \Sigma_1 )^*$ and $t: \Sigma_0 \to ( \Sigma_1 )^*$.
Intuitively, $\Sigma_0$ is the set of types in the category, and $\Sigma_1$ is the set of operations in the category.
The functions $s$ and $t$ pick out the input type and output type of each operation.
Since the category is premonoidal, then $s$ and $t$ map into $( \Sigma_1 )^*$ as opposed to $\Sigma_1$.
Given two monoidal signatures $\Sigma$ and $\Pi$, a structure-preserving transformation $\gamma: \Sigma \to \Pi$ is a pair of functions $\gamma_0: \Sigma_0 \to ( \Pi_0 )^*$ and $\gamma_1: \Sigma_1 \to \Pi_1$ such that $s \circ \gamma_1 = \gamma_0^* \circ s$ and $t \circ \gamma_1 = \gamma_0^* \circ t$.
That is, $\gamma$ respects sources and targets.

Every (small) premonoidal category $\mathcal{C}$ has an underlying monoidal signature $U( \mathcal{C} )$ defined as follows.
The set of types in $\mathcal{C}$ is $\mathcal{C}_0$, so $U( \mathcal{C} )_0 = \mathcal{C}_0$.
Then set of all operations in $\mathcal{C}$ is the disjoint union $U( \mathcal{C} )_1 = \bigsqcup_{(X,Y) \in \mathcal{C}_0 \times \mathcal{C}_0} \mathcal{C}( X, Y )$.
Then for each $f \in \mathcal{C}_1$, define $s( f ) = X$ and $t( f ) = Y$ where $( X, Y )$ is the unique element in $\mathcal{C}^0 \times \mathcal{C}^0$ such that $f \in \mathcal{C}( X, Y )$.
Moreover, if $F: \mathcal{C} \to \mathcal{D}$ is a premonoidal functor, then $U( F ): U( \mathcal{C} ) \to U( \mathcal{D} )$ is the structure-preserving map induced by the components of $F$.

Given a monoidal signature $\Sigma$, the \emph{free premonoidal category generated by $\Sigma$} is a premonoidal category $\mathcal{C}$ with a structure preserving inclusion $\iota: \Sigma \hookrightarrow U( \mathcal{C} )$ such that for each premonoidal category $\mathcal{D}$ and each structure-preserving map $\gamma: \Sigma \to U( \mathcal{D} )$, there exists a unique premonoidal functor $F: \mathcal{C} \to \mathcal{D}$ satisfying $U( F) \circ \iota = \gamma$.
In practice, this means that any structure preserving map out of $\Sigma$ defines a unique premonoidal functor out of $\mathcal{C}$ such that $F$ agrees with $\gamma$ when evaluated on the generating types and operations.
It can be shown that $\mathcal{C}$ is unique up to isomorphism, so we write $\Sigma^{\textbf{Pre}(*)}$ for \emph{the} premonoidal category generated by $\Sigma$.
Then, without loss of generality, $( \Sigma^{\textsf{Pre}(*)} )_0 = ( \Sigma_0 )^*$ and $F_0 = \gamma_0^*$.
The evaluation of $F$ on operations then follows inductively from the structure of a premonoidal category, starting from the operations in $\Sigma_1$.
The construction is tedious, and all details can be found in~\cite{CurienMimram2017}.

\begin{example}[Circuits and Free Premonoidal Categories]
    \label{Ex:FreePCirc}
    In~\cref{Ex:CircPMon}, it was shown that $\Circ( \mathcal{G}, \mathcal{H} )$ is a premonoidal category.
    Moreover, $\Circ( \mathcal{G}, \mathcal{H} )$ is the quotient of a free premonoidal category (this quotient is described in the next subsection).
    The monoidal signature $\Sigma$ used to generate this category is defined as follows.
    \begin{itemize}
    \item $\Sigma_0 = \{ \bullet \}$, since $\mathbb{N} \cong \{ \bullet \}^*$.
    \item $\Sigma_1 = \Sigma( \mathcal{G}, \mathcal{H} )$, since the gates in $\Sigma( \mathcal{G}, \mathcal{H} )$ are generating operations.
    \item $s, t: \Sigma_1 \to ( \Sigma_0 )^*$ correspond to $\win( - )$ and $\wout( - )$ respectively.
    \end{itemize}
    In the premonoidal case, parallel induction is restricted so that if $C \in \Circ( \mathcal{G}, \mathcal{H} )$ and $n \in \mathbb{N}$, then both $C // 1_n$ and $1_n // C$ are in $\Circ( \mathcal{G}, \mathcal{H} )$.
    Intuitively, premonoidal categories represent circuits as sequences of gates applied to subsets of adjacent wires, as opposed to directed acyclic graphs.
    \qed
\end{example}

\begin{example}[Semantic Interpretations]
    Let $\Sigma$ denote the monoidal signature defined in~\cref{Ex:FreePCirc}.
    The semantic interpretation map $\interp{-}$ can be defined as the free (pre)monoidal functor induced by some $\gamma: \Sigma \to U( \textbf{Param}( \mathbb{R}^k, \textbf{FHilb} ) )$.
    The first component of $\gamma$ is $\gamma_0( \bullet ) = 2$, since the state of a qubit is a $2$-dimensional vector space.
    The second component of $\gamma$ is defined to be the following interpretation of $T( \Sigma_G )$
    \begin{itemize}
    \item If $G \in \mathcal{G}$, then $v( G ) = f$ where $f( \theta ) = G$.
    \item If $M \in \mathcal{H}$, then $v( M ) = M$.
    \item If $p \in \mathcal{F}$, then $v( p ) = p$.
    \item $v( C ) = ( f \mapsto ( \theta \mapsto I \oplus f( \theta ) ) )$.
    \item $v( \textsf{Rot} ) = ( ( M, p ) \mapsto ( \theta \mapsto \cos( f( \theta ) ) I + i \sin( f( \theta ) ) M ) )$.
    \end{itemize}
    Then $\gamma$ defines a unique premonoidal functor denoted by $\interp{-}$.
    \qed
\end{example}

\begin{example}[Polynomial Semantics]
    Let $\Sigma$ denote the monoidal signature defined in~\cref{Ex:FreePCirc}.
    The polynomial semantics $\interp{-}_{\textsf{Poly}}$ can be defined as the free premonoidal functor induced by some $\gamma: \Sigma \to U( \textbf{PolyMat} )$.
    The first component of $\gamma$ is $\gamma_0( \bullet ) = 2$, since the polynomial semantics are meant to abstract the concrete semantics.
    The second component of $\gamma$ is defined to be the following interpretation of $T( \Sigma_G )$.
    \begin{itemize}
    \item If $G \in \mathcal{G}$, then $v( G ) = G$.
    \item If $M \in \mathcal{H}$, then $v( M ) = M$.
    \item If $p \in \mathcal{F}$, then $v( p ) = p$.
    \item $v( C ) = ( M \mapsto I \oplus M )$.
    \item $v( \textsf{Rot} ) = ( ( M, p ) \mapsto \textsf{CPoly}( f ) I + i \, \textsf{SPoly}( f ) M )$.
    \end{itemize}
    Then $\gamma$ defines a unique premonoidal functor denoted by $\interp{-}_{\textsf{Poly}}$.
    \qed
\end{example}

\begin{example}[Coefficient Abstraction]
    Let $\Sigma$ denote the monoidal signature defined in~\cref{Ex:FreePCirc}.
    The coefficient abstraction $A(-)$ can be defined as the free premonoidal functor induced by some $\gamma: \Sigma \to U( B( \mathbb{Q}^k )^* )$.
    The first component of $\gamma$ is $\gamma_0( \bullet ) = \star$, since $\star$ is the only type in $B( \mathbb{Q}^k )^*$.
    The second component of $\gamma$ is defined to be the following interpretation of $T( \Sigma_G )$.
    \begin{itemize}
    \item If $M \in \mathcal{H}$, then $v(M) = ()$.
    \item If $G \in \mathcal{G}$, then $v(G) = ()$.
    \item If $f \in \mathcal{F}$ and $f( \theta ) = a_1 \theta_1 + a_2 \theta_2 + \cdots + a_k + \theta_k + q$, then $v(f) = ( a_1, a_2, \ldots, a_k )$.
    \item $v( \textsf{Rot} ) = ( ( x, a ) \mapsto a )$ and $v( C ) = ( a \mapsto x )$.
    \end{itemize}
    Then $\gamma$ defines a unique premonoidal functor denoted by $A( - )$.
    \qed
\end{example}

\begin{example}[Syntactic Transformations]
    Let $\Sigma$ denote the monoidal signature defined in~\cref{Ex:FreePCirc}.
    Fix some $v \in \mathbb{Q}^k$.
    The syntactic transformation $F_v$ can be defined as the free premonoidal functor induced by some $\gamma: \Sigma \to U( \Sigma^{\textsf{Pre}(*)} )$.
    The first component of $\gamma$ is $\gamma( \bullet ) = \bullet$, since the number of wires is preserved by this family of syntactic transformations.
    The second component of $\gamma$ is defined to be the following interpretation of $T( \Sigma_G )$.
    \begin{itemize}
    \item If $G \in \mathcal{G}$, then $\gamma_1( G ) = G$.
    \item If $M \in \mathcal{H}$, then $\gamma_1( M ) = H$.
    \item If $f \in \mathcal{F}$ such that $f( \theta ) = a_1 \theta_1 + a_2 \theta_2 + \cdots + v_k \theta_k + q$, then $\gamma_1( f ) = g$ where $g( \theta ) = (a_1v_1) \theta_1 + (a_2v_2) + \theta_2 + \cdots + (a_kv_k) \theta_k$.
    \item $v( \textsf{Rot} ) = ( ( M, p ) \mapsto \textsf{Rot}( M, p ) )$ and $v( C ) = ( G \mapsto C( G ) )$.
    \end{itemize}
    Then $\gamma$ defines a unique premonoidal functor denoted by $F_v( - )$.
    \qed
\end{example}

\subsection{Monoidal Categories and Side-Effect Free Composition}

In many premonoidal categories, it is the case that for each pair of operations, $X \xrightarrow{f} Y$ and $X' \xrightarrow{g} Y'$, the equation $( Y // g ) \circ ( f // X' ) = ( f // Y' ) \circ ( X // g )$ holds.
From a computational point of view, this equation says that the operations $f$ and $g$ are side-effect free~\cite{PowerRobinson1997}.
When a premonoidal $\mathcal{C}$ satisfies this property, we say that $\mathcal{C}$ is a \emph{monoidal category}.

\begin{example}[Monoids as Monoidal Categories]
    \label{Ex:MonoidMon}
    Recall from~\cref{Ex:MonoidPMon} that $BM$ is a premonoidal category in a trivial way.
    If $BM$ is in fact a monoidal category, then $( \star // f ) \circ ( g // \star ) = ( g // \star ) \circ ( \star // f )$ for each pair of operations $\star \xrightarrow{f} \star$ and $\star \xrightarrow{g} \star$.
    Since $\star // ( - )$ and $( - ) // \star$ are trivial, then $BM$ is monoidal if and only if $fg = gf$ for all $f \in M$ and $g \in M$.
    In other words, $BM$ is a monoidal category if and only if $M$ is a commutative monoid.
    In particular, if $X$ is a set, then $B( X^* )$ is not a monoidal category.
    \qed
\end{example}

\begin{example}[{$\mathbf{FHilb}$} is a Monoidal Category]
    Recall from~\cref{Ex:MatPMon} that $\textbf{FHilb}$ is a premonoidal category with respect to the Kronecker tensor product $( \otimes )$.
    An important property of the Kronecker tensor product is bilinearity, which states that $(NM) \otimes (LK) = (M \otimes K)(N \otimes L)$ for any matrices $x \xrightarrow{M} y \xrightarrow{N} z$ and $x' \xrightarrow{K} y' \xrightarrow{L} z'$.
    This means that $( Y \otimes g ) \circ ( f \otimes X' ) = ( f \otimes Y' ) \circ ( X \otimes g )$ for any pair of matrices matrices $X \xrightarrow{f} Y$ and $X' \xrightarrow{g} Y'$.
    \qed
\end{example}

\begin{figure}[t]
    \begin{align*}
        \input{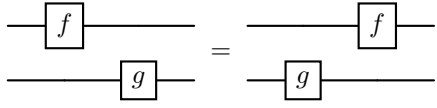}
    \end{align*}
    \vspace{-0.5em}
    \caption{A circuit diagram for the equation $( Y // g ) \circ ( f // X' ) = ( f // Y' ) \circ ( X // g )$.}
    \label{Fig:Bifunct}
    \vspace{-0.5em}
\end{figure}

\begin{example}[Circuits Form a Monoidal Category]
    Recall from~\cref{Ex:CircPMon} that circuits form monoidal categories with respect to parallel composition of wires.
    Graphically, the equation $( Y // g ) \circ ( f // X' ) = ( f // Y' ) \circ ( X // g )$ states that the two circuits in~\cref{Fig:Bifunct} should represent the same operation.
    Obviously, this is the case, so circuits are in fact monoidal categories.
    \qed
\end{example}

A monoidal functor is a premonoidal functor between monoidal categories.
This should make sense, since monoidal categories are premonoidal categories with extra properties, as opposed to extra data.
Given a monoidal signature $\Sigma$, it is also possible to construct a free monoidal category $\Sigma^*$.
Intuitively, $\Sigma^*$ is defined to be the premonoidal category $\Sigma^{\textsf{Pre}(*)}$ modulo the family of relations $( Y // g ) \circ ( f // X' ) = ( f // Y' ) \circ ( X // g )$.
This is constructed explicitly in~\cite{CurienMimram2017}.

\begin{example}[Coefficient Abstraction is not Monoidal]
    Recall from \cref{Ex:MonoidMon} that if $BM$ is a monoidal category, then $M$ is commutative.
    Since free monoids on more than one element are non-commutative, then $A( - )$ cannot be a monoidal functor.
    To illustrate way, consider the case where $k = 1$ and define two circuits, $C_1 = R_X( \theta_1 ) // 1_{\bullet}$ and $C_2 = 1_{\bullet} // R_X( -\theta_1 )$.
    It follows by definition of $A( - )$ that $A( C_1 ) = ( 1 )$ and $A( C_2 ) = ( -1 )$.
    Then,
    {\small\begin{equation*}
        A( C_1 ) \cdot A( C_2 )
        =
        ( 1, -1 )
        \ne
        ( -1, 1 )
        = A( C_2 ) \cdot A( C_1 ).
    \end{equation*}\par}\noindent%
    However, $C_1 \circ C_2 = C_2 \circ C_1$, when defined as a free monoidal category.
    \qed
\end{example}

It remains to be showing that the premonoidal abstraction $A( - )$ can be used to reason soundly about the monoidal semantics $\interp{-}$.
This is possible since $\interp{-}$ is definitionally a premonoidal functor.
Let $j: \Sigma^{\textsf{Pre(*)}} \rightarrow \Sigma^*$ denote the quotient map obtained through the construction of $\Sigma^*$.
The inclusion map from $\Sigma$ into $U( \Sigma^* )$ is precisely $U( j ) \circ \iota$.
If $\gamma: \Sigma \to U( \textbf{Param}( \mathbb{R}^k, \textbf{FHilb} ) )$ is the defining map for $\interp{-}$, then $\interp{-}$ is the unique map such that $U( \interp{-} ) \circ ( U( j ) \circ \iota ) = \gamma$.
There also exists a unique premonoidal functor $F: \Sigma^{\textsf{Pre(*)}} \to \textbf{Param}( \mathbb{R}^k, \textbf{FHilb} )$ such that $F$ is the solution to  $U( - ) \circ \iota = \gamma$.
However, $\interp{-} \circ j$ is a solution to $U( - ) \circ \iota = \gamma$.
Then $F = \interp{-} \circ j$, so the monoidal functor defined by $\gamma$ is precisely the quotient of the premonoidal functor defined by $\gamma$.
Then it suffices to prove \cref{Thm:QuantElim} and \cref{Thm:ProbEquiv} in the premonoidal setting.

\section{Proving the Correctness of the Polynomial Abstraction Bounds}
\label{Appendix:Poly}

In~\cref{Append:Poly:Abs}, the soundness of the polynomial abstract is established.
In~\cref{Append:Poly:Prelim}, some lemmas are introduced which show how certain functions naturally satisfy (B1) through to (B3).
In~\cref{Append:Poly:Gate}, these lemmas are combined with~\cref{Prop:GateInd} to show that all gates in $\Sigma_{\mathbb{Z}}( \mathcal{G}, \mathcal{H} )$ satisfy $\Bnd( - )$ when viewed as circuits with only one gate.
In~\cref{Append:Poly:Circ}, this result is then extended to show that all circuits in $\ZCirc( \mathcal{G}, \mathcal{H} )$ satisfy $\pBnd( - )$.

\subsection{Establishing the Polynomial Abstraction}
\label{Append:Poly:Abs}

This section shows that the polynomial abstraction is sound with respect to $\interp{-}$.

\polyabs*
\begin{proof}
    Let $\rho_j = e^{-i\theta_j/2}$ for each $j \in [k]$.
    First, it is shown that $\pinterp{-}$ holds for singleton circuits.
    This follows by \cref{Prop:GateInd}.
    \begin{itemize}
    \item \textbf{Base Case (1)}.
          Let $G \in \mathcal{G}$.
          Then by definition, $\interp{G}( \theta_1, \ldots, \theta_k ) = G = \pinterp{G}( \rho_1, \ldots, \rho_k )$.
    \item \textbf{Base Case (2)}.
          Let $M \in \mathcal{H}$ and $f \in \mathcal{F}$ have integral coefficients.
          Then by construction, the following equations hold.
          {\small\begin{align*}
            \CPoly( f )( \rho_1, \ldots, \rho_k ) &= \cos( -f( \theta ) / 2 )
            &
            \SPoly( f )( \rho_1, \ldots, \rho_k ) &= \sin( -f( \theta ) / 2 )
          \end{align*}\par}\noindent%
          Then by definition, the following equation holds.
          {\small\begin{align*}
              \pinterp{R_H( f )}( \rho_1, \ldots, \rho_k )
              &=
              \CPoly( f )( \rho_1, \ldots, \rho_k ) I + i\,\SPoly( f )( \rho_1, \ldots, \rho_k ) M
              \\
              &=
              \cos( -f( \theta ) / 2 ) I + i\cos( -f( \theta ) / 2 ) M
              \\
              &=
              \interp{R_H( f )}( \theta_1, \ldots, \theta_k )
          \end{align*}\par}\noindent%
    \item \textbf{Control Induction}.
          Let $G \in \Sigma_{\mathbb{Z}}( \mathcal{G}, \mathcal{H} )$.
          Assume that $G$ is unitary gate and satisfies $\interp{G}( \theta_1, \ldots, \theta_k ) = \pinterp{G}( \rho_1, \ldots, \rho_k )$.
          Since $G$ is unitary, then there exists some $n \in \mathbb{N}$ such that $\win( G ) = n = \wout( G )$.
          Then by definition, the following equation holds.
          {\small\begin{equation*}
              \pinterp{C( G )}( \rho_1, \ldots, \rho_k )
              =
              I_{2^n} \oplus \pinterp{G}( \rho_1, \ldots, \rho_k )
              =
              I_{2^n} \oplus \interp{G}( \theta_1, \ldots, \theta_k )
              =
              \interp{C( G )}( \theta_1, \ldots, \theta_k )
          \end{equation*}\par}\noindent%
    \end{itemize}
    Then by \cref{Prop:GateInd}, $\interp{G}( \theta_1, \ldots, \theta_k ) = \pinterp{C}\left( e^{-i\theta_1/2}, \ldots, e^{i\theta_k/2} \right)$ for all $G \in \Sigma_{\mathbb{Z}}( \mathcal{G}, \mathcal{H} )$.
    This can be extended to all of $\Circ( \mathcal{G}, \mathcal{H} )$ by \cref{Prop:CircInd}.
    \begin{itemize}
    \item \textbf{Base Case (1)}.
          By definition, $\pinterp{\epsilon}( \rho_1, \ldots, \rho_k ) = I_2 = \interp{\epsilon}( \theta_1, \ldots, \theta_k )$.
    \item \textbf{Base Case (2)}.
          From above, $\interp{G}( \theta_1, \ldots, \theta_k ) = \pinterp{C}( \rho_1, \ldots, \rho_k )$ for all $G \in \Sigma_{\mathbb{Z}}( \mathcal{G}, \mathcal{H} )$.
    \item \textbf{Parallel Induction}.
          Let $C_1 \in \ZCirc( \mathcal{G}, \mathcal{H} )$ and $C_2 \in \ZCirc( \mathcal{G}, \mathcal{H} )$.
          Assume that both $\pinterp{C_1}( \rho_1, \ldots, \rho_k ) = \interp{C_1}( \theta_1, \ldots, \theta_k )$ and $\pinterp{C_2}( \rho_1, \ldots, \rho_k ) = \interp{C_2}( \theta_1, \ldots, \theta_k )$.
          Then by definition, the following equation holds.
          {\small\begin{align*}
              \pinterp{C_1 // C_2}( \rho_1, \ldots, \rho_k )
              &=
              \pinterp{C_1}( \rho_1, \ldots, \rho_k ) \otimes \pinterp{C_2}( \rho_1, \ldots, \rho_k )
              \\
              &=
              \interp{C_1}( \theta_1, \ldots, \theta_k ) \otimes \interp{C_2}( \theta_1, \ldots, \theta_k )
              \\
              &=
              \interp{C_1 // C_2}( \theta_1, \ldots, \theta_k )
          \end{align*}\par}\noindent%
    \item \textbf{Sequential Induction}.
          Let $C_1 \in \ZCirc( \mathcal{G}, \mathcal{H} )$ and $C_2 \in \ZCirc( \mathcal{G}, \mathcal{H} )$.
          Assume that both $\pinterp{C_1}( \rho_1, \ldots, \rho_k ) = \interp{C_1}( \theta_1, \ldots, \theta_k )$ and $\pinterp{C_2}( \rho_1, \ldots, \rho_k ) = \interp{C_2}( \theta_1, \ldots, \theta_k )$ with $\wout( C_1 ) = \win( C_2 )$.
          Then by definition, the following equation holds.
          {\small\begin{align*}
              \pinterp{C_2 \circ C_1}( \rho_1, \ldots, \rho_k )
              &=
              \pinterp{C_2}( \rho_1, \ldots, \rho_k ) \pinterp{C_1}( \rho_1, \ldots, \rho_k )
              \\
              &=
              \interp{C_2}( \theta_1, \ldots, \theta_k ) \interp{C_1}( \theta_1, \ldots, \theta_k )
              \\
              &=
              \interp{C_2 \circ C_1}( \theta_1, \ldots, \theta_k )
          \end{align*}\par}\noindent%
    \end{itemize}
    Then by \cref{Prop:GateInd}, $\interp{C}( \theta_1, \ldots, \theta_k ) = \pinterp{C}\left( e^{-i\theta_1/2}, \ldots, e^{i\theta_k/2} \right)$ for all $C \in \ZCirc( \mathcal{G}, \mathcal{H} )$.
    Then $\pinterp{-}$ is a polynomial abstraction.
\end{proof}

\subsection{Preliminary Lemmas}
\label{Append:Poly:Prelim}

The section provides preliminary lemmas to prove~\cref{Thm:PolyMat}.
The first lemma (\cref{Lemma:ConstOp}) shows that constant polynomials trivially satisfy (B1) through to (B3).
The second lemma (\cref{Lemma:TrigPoly}) shows that the polynomials corresponding to $\sin( - )$ and $\cos( - )$ respect stricter versions of (B1) through to (B3), denoted (P1) through to (P3).

\begin{lemma}
    \label{Lemma:ConstOp}
    Let $C \in \ZCirc( \mathcal{G}, \mathcal{H} )$ with $n = \win( C )$ and $m = \wout( C )$ and $\interp{-}_{*}$ a polynomial abstraction.
    For each pair of indices $s \in [2^n]$ and $t \in [2^m]$, if $( \interp{C}_{*} )_{s,t}$ is a constant polynomial, then $C$ satisfies (B1) through to (B3) with respect to $\interp{-}_{*}$ and $( s, t )$.
\end{lemma}

\begin{proof}
    Let $s \in [2^n]$ and $t \in [2^m]$ where $n = \win( C )$ and $m = \wout( G )$.
    It remains to be shown that $C$ satisfies (B1) through to (B3) with respect to $\interp{-}_{*}$ and $( s, t )$.
    Let $f = ( \interp{C}_{*} )_{s,t}$.
    \begin{itemize}
    \item Let $j \in [k]$.
          Clearly $\sum_{\alpha \in A(C)} |\alpha_j| \ge 0$.
          Since $f$ is a constant polynomial, then either $f = 0$ and $\deg_{z_j}^{\pm}( f ) = -\infty$, or $f \ne 0$ and $\deg_{z_j}^{\pm}( f ) = 0$.
          In either case, $\deg_{z_j}^{\pm}( f ) \le \sum_{\alpha \in A(C)} |\alpha_j|$.
          Since $j$ was arbitrary, then $C$ satisfies both (B1) and (B2) with respect to $\interp{-}_{*}$ and $( s, t )$.
    \item Since $f$ is a constant polynomial, then either $f = 0$ and $\deg^{+}( f ) = -\infty$, or $f \ne 0$ and $\deg^{+}( f ) = 0$.
          In either case, $\deg^{+}( f ) \le 0$.
          Let $\alpha \in A( G )$ and $j \in [k]$.
          If $\alpha_j \ge 0$, then $\alpha_j^{+} \ge 0$ and $-\alpha_j^{-} = 0$.
          If $\alpha_j < 0$, then $\alpha_j^{+} = 0$ and $-\alpha_j^{-} > 0$.
          In either case, $\alpha_j^{+} \ge 0$ and $-\alpha_j^{-} \ge 0$.
          Since $j$ was arbitrary, then $\sum_{j=1}^{k} \alpha_k^{+} \ge 0$ and $\sum_{j=1}^{k} -\alpha_j^{-} \ge 0$.
          Then $\max\{ \sum_{j=1}^{k} \alpha_k^{+}, \sum_{j=1}^{k} -\alpha_j^{-} \} \ge 0$.
          Since $\alpha$ was arbitrary, then the following inequality holds by the monotonicity of sums.
          {\small\begin{equation*}
              \sum_{\alpha \in A(G)} \max\left\{ \sum_{j=1}^{k} \alpha_j^{+}, \sum_{j=1}^{k} -\alpha_j^{-} \right\}
              \ge
              \sum_{\alpha \in A(G)} 0
              =
              0
              \ge 
              \deg^{+}( f )
          \end{equation*}\par}\noindent%
          Then $C$ satisfies (B3) with respect to $\interp{-}_{*}$ and $( s, t )$.
    \end{itemize}
    In conclusion, $C$ satisfies (B1) through to (B3) with respect to $\interp{-}_{*}$ and $( s, t )$.
\end{proof}

\begin{lemma}
    \label{Lemma:TrigPoly}
    If $p( z_1, \ldots, z_k ) = a_1 z_1 + \cdots + a_k z_k + r \pi$ with $a \in \mathbb{Z}^k \setminus \{ 0 \}$ and $r \in \mathbb{Q}$,
    \begin{itemize}
    \item (P1). $\deg_{z_j}^{+}( \SPoly( p ) ) = \deg_{z_j}^{+}( \CPoly( p ) ) = |a_j|$ for each $j \in [k]$,
    \item (P2). $\deg_{z_j}^{-}( \SPoly( p ) ) = \deg_{z_j}^{-}( \CPoly( p ) ) = |a_j|$ for each $j \in [k]$,
    \item (P3). $\deg^{+}(\SPoly( p ) ) = \deg^{+}( \CPoly( p ) ) = \kappa( a )$.
    \end{itemize}
\end{lemma}

\begin{proof}
    Let $f = \CPoly( p )$ and $g = \SPoly( p )$.
    Since $a \ne 0$, then $f \ne 0$ and $g \ne 0$.
    It must be shown that $f$ and $g$ satisfy (P1) through to (P3).
    \begin{itemize}
    \item Let $j \in [k]$.
          There are three cases to consider.
          \begin{enumerate}
          \item Assume that $a_j = 0$.
                Then $z_j$ appears in neither $f$ nor $g$.
                Since $f \ne 0$ and $g \ne 0$, then $\deg_{z_j}^{+}( f ) = 0$ and $\deg_{z_j}^{+}( g ) = 0$.
                Then $\deg_{z_j}^{+}( f ) = \deg_{z_j}^{+}( g ) = 0 = |a_j|$.
          \item Assume that $a_j > 0$.
                Then $\deg_{z_j}^{+}( f ) = a_j$ and $\deg_{z_j}^{+}( g ) = a_j$.
                Since $a_j > 0$, then $a_j = |a_j|$.
                Then $\deg_{z_j}^{+}( f ) = \deg_{z_j}^{+}( g ) = a_j = |a_j|$.
          \item Assume that $a_j < 0$.
                Then $\deg_{z_j}^{+}( f ) = -a_j$ and $\deg_{z_j}^{+}( g ) = -a_j$.
                Since $a_j < 0$, then $-a_j = |a_j|$.
                Then $\deg_{z_j}^{+}( f ) = \deg_{z_j}^{+}( g ) = -a_j = |\alpha_j|$.
          \end{enumerate}
          In each case, $\deg_{z_j}^{+}( f ) = \deg_{z_j}^{+}( g ) = |a_j|$.
          Since $j$ was arbitrary, then $f$ and $g$ satisfy (P1).
    \item Since $\alpha \ne -\alpha$, the the following equation holds.
          {\small\begin{equation*}
              \deg^{+}( f ) =
              \deg^{+}( g ) =
              \max\left\{ \deg^{+}\left( \prod_{j=1}^{k} (x_k)^{\alpha_j} \right), \deg^{+}\left( \prod_{j=1}^{k} (x_k)^{-\alpha_j} \right) \right\}
          \end{equation*}\par}\noindent%
          Then by the additivity of $\deg^{+}( - )$,
          {\small\begin{align*}
              \deg^{+}\left( \prod_{j=1}^{k} (z_k)^{\alpha_j} \right)
              &=
              \sum_{j=1}^{k} \alpha_j^+
              &\text{and}&&
              \deg^{+}\left( \prod_{j=1}^{k} (z_k)^{-\alpha_j} \right)
              &=
              \sum_{j=1}^{k} -\alpha_j^-.
          \end{align*}\par}\noindent%
          It follows that $\deg^{+}( f ) = \deg^{+}( g ) = \kappa( a )$.
          Then $f$ and $g$ satisfy (P3).
    \end{itemize}
    Therefore, $f$ and $g$ satisfy (P1) through to (P3).
\end{proof}

\subsection{Establishing Degree Bounds for the Gate Set}
\label{Append:Poly:Gate}

This section uses \cref{Prop:GateInd} to prove that all gates in $\Sigma_{\mathbb{Z}}( \mathcal{G}, \mathcal{H} )$ satisfy $\pBnd( - )$ when viewed as singleton circuits.
For the sake of readability, each case in \cref{Prop:GateInd} is presented as a lemma.
These lemmas are combined in \cref{Lemma:PolyGates} to prove that $\pBnd( - )$ is always satisfied.

\begin{lemma}
    \label{Lemma:GenParamFree}
    Let $\interp{-}_{*}$ be a polynomial abstraction.
    If $C \in \ZCirc( \mathcal{G}, \mathcal{H} )$ and there exists a complex matrix $M$ such that $\interp{G}_{*} = M$, then $\Bnd_{*}( C )$.
\end{lemma}

\begin{proof}
    Let $s \in [2^n]$ and $t \in [2^m]$ where $n = \win( G )$ and $m = \wout( G )$.
    Then $\left( \interp{C}_{*} \right)_{s,t} = M_{s,t}$ is a constant polynomial.
    Then by \cref{Lemma:ConstOp}, $C$ satisfies (B1) through to (B3) with respect to $\interp{-}_{*}$ an $( s, t )$.
    Since $s$ and $t$ were arbitrary, then $\Bnd_{*}( G )$ holds.
\end{proof}

\begin{lemma}
    \label{Lemma:BasicRot}
    If $M \in \mathcal{H}$ and $p \in \mathcal{F}$ has integral coefficients, then $\Bnd( R_M( p ) )$.
\end{lemma}

\begin{proof}
    Let $G = R_M( p )$, $f = \CPoly( p )$, and $g = \SPoly( p )$.
    Since $p$ has integral coefficients, then there exists some $a \in \mathbb{Z}^k$ and $r \in \mathbb{Q}$ such that $p( \theta ) = a_1 \theta_1 + \cdots + a_k \theta_k + r$.
    There are two cases to consider.
    First, assume that $a  = 0$.
    Then $\CPoly( p )$ and $\SPoly( p )$ are constant by definition.
    Then $\pinterp{G}( z ) = f( z ) I + i g( z ) M$ is constant.
    Then $\pBnd( G )$ holds by~\cref{Lemma:GenParamFree}.
    Assume instead that $a \ne 0$.
    Since $a \ne 0$, then by \cref{Lemma:TrigPoly}, $f$ and $g$ satisfy (P1) through to (P3).
    Let $s \in [2^n]$ and $t \in [2^m]$ where $n = \win( G )$ and $m = \wout( G )$.
    Write $h( z ) = \left( \pinterp{G} \right)_{s,t}( z )$.
    By definition, $h( z ) = c_1 f( z ) + c_2 g( z )$ where $c_1 = I_{s,t}$ and $c_2 = i M_{s,t}$.
    It must be shown that $h$ satisfies (B1) through to (B3).
    \begin{itemize}
    \item Let $j \in [k]$.
          Then $\deg_{z_j}^{\pm}( h ) = \deg_{z_j}^{\pm}( c_1 f + c_2 g ) \le \max\{ \deg_{z_j}^{\pm}( f ), \deg_{z_j}^{\pm}( g )\}$.
          Since $f$ and $g$ satisfy (P1) and (P2), then the following inequality holds.
          {\small\begin{equation*}
            \deg_{z_j}^{\pm}( h )
            \le
            \max\left\{ \deg_{z_j}^{\pm}( f ), \deg_{z_j}^{\pm}( g ) \right\}
            =
            \max\left\{ |a_j|, |a_j| \right\}
            =
            |a_j|
            =
            \sum_{\alpha \in (a)} \alpha_j
            =
            \sum_{\alpha \in A( G )} \alpha_j
          \end{equation*}\par}\noindent%
          Since $j$ was arbitrary, then $\deg_{z_j}^{\pm}( h ) \le \sum_{\alpha \in A( G )} |\alpha_j|$ for each $j \in [k]$.
          Then $G$ satisfies (B1) and (B2) with respect to $\pinterp{-}$ and $( s, t )$.
    \item Since $h( z ) = c_1 f( z ) + c_2 g( z )$, then $\deg^{+}( h ) \le \max\{ \deg^{+}( f ), \deg^{+}( g ) \}$.
          Since $f$ and $g$ satisfy (P3), then the following inequality holds.
          {\small\begin{equation*}
              \deg^{+}( h )
              \le
              \max\left\{ \deg^{+}( f ), \deg^{+}( g ) \right\}
              =
              \max\left\{ \kappa( a ), \kappa( a ) \right\}
              =
              \kappa( a )
              =
             \sum_{\alpha \in (a)} \kappa( \alpha )
              =
             \sum_{\alpha \in A( G )} \kappa( \alpha )
          \end{equation*}\par}\noindent%
          Then $G$ satisfies (B3) with respect to $\pinterp{-}$ and $( s, t )$.
    \end{itemize}
    Since $s$ and $t$ were arbitrary, then $\pBnd( G )$ holds.
\end{proof}

\begin{lemma}
    \label{Lemma:CtrlGate}
    If $G \in \Sigma_{\mathbb{Z}}( \mathcal{G}, \mathcal{H} )$ is unitary and $\pBnd( G )$ holds, then $\pBnd( C( G ) )$ holds.
\end{lemma}

\begin{proof}
    Since $G$ is unitary, then there exists an $n \in \mathbb{N}$ such that $n = \win( G )$ and $n = \wout( G )$.
    Let $s \in [2^{n+1}]$ and $t \in [2^{n+1}]$.
    There are three cases to consider.
    \begin{itemize}
    \item Assume that $s, t \le 2^{n}$.
          Then by definition, $\left( \pinterp{C(G)} \right)_{s,t} = \left( I_{2^n} \oplus \pinterp{G} \right)_{s,t} = \left( I_{2^n} \right)_{s,t}$ is a constant polynomial.
          Then by \cref{Lemma:ConstOp}, $C(G)$ satisfies (B1) through to (B3) with respect to $\interp{-}_{*}$ and $( s, t )$.
    \item Assume that either $( s \le 2^{n} ) \land ( t > 2^{n} )$ or $( s > 2^{n} ) \land ( t \le 2^{n} )$.
          Then by definition, $\left( \pinterp{C(G)} \right)_{s,t} = \left( I_{2^n} \oplus \pinterp{G} \right)_{s,t} = 0$ is a constant polynomial.
          Then by \cref{Lemma:ConstOp}, $C(G)$ satisfies (B1) through to (B3) with respect to $\interp{-}_{*}$ and $( s, t )$.
    \item Assume that $s, t > 2^n$.
          Define $s' = s - 2^n$ and $t' = t - 2^n$.
          Then by definition, $\left( \pinterp{C(G)} \right)_{s,t} = \left( I_{2^n} \oplus \pinterp{G} \right)_{s,t} = \left( \pinterp{G} \right)_{s',t'}$.
          Since $\pBnd( G )$ holds by assumption, then in particular, $\left( \pinterp{G} \right)_{s',t'}$ satisfies equations (B1) through to (B3) with respect to $A( G )$.
          Since $\left( \pinterp{C(G)} \right)_{s,t} = \left( \pinterp{G} \right)_{s',t'}$ and $A( C( G ) ) = A( G )$, then $\left( \pinterp{C(G)} \right)_{s,t}$ satisfies the same equations with respect to $A( C( G ) )$.
          Then $C( G )$ satisfies (B1) through to (B3) with respect to $\pinterp{-}$ and $( s, t )$.
    \end{itemize}
    In each case, $G( C )$ satisfies (B1) through to (B3) with respect to $\pinterp{-}$ and $( s, t )$.
    Since $s$ and $t$ were arbitrary, then $\Bnd( C( G ) )$ holds.
\end{proof}

\begin{theorem}
    \label{Lemma:PolyGates}
    If $G \in \Sigma_{\mathbb{Z}}( \mathcal{G}, \mathcal{H} )$, then $\pBnd( G )$.
\end{theorem}

\begin{proof}
    The proof follows by \cref{Prop:GateInd}.
    \begin{itemize}
    \item \textbf{Base Case (1)}.
          Let $G \in \mathcal{G}$.
          Then by definition, $\pinterp{G} = G$ where $G$ is a complex matrix.
          Then by \cref{Lemma:GenParamFree}, $\pBnd( G )$.
    \item \textbf{Base Case (2)}.
          Let $M \in \mathcal{H}$ and $f \in \mathcal{F}$ have integral coefficients.
          Then $\pBnd( R_H( f ) )$ by \cref{Lemma:BasicRot}.
    \item \textbf{Control Induction}.
          Let $G \in \Sigma_{\mathbb{Z}}( \mathcal{G}, \mathcal{H} )$.
          Assume that $G$ is unitary gate and that $\pBnd( G )$ holds.
          Then $\pBnd( C( G ) )$ holds by \cref{Lemma:CtrlGate}.
    \end{itemize}
    Then by \cref{Prop:GateInd}, $\Bnd( - )$ holds for all elements of $\Sigma( \mathcal{G}, \mathcal{H} )$.
\end{proof}

\subsection{Establishing Degree Bounds for the Circuits}
\label{Append:Poly:Circ}

This section proves that every circuit in $\ZCirc( \mathcal{G}, \mathcal{H} )$ satisfies $\pBnd( - )$.
First, it is shown (\cref{Lemma:CompOp}) that given any two circuits $C_1$ and $C_2$ which satisfy $\pBnd( - )$, an arbitrary sum over products of the components in $\pinterp{C_1}$ and $\pinterp{C_2}$ will satisfy (B1) through to (B3) with respect to any composite of $C_1$ and $C_2$.
This lemma subsumes both sequential and parallel composition.
Using this lemma, it is shown that $C_2 \circ C_1$ (\cref{Lemma:SeqComp}) and $C_1 // C_2$ (\cref{Lemma:ParGate}) also satisfy $\pBnd( - )$.
Finally, these results are combined with \cref{Prop:CircInd} and \cref{Lemma:PolyGates} in \cref{Thm:PolyMat}, to show that every circuit in $\ZCirc( \mathcal{G}, \mathcal{H} )$ satisfies $\pBnd( - )$.

\begin{lemma}
    \label{Lemma:CompOp}
    Let $C_1, C_2 \in \ZCirc( \mathcal{G}, \mathcal{H} )$ with $n = \win( C_1 )$, $m = \wout( C_1 )$, $n' = \win( C_2 )$, and $m' = \wout( C_2 )$.
    If $\Bnd_{*}( C_1 )$ and $\Bnd_{*}( C_2 )$, then for each $J \subseteq [ n ] \times [ m ] \times [ n' ] \times [ m' ]$ the Laurent polynomial $f_J = \sum_{( s, t, s', t' ) \in J} \left( \interp{C_1}_{*} \right)_{s,t} \left( \interp{C_2}_{*} \right)_{s',t'}$ satisfies,
    \begin{itemize}
    \item (C1). $\deg_{z_j}^{+}( f ) \le \sum_{\alpha \in A( C_2 ) \cdot A( C_1 )} |\alpha_j|$ for each $j \in [k]$,
    \item (C2). $\deg_{z_j}^{-}( f ) \le \sum_{\alpha \in A( C_2 ) \cdot A( C_1 )} |\alpha_j|$ for each $j \in [k]$,
    \item (C3). $\deg^{+}( f ) \le \sum_{\alpha \in A( C_2 ) \cdot A( C_1 )} \kappa( \alpha )$.
    \end{itemize}
\end{lemma}

\begin{proof}
    Let $X = A( C_2 ) \circ A( C_1 )$.
    The proof follows by induction on the size of $J$.
    \begin{itemize}
    \item \textbf{Base Case}.
          Assume that $J = \varnothing$.
          Then $f_J$ is the zero polynomial.
          This case is follows by the same argument as \cref{Lemma:ConstOp}.
    \item \textbf{Inductive Hypothesis}.
          For some $r \in \mathbb{N}$, if $J \subseteq [n] \times [m] \times [n'] \times [m']$ and $|J| = r$, then $f_J$ satisfies (C1) through to (C3) with respect to $C_1$, $C_2$, and $J$.
    \item \textbf{Inductive Step}.
          Let $J \subseteq [n] \times [m] \times [n'] \times [m']$ and assume that $|J| = r + 1$.
          Fix some $( s, t, s', t' ) \in J$ and let $J' = J \setminus \{ ( s, t, s', t' ) \}$.
          Then $|J'| = |J| - 1 = r$.
          Then by the inductive hypothesis, $J'$ satisfies (C1) through to (C3) with respect to $\interp{-}_*$, $C_1$, and $C_2$.
          Let $g = \left( \interp{C_1}_* \right)_{s,t}$ and $h = \left( \interp{C_2}_* \right){s',t'}$.
          Then by definition, $f_J = f_{J'} + gh$.
          Since $\Bnd_{*}( C_1 )$ holds, then $C_1$ satisfies (B1) through to (B3) with respect to $\interp{-}_{*}$ and $( s, t )$.
          Since $\Bnd_{*}( C_2 )$ holds, then $C_2$ satisfies (B1) through to (B3) with respect to $\interp{-}_{*}$ and $( s', t' )$.
          It remains to be shown that $J$ satisfies (C1) through to (C3).
          \begin{itemize}
          \item Let $j \in [k]$.
                Since $C_1$ satisfies (B1) and (B2) with respect to $\interp{-}_*$ and $( s, t )$, then $\deg_{z_j}^{\pm}( g ) \le \sum_{\alpha \in A( G )} |\alpha_j|$.
                Since $C_2$ satisfies (B1) and (B2) with respect to $\interp{-}_{*}$ and $( s', t' )$, then $\deg_{z_j}^{\pm}( h ) \le \sum_{\alpha \in A( H )} |\alpha_j|$.
                Then the following inequality holds.
                {\small\begin{equation*}
                    \deg_{z_j}^{\pm}( gh )
                    \le
                    \deg_{z_j}^{\pm}( g ) + \deg_{z_j}^{\pm}( h )
                    \le
                    \sum_{\alpha \in A(C_1)} |\alpha_j| + \sum_{\alpha \in A(C_2)} |\alpha_j|
                    =
                    \sum_{\alpha \in X} |\alpha_j|
                \end{equation*}\par}\noindent%
                Next, since $J'$ satisfies (C1) and (C2) with respect to $\interp{-}_*$, $C_1$, and $C_2$, then $\deg( f_{J'} )_{z_j}^{\pm} \le \sum_{\alpha \in X} |\alpha_j|$.
                Then the following inequality holds.
                {\small\begin{equation*}
                    \deg_{z_j}^{\pm}( f_J )
                    \le
                    \max\left\{ \deg_{z_j}^{\pm}( f_{J'} ), \deg_{z_j}^{\pm}( gh ) \right\}
                    \le
                    \sum_{\alpha \in X} |\alpha_j|
                \end{equation*}\par}\noindent%
                Since $j$ was arbitrary, then $J$ satisfies (C1) and (C2) with respect to $\interp{-}_{*}$, $C_1$, and $C_2$.
          \item Since $C_1$ satisfies (B3) with respect to $\interp{-}_*$ and $( s, t )$, then $\deg^{+}( g ) \le \sum_{\alpha \in A( G )} \kappa( \alpha )$.
                Since $C_2$ satisfies (B3) with respect to $\interp{-}_*$ and $( s, t )$, then $\deg^{+}( h ) \le \sum_{\alpha \in A( H )} \kappa( \alpha )$.
                Then the following inequality holds.
                {\small\begin{equation*}
                    \deg^{+}( gh )
                    \le
                    \deg^{+}( g ) + \deg^{+}( h )
                    =
                    \sum_{\alpha \in A( C_1 )} \kappa( \alpha ) + \sum_{\alpha \in A( C_2 )} \kappa( \alpha )
                    =
                    \sum_{\alpha \in X} \kappa( \alpha )
                \end{equation*}\par}\noindent%
                Since $J'$ satisfies (C3) with respect to $\interp{-}_*$, $C_1$, and $C_2$, then $\deg^{+}( f ) \le \sum_{\alpha \in X} \kappa( \alpha )$.
                Then the following inequality holds.
                {\small\begin{equation*}
                    \deg^{+}( f_J )
                    \le
                    \max\left\{ \deg^{+}( f_{J'} ), \deg^{+}( gh ) \right\}
                    \le
                    \sum_{\alpha \in X} \kappa( \alpha )
                \end{equation*}\par}\noindent%
                Then $J$ satisfies (C3) with respect to $\interp{-}_{*}$, $C_1$, and $C_2$.
          \end{itemize}
          Then the inductive step holds.
    \end{itemize}
    Then by the principle of induction, for each choice of $J \subseteq [n] \times [m] \times [n'] \times [m']$, the Laurent polynomial $f_J$ satisfies (C1) through to (C4) with respect to $C_1$, $C_2$, and $J$.
\end{proof}

\begin{lemma}
    \label{Lemma:SeqComp}
    If $C_1, C_2 \in \ZCirc( \mathcal{G}, \mathcal{H} )$ such that $\win( C_2 ) = \wout ( C_1 )$ with $\pBnd( C_1 )$ and $\pBnd( C_2 )$, then $\pBnd( C_2 \circ C_1 )$.
\end{lemma}

\begin{proof}
    Let $s \in [2^n]$ and $t \in [2^m]$ where $n = \win( C_1 )$ and $m = \wout( C_2 )$.
    Moreover, let $\ell = \win( C_1 ) = \wout( C_2 )$.
    Define $J = \{ ( t, j, j, s ) : j \in [\ell] \}$.
    Since $\pBnd( C_1 )$ and $\pBnd( C_2 )$ hold, then by \cref{Lemma:CompOp}, $J$ satisfies (C1) through to (C3) with respect to $\pinterp{-}$, $C_1$, and $C_2$.
    By definition of $f_J$,  the following equation holds.
    {\small\begin{equation*}
        f_J
        =
        \sum_{(a,b,c,d) \in J} \left( \pinterp{C_2} \right)_{a,b} \left( \pinterp{C_1} \right)_{c,d}
        =
        \sum_{j \in [\ell]} \left( \pinterp{C_2} \right)_{t,j} \left( \pinterp{C_1} \right)_{j,s}
        =
        \left( \pinterp{C_2 \circ C_1} \right)_{s,t}
    \end{equation*}\par}\noindent%
    It remains to be shown that $C_2 \circ C_2$ satisfies (B1) to (B3) with respect to $\pinterp{-}$ and $( s, t )$.
    \begin{itemize}
    \item Since $A( C_2 \circ C_1 ) = A( C_2 ) \cdot A( C_1 )$ and $f_J = \left( \pinterp{C_2 \circ C_1} \right)_{s,t}$, then $J$ satisfies (C1) with respect to $\pinterp{-}$, $C_1$, and $C_2$, if and only if $C_2 \circ C_1$ satisfies (B1) with respect to $\pinterp{-}$ and $( s, t )$.
          Therefore, (B1) is satisfied.
    \item By the same argument, $C_2 \circ C_1$ satisfies (B2) with respect to $\pinterp{-}$ and $( s, t )$.
    \item By the same argument, $C_2 \circ C_1$ satisfies (B2) with respect to $\pinterp{-}$ and $( s, t )$.
    \end{itemize}
    Since $s$ and $t$ were arbitrary, then $\Bnd( C_2 \circ C_1 )$ holds.
\end{proof}

\begin{lemma}
    \label{Lemma:ParGate}
    If $C_1 \in \ZCirc( \mathcal{G} )$ satisfies $\pBnd( C_1 )$ and $C_2 \in \ZCirc( \mathcal{G} )$ satisfies $\pBnd( C_2 )$, then $\pBnd( C_1 // C_2 )$.
\end{lemma}

\begin{proof}
    Let $s \in [2^{n+n'}]$ and $t \in [2^{m+m'}]$ where $n = \win( C_1 )$, $n' = \win( C_2 )$, $m = \win( C_1 )$, and $m' = \wout( C_2 )$.
    Then by the definition of $\otimes$, there exists indices $q, q', r, r' \in \mathbb{N}$ such that $\left( \pinterp{C_1} \otimes \pinterp{C_2} \right)_{s,t} = \left( \pinterp{C_1} \right)_{q,r} \left( \pinterp{C_2} \right)_{q',r'}$.
    Since $\pBnd( C_1 )$ and $\pBnd( C_2 )$ hold, then by \cref{Lemma:CompOp}, $J$ satisfies (C1) through to (C3) with respect to $\pinterp{-}$, $C_2$, and $C_1$.
    By definition of $f_J$,  the following equation holds.
    {\small\begin{equation*}
        f_J
        =
        \sum_{(a,c,b,d) \in J} \left( \pinterp{C_2} \right)_{a,b} \left( \pinterp{C_1} \right)_{c,d}
        =
        \left( \pinterp{C_2} \right)_{q,r} \left( \pinterp{C_1} \right)_{q',r'}
        =
        \pinterp{C_1 // C_2}
    \end{equation*}\par}\noindent%
    It remains to be shown that $C_1 // C_2$ satisfies (B1) to (B3) with respect to $\pinterp{-}$ and $( s, t )$.
    \begin{itemize}
    \item Since $A( C_1 // C_2 ) = A( C_1 ) \cdot A( C_2 )$ and $f_J = \left( \pinterp{C_1 // C_2} \right)_{s,t}$, then $J$ satisfies (C1) with respect to $\pinterp{-}$, $C_2$, and $C_1$, if and only if $C_1 // C_2$ satisfies (B1) with respect to $\pinterp{-}$ and $( s, t )$.
          Therefore, (B1) is satisfied.
    \item By the same argument, $C_1 // C_2$ satisfies (B2) with respect to $\pinterp{-}$ and $( s, t )$.
    \item By the same argument, $C_1 // C_2$ satisfies (B3) with respect to $\pinterp{-}$ and $( s, t )$.
    \end{itemize}
    Since $s$ and $t$ were arbitrary, then $\Bnd( C_2 // C_1 )$ holds.
\end{proof}

\circbnd*
\begin{proof}
    The proof follows by \cref{Prop:CircInd}.
    \begin{itemize}
    \item \textbf{Base Case (1)}.
          Since $\interp{\epsilon} = I_2$, then $\pBnd( \epsilon )$ by \cref{Lemma:GenParamFree}.
    \item \textbf{Base Case (2)}.
          If $G \in \Sigma_{\mathbb{Z}}( \mathcal{G}, \mathcal{H} )$, then $\pBnd( G )$ by \cref{Lemma:PolyGates}.
    \item \textbf{Parallel Induction}.
          Let $C_1, C_2 \in \ZCirc( \mathcal{G}, \mathcal{H} )$.
          Assume that both $\pBnd( C_1 )$ and $\pBnd( C_2 )$ hold.
          Then $\pBnd( C_1 // C_2 )$ holds by \cref{Lemma:ParGate}.
    \item \textbf{Sequential Induction}.
          Let $C_1, C_2 \in \ZCirc( \mathcal{G}, \mathcal{H} )$.
          Assume that both $\pBnd( C_1 )$ and $\pBnd( C_2 )$ hold with $\win( C_2 ) = \wout( C_1 )$.
          Then $\pBnd( C_2 \circ C_1 )$ holds by \cref{Lemma:SeqComp}.
    \end{itemize}
    Then by \cref{Prop:CircInd}, $\pBnd( - )$ holds for all elements of $\ZCirc( \mathcal{G}, \mathcal{H} )$.
\end{proof}

\section{Proving the Quantifier Elimination Scheme}
\label{Appendix:Elim}

\diffpoly*
\begin{proof}
    Let $s \in [2^n]$ and $t \in [2^m]$.
    Since $G \in \ZCirc( \mathcal{G}, \mathcal{H} )$, then $\Bnd( G )$ holds by \cref{Thm:PolyMat}.
    Then there exists an $f \in \mathbb{C}[ x_1, x_1^{-1}, \ldots, x_k, x_k^{-1} ]$ such that $f$ satisfies (B1) through to (B4) with respect to $G$ and $( s, t )$.
    Since $H \in \ZCirc( \mathcal{G}, \mathcal{H} )$, then $\Bnd( H )$ holds by \cref{Thm:PolyMat}.
    Then there exists $g \in \mathbb{C}[ x_1, x_1^{-1}, \ldots, x_k, x_k^{-1} ]$ such that $g$ satisfies (B1) through to (B4) with respect to $H$ and $( s, t )$.
    First, the polynomial $h$ is constructed by clearing all denominators of $f - g$ with the term $\prod_{j=1}^k (x_j)^{\beta_j}$.
    Let $j \in [k]$.
    Since $f$ satisfies (B2) with respect to $G$ and $( s, t )$, then $\deg_{x_j}^{-}( f ) \le \sum_{\alpha \in C( G )} |\alpha_j|$.
    Likewise, since $g$ satisfies (B2) with respect to $H$ and $( s, t )$, then $\deg_{x_j}^{-}( g ) \le \sum_{\alpha \in C( H )} |\alpha_j|$.
    It follows that,
    {\small\begin{equation*}
        \deg_{x_j}^{-}\left( \left( \prod_{j=1}^k (x_j)^{\lambda_j} \right) (f - g) \right)
        =
        \max\{ \deg_{x_j}^{-}( f ), \deg_{x_j}^{-}( g ) \} - \lambda_j
        \le
        \lambda_j - \lambda_j
        =
        0
    \end{equation*}\par}\noindent%
    Since $j$ was arbitrary, then $h \in \mathbb{C}[ x_1, \ldots, x_k ]$ where,
    {\small\begin{equation*}
        h( x_1, \ldots, x_k )
        =
        \left( \prod_{j=1}^k (x_j)^{\lambda_j} \right)( f( x_1, \ldots, x_k ) - g( x_1, \ldots, x_k ) ).
    \end{equation*}\par}\noindent%
    This completes the construction of $h$.
    It remains to be shown that $h$ satisfies (D1) through to (D3) with respect to $G$, $H$, and $( s, t )$.
    \begin{enumerate}
    \item Let $j \in [k]$.
          Since $\deg_{x_j}( h ) = \deg_{x_j}^{+}( h )$,
          {\small\begin{equation*}
            \deg_{x_j}( h )
            \le
            \deg_{x_j}\left( \prod_{k=1}^{k} ( x_j )^{\lambda_j} \right) + \max\{ \deg_{x_j}^{+}( f ), \deg_{x_j}^{+}( g ) \} \le \lambda_j + \lambda_j.
          \end{equation*}\par}\noindent%
          Since $j$ was arbitrary, then $\deg_{x_j}( h ) \le 2 \lambda_j$ for each $j \in [k]$.
          Then $h$ satisfies (D1) with respect to $G$, $H$, and $( s, t )$.
    \item Since $\deg( h ) = \deg^{+}( h )$,
          {\small\begin{equation*}
            \deg( h )
            \le
            \deg^{+}\left( \prod_{j=1}^{k} ( x_j )^{\lambda_j} \right) + \max\{ \deg^{+}( f ), \deg^{+}( g ) \}.
          \end{equation*}\par}\noindent%
          Let $d = \max\{ \deg^{+}( f ), \deg^{+}( g ) \}$/
          Since $f$ satisfies (B3) with respect to $G$ and $( s, t )$, then $\deg^{+}( f ) \le \sum_{\alpha \in C( G )} \kappa( \alpha )$.
          Since $g$ satisfies (B3) with respect to $H$ and $( s, t )$, then $\deg^{+}( g ) \le \sum_{\alpha \in C( H )} \kappa( \alpha )$.
          Then by the monotonicity of $\max$, $d \le \max\{ \sum_{\alpha \in A(C)} \kappa( \alpha ) : C \in \{ G, H \} \}$.
          It follows that,
          {\small\begin{equation*}
            \deg( h )
            \le
            \deg^{+}\left( \prod_{j=1}^{k} ( x_j )^{\lambda_j} \right) + d
            \le
            \sum_{j=1}^{k} \lambda_j + \max\left\{ \sum_{\alpha \in A(C)} \kappa( \alpha ) : C \in \{ G, H \} \right\}.
          \end{equation*}\par}\noindent%
          Then $h$ satisfies (D2) with respect to $G$, $H$, and $( s, t )$.
    \item Write $z_j = \exp( -i\theta_j / 2 )$ for each $j \in [k]$.
          Since $f$ satisfies (B4) with respect to $G$ and $( s, t )$, then $\interp{G}_{s,t}( \theta ) = f( z_1, \ldots, z_k )$.
          Since $g$ satisfies (B4) with respect to $H$ and $( s, t )$, then $\interp{H}_{s,t}( \theta ) = g( z_1, \ldots, z_k )$.
          Then,
          {\small\begin{equation*}
            h( z_1, \ldots, z_k )
            =
            \left( \prod_{j=1}^k (z_j)^{\lambda_j} \right)( \interp{G}_{s,t}( \theta ) - \interp{H}_{s,t}( \theta ) )
          \end{equation*}\par}\noindent%
          Since $\exp( i- )$ does not have any zeros on $\mathbb{R}$, then, $\prod_{j=1}^k (z_j)^{\lambda_j} \ne 0$.
          This means that $( \interp{G} - \interp{H} )_{s,t}( \theta ) = 0$ if and only if $h( z_1, \ldots, z_k ) = 0$.
          Then $h$ satisfies (D3) with respect to $G$, $H$, and $( s, t )$.
    \end{enumerate}
    Since $s$ and $t$ were arbitrary, then the proof is complete.
    \qed
\end{proof}

\exacteq*
\begin{proof}
    If $\interp{G} = \interp{H}$, then $\interp{G}(v) = \interp{H}(v)$ for all $v \in X_1 \times \cdots \times X_k$ by definition of function equality.
    Assume instead that $\interp{G} \ne \interp{H}$.
    Then $\interp{G} - \interp{H} \ne 0$.
    Then there exists a pair of indices $s \in [2^n]$ and $t \in [2^m]$ such that $(\interp{G} - \interp{H})_{s,t} \ne 0$.
    Then by \cref{Cor:PolyDiff}, there exists a polynomial $f \in \mathbb{C}[x_1, \ldots, x_k]$ which satisfy (D1) through to (D3) with respect to $G$, $H$, and $( s, t )$.
    Since $f$ satisfies (D3) with respect to $G$, $H$, and $( s, t )$, then $f \ne 0$.
    For each $j \in [k]$, define a set,
    {\small\begin{equation*}
        X^*_j = \{ \exp( -ix/2 ) \mid x \in X_j \}.
    \end{equation*}\par}\noindent%
    Since $x \mapsto \exp( -ix / 2 )$ is bijective on $[0, 4\pi)$, $|X_j^*| = |X_j| > 2\lambda_j$ for each $j \in [k]$.
    Since $f$ satisfies (D1) with respect to $G$, $H$, and $( s, t )$, then $\deg_{x_j}( f ) < |X_j^*|$ for each $j \in [k]$.
    Then by \cref{Thm:Nullstellensatz}, there exists a $v^* \in X_1^* \times \cdots \times X_k^*$ such that $f( v^* ) \ne 0$.
    Then there exists some $v \in X_1 \times \cdots \times X_k$ such that $v_j^* = \exp( -iv_j/2 )$ for each $j \in [k]$.
    Since $f$ satisfies (D3) with respect to $G$, $H$, and $( s, t )$, then $( \interp{G} - \interp{H} )_{s,t}( v ) \ne 0$.
    Then $(\interp{G} - \interp{H})( v ) \ne 0$.
    Then $\interp{G}( v ) \ne \interp{H}( v )$.
    Then the assumption $\interp{G} \ne \interp{H}$ implies that there exists a $v \in X_1 \cdots X_k$ such that $\interp{G}( v ) \ne \interp{H}( v )$.
    In conclusion, $\interp{G} = \interp{H}$ if and only if $\interp{G}( v ) = \interp{H}( v )$ for all $v \in X_1 \times \cdots \times X_k$.
\end{proof}

\probeq*
\begin{proof}
    Since $\interp{C_1} \ne \interp{C_2}$, then $\interp{C_1} - \interp{C_2} \ne 0$.
    Then there exists a pair of indices $j \in [2^n]$ and $\ell \in [2^m]$ such that $(\interp{C_1} - \interp{C_2})_{j,\ell} \ne 0$.
    Then by \cref{Cor:PolyDiff}, there exists a polynomial $f \in \mathbb{C}[x_1, \ldots, x_k]$ such that (D1) through to (D3) hold with respect to $C_1$, $C_2$, and $( j, \ell )$.
    Since $f$ satisfies (D2) with respect to $C_1$, $C_2$, and $( j, \ell )$, then $\deg( f ) \le d$.
    Since $x \mapsto \exp( -ix/2 )$ is a bijection on $[0, 2\pi)$ and bijections preserve discrete uniform distributions, then $\exp( -is_1/2 ), \ldots, \exp( -is_k/2 )$ are independent random variables selected uniformly from,
    {\small\begin{equation*}
        S^* = \{ \exp( -is/2 ) \mid s \in S \},
    \end{equation*}\par}\noindent%
    with $|S^*| = |S|$.
    Let $E_1$ denote the event $f( \exp( -is_1/2 ), \ldots, \exp( -is_k/2 ) ) = 0$ and $E_2$ denote the event $\interp{C_1}( s_1, \ldots, s_k ) = \interp{C_2}( s_1, \ldots, s_k )$.
    If $E_2$ occurs, then in particular $( \interp{C_1}( s_1, \ldots, s_k ) )_{j,\ell} = ( \interp{C_2}( s_1, \ldots, s_k ) )_{j,\ell}$.
    Since $f$ satisfies (D3) with respect to $C_1$, $C_2$, and $( j, \ell )$, then $f( \exp( -is_1/2 ), \ldots, \exp( -is_k/2 ) ) = 0$.
    Then $E_2 \subseteq E_1$.
    Then by the monotonicity of probability, $\Pr( E_ 2 ) \le \Pr( E_1 )$.
    Then, it suffices to show that $\Pr( E_1 ) \le d / |S|$.
    Since $f$ satisfies (D3) with respect to $C_1$, $C_2$, and $( j, \ell )$, and since $\interp{C_1}_{j,\ell} \ne \interp{C_2}_{j,\ell}$, then $f \ne 0$.
    Then by \cref{Thm:DLSZ}, $\Pr( E_1 ) ) \le d / |X|$.
    Therefore, $\Pr( \interp{C_1}( s_1, \ldots, s_k ) = \interp{C_2}( s_1, \ldots, s_k) ) \le d / |S|$.
\end{proof}

\subsection{Decidability for Integral Cyclotomic Circuits}

This section proves that parameterized equivalence checking is decidable for integral circuits with $\mathcal{G}$ and $\mathcal{H}$ consisting of cyclotomic circuits.
This follows from the fact that evaluating such a circuit at a rational multiple of $\pi$ yields a cyclotomic matrix.

\begin{theorem}
    \label{Thm:Cyclo}
    If $\mathcal{G}$ and $\mathcal{H}$ consists of matrices over the universal cyclotomic field and $C_1 \in \Circ( \mathcal{G}, \mathcal{H} )$, then $\interp{C_1}( \theta )$ is a matrix over the universal cyclotomic field for each $\theta \in ( \mathbb{Q} \pi )^k$.
\end{theorem}

\begin{proof}
    Let $\zpred( - )$ denote the predicate on $\ZCirc( \mathcal{G}, \mathcal{H} )$ such that $\zpred( C )$ if and only if $\interp{C}( \theta )$ is a matrix over the universal cyclotomic field for each $\theta \in ( \mathbb{Q} \pi )^k$.
    First, the claim is proven for singleton circuits using~\cref{Prop:GateInd}.
    \begin{itemize}
    \item \textbf{Base Case (1)}.
          Let $G \in \mathcal{G}$.
          Let $\theta \in ( \mathcal{Q} \pi )^k$.
          Then $\interp{G}( \theta ) = G$ with $G$ a matrix over the universal cyclotomic field by assumption.
          Since $\theta$ was arbitrary, then $\zpred( G )$.
    \item \textbf{Base Case (2)}.
          Let $M \in \mathcal{H}$ and $p \in \mathcal{F}$.
          Let $\theta \in ( \mathcal{Q} \pi )^k$.
          Since $p( \theta ) = \alpha_1 \theta_1 + \cdots \alpha_k \theta_k + q \pi$ for some $\alpha \in \mathbb{Z}^k$ and $q \in \mathbb{Q}$, then $p( \theta ) \in \mathbb{Q} \pi$.
          Since $p( \theta )$ is a rational multiple of $\pi$, then $\sin( \theta )$ and $\cos( \theta )$ are cyclotomic numbers.
          Recall that $M$ and $I$ are matrices over the universal cyclotomic field.
          Since matrices rings are closed under additional and scalar multiplication, then $\interp{R_M( p )} = \cos( p( \theta ) ) I + i \sin( p( \theta ) ) M$ is a matrix over the universal cyclotomic field.
          Since $\theta$ was arbitrary, then $\zpred( R_M( p ) )$.
    \item \textbf{Control Induction}.
          Let $G \in \Sigma( \mathcal{G}, \mathcal{H} )$.
          Assume that $G$ is a unitary gate on $n$ wires and that $\zpred( G )$ holds.
          Let $\theta \in ( \mathcal{Q} \pi )^k$.
          Since $\zpred( G )$ holds, then $\interp{G}( \theta )$ is a matrix over the universal cyclotomic field.
          Since $I_{2^n}$ is a matrix over the universal cyclotomic field, and the direct sum of two matrices over the same field yield a matrix over the same field, then $\interp{C( G )} = I_{2^n} \oplus \interp{G}$ is a matrix over the universal cyclotomic field.
          Since $\theta$ was arbitrary, then $\zpred( C( p ) )$.
    \end{itemize}
    Then by \cref{Prop:GateInd}, $\zpred( G )$ for all $G \in \Sigma( \mathcal{G}, \mathcal{H} )$.
    Next, \cref{Prop:CircInd} is used to prove the claim for all circuits.
    \begin{itemize}
    \item \textbf{Base Case (1)}.
          Let $\theta \in ( \mathcal{Q} \pi )^k$.
          Since the identity matrix is  amatrix over the universal cyclotomic field, then $\interp{\epsilon}( \theta ) = I_2$ is a matrix over the universal cyclotomic field.
          Since $\theta$ was arbitrary, then $G( \epsilon )$.
    \item \textbf{Base Case (2)}.
          If $C \in \Sigma( \mathcal{G}, \mathcal{H} )$, then $\zpred( C )$ holds by the first sub-proof.
    \item \textbf{Parallel Induction}.
          Let $G \in \Sigma( \mathcal{G}, \mathcal{H} )$ and $H \in \Sigma( \mathcal{G}, \mathcal{H} )$.
          Assume that $\zpred( C_1 )$ and $\zpred( C_2 )$ holds.
          Let $\theta \in ( \mathcal{Q} \pi )^k$.
          Since $\zpred( C_1 )$ holds, then $\interp{C_1}( \theta )$ is a matrix over the universal cyclotomic field.
          Since $\zpred( C_2 )$ holds, then $\interp{C_2}( \theta )$ is a matrix over the universal cyclotomic field.
          Since the tensor produce of two matrices over the same field yield a matrix over the same field, then $\interp{C_1 // C_2}( \theta ) = \interp{C_1}( \theta ) \otimes \interp{C_2}( \theta )$ is a matrix over the universal cyclotomic field.
          Since $\theta$ was arbitrary, then $\zpred( C_1 // C_2 )$ holds.
    \item \textbf{Sequential Induction}.
          Let $C_1 \in \Sigma( \mathcal{G}, \mathcal{H} )$ and $C_2 \in \Sigma( \mathcal{G}, \mathcal{H} )$ with $C_1$ and $C_2$ composable.
          Assume that $\zpred( C_1 )$ and $\zpred( C_2 )$ holds.
          Let $\theta \in ( \mathcal{Q} \pi )^k$.
          Since $\zpred( C_1 )$ holds, then $\interp{C_1}( \theta )$ is injective.
          Since $\zpred( C_2 )$ holds, then $\interp{C_2}( \theta )$ is injective.
          Since the tensor produce of two matrices over the same field yield a matrix over the same field, then $\interp{C_1 \circ C_2}( \theta ) = \interp{C_1}( \theta )  \interp{C_2}( \theta )$ is a matrix over the universal cyclotomic field.
          Since $\theta$ was arbitrary, then $\zpred( C_1 \circ C_2 )$ holds.
    \end{itemize}
    Then by \cref{Prop:CircInd}, $\zpred( C )$ for all $C \in \Sigma( \mathcal{G}, \mathcal{H} )$.
\end{proof}

\decidable*
\begin{proof}
    By \cref{Thm:QuantElim}, parameterized equivalence can be decided by comparing matrices obtained from each parameter in $S_1 \times S_2 \times \cdots \times S_k$.
    If $S_1 \times S_2 \times \cdots \times S_k$ consists of rational multiples of $\pi$, then by \cref{Thm:Cyclo}, each matrix will be over the universal cyclotomic field.
    Since the universal cyclotomic field is computable and has decidable equality, then this gives a decision procedure for the parameterized equivalence checking problem.
\end{proof}

\section{Circuit Reparameterization}
\label{Append:Reparam}

The section establishes all theorems in \cref{Sect:Implementation:Reparam}.
\cref{Lemma:AlphaTx} is introduced to relate the premonoidal structure of the parameter sequence to the syntactic transformation.

\biject*
\begin{proof}
    This proof has two directions.
    \begin{itemize}
    \item ($\Rightarrow$).
          Assume that $\interp{C_1} = \interp{C_2}$.
          Then $\interp{C_1} \circ f = \interp{C_2} \circ f$.
    \item ($\Leftarrow$).
          Assume that $\interp{C_1} \circ f = \interp{C_2} \circ f$.
          Let $x \in \mathbb{R}^k$.
          Evaluating $\interp{C_1} \circ f$ and $\interp{C_2} \circ f$ at $f^{-1}( x )$, it follows that $\interp{C_1}( x ) = ( \interp{C_1} \circ f )( f^{-1}( x ) ) = ( \interp{C_2} \circ f )( f^{-1}( x ) ) = \interp{C_2}( x )$.
          Since $x$ was arbitrary, then $\interp{C_1} = \interp{C_2}$.
    \end{itemize}
    In conclusion, $\interp{C_1} = \interp{C_2}$ if and only if $\interp{C_1} \circ f = \interp{C_2} \circ f$.
\end{proof}

\synparam*
\begin{proof}
    Let $g( \theta ) = ( \theta_1 / v_1, \theta_2 / v_2, \ldots, \theta_k / v_k )$.
    This is well-defined, since $v_j \ne 0$ for each $j \in [k]$.
    Clearly $f(g(\theta)) = \theta = g(f(\theta))$.
    Then $f$ is a bijection.
    It remains to be shown that $F_v$ is a syntactic reparameterization with respect to $f$.
    This follows by induction on the structure of $\Circ( \mathcal{G}, \mathcal{H} )$.
    First, the claim is proven for singleton circuits using~\cref{Prop:GateInd}.
    \begin{itemize}
    \item \textbf{Base Case (1)}.
          Let $G \in \mathcal{G}$.
          Then $F_v( G ) = G$ and $\interp{G}( \theta ) = G$ by definition.
          Let $x \in \mathbb{R}^k$.
          Then $\interp{F_v(G)}( x ) = \interp{G}( x ) = G = \interp{G}( f( x ) ) = (\interp{G} \circ f)( x )$.
          Since $x$ was arbitrary, then $\interp{F_v(G)} = \interp{G} \circ f$.
    \item \textbf{Base Case (2)}.
          Let $M \in \mathcal{H}$ and $p \in \mathcal{F}$.
          Then there exists coefficients $a \in \mathbb{Q}^k$ and $r \in \mathbb{R}$ such that $p(\theta) = a_1 \theta_1 + \cdots + a_k \theta_k + r$.
          Then $F_v( R_M( p ) ) = R_M( q )$ where $q( \theta ) = (v_1 a_1) \theta_1 + \cdots + (v_k a_k) \theta_k + r$.
          Then the following equation holds.
          {\small\begin{equation*}
              p( f( \theta ) )
              =
              a_1 (v_1 \theta_1) + \cdots + a_k (v_k \theta_k) + r
              =
              v_1 (a_1 \theta_1) + \cdots + v_k (a_k \theta_k) + r
              =
              q( \theta )
          \end{equation*}\par}\noindent%
          Then $\interp{F_v( R_M( p ) )} = \interp{R_M( q )} = \cos( -p( f( \theta ) / 2 ) I + i \sin( -p( f( \theta ) ) / 2 ) M = \interp{G} \circ f$.
    \item \textbf{Control Induction}.
          Let $G \in \Sigma( \mathcal{G}, \mathcal{H} )$.
          Assume that $G$ is unitary and $\interp{F_v(G)} = \interp{G} \circ f$. Then $\interp{F_v(C(G))} = \interp{C(F_v(G))} = I \oplus \interp{F_v(G)} = I \oplus (\interp{G} \circ f) = ( I \oplus \interp{G} ) \circ f = \interp{C(G)} \circ f$.
    \end{itemize}
    Then by~\cref{Prop:GateInd}, $\interp{F_v( G )} = \interp{G} \circ f$, for all $G \in \Sigma( \mathcal{G}, \mathcal{H} )$.
    Next, \cref{Prop:CircInd} is used to prove the claim for all circuits.
    \begin{itemize}
    \item \textbf{Base Case (1)}.
          By definition of $\epsilon$,  $F_v( \epsilon ) = \epsilon$ and $\interp{\epsilon}( \theta ) = I_2$.
          Let $x \in \mathbb{R}^k$.
          Then $\interp{\epsilon}( x ) = I = \interp{\epsilon}( f( x ) ) = ( \interp{\epsilon} \circ f )( x )$.
          Since $x$ was arbitrary, then $\interp{F_v( G )} = \interp{G} \circ f$.
    \item \textbf{Base Case (2)}.
          If $C \in \Sigma( \mathcal{G}, \mathcal{H} )$, then $\interp{F_v( C ) } = \interp{C} \circ f$ by the first sub-proof.
    \item \textbf{Parallel Induction}.
          Let $C_1 \in \Circ( \mathcal{G}, \mathcal{H} )$ and $C_2 \in \Circ( \mathcal{G}, \mathcal{H} )$.
          Assume that $\interp{F_v( C_1 )} = \interp{C_1} \circ f$ and $\interp{F_v( C_2 )} = \interp{C_2} \circ f$.
          Then for each $\theta \in \mathbb{R}^k$, the following equation holds
          {\small\begin{equation*}
              \interp{F_v( C_1 )}( \theta ) \otimes \interp{F_v( C_2 )}( \theta )
              =
              \interp{C_1}( f( \theta ) ) \otimes \interp{C_2}( f( \theta ) )
              =
              ( \interp{C_1} \otimes \interp{C_2} )( f( \theta ) )
          \end{equation*}\par}\noindent%
          Since $\theta$ was arbitrary, then $\interp{F_v( C_1 )} \otimes \interp{F_v( C_2 )} = (\interp{C_1} \otimes \interp{C_2}) \circ f$ is true for all $\theta \in \mathbb{R}^k$.
          Then the following equation holds.
          {\small\begin{equation*}
              \interp{F_v( C_1 // C_2 )}
              =
              \interp{F_v( C_1 ) // F_v( C_2 )}
              =
              \interp{F_v( C_1 )} \otimes \interp{F_v( C_2 )}
              =
              (\interp{C_1} \otimes \interp{C_2}) \circ f
              =
              \interp{C_1 // C_2} \circ f
          \end{equation*}\par}\noindent%
    \item \textbf{Sequential Induction}.
          Let $C_1 \in \Circ( \mathcal{G}, \mathcal{H} )$ and $C_2 \in \Circ( \mathcal{G}, \mathcal{H} )$ with $\wout( C_1 ) = \win( C_2 )$.
          Assume that $\interp{F_v( C_1 )} = \interp{C_1} \circ f$ and $\interp{F_v( C_2 )} = \interp{C_2} \circ f$.
          Then for each $\theta \in \mathbb{R}^k$, the following equation holds.
          {\small\begin{equation*}
              \interp{F_v( C_2 )}( \theta ) \interp{F_v( C_1 )}( \theta )
              =
              \interp{C_2}( f( \theta ) ) \interp{C_1}( f( \theta ) )
              =
              ( \interp{C_2} \interp{C_1} )( f( \theta ) )
          \end{equation*}\par}\noindent%
          Since $\theta$ was arbitrary, then $\interp{F_v( C_2 )} \interp{F_v( C_1 )} = (\interp{C_12} \interp{C_1}) \circ f$ is true for all $\theta \in \mathbb{R}^k$.
          Then the following equation holds.
          {\small\begin{equation*}
              \interp{F_v( C_2 \circ C_1 )}
              =
              \interp{F_v( C_2 ) \circ F_v( C_1 )}
              =
              \interp{F_v( C_2 )} \interp{F_v( C_1 )}
              =
              (\interp{C_2} \interp{C_1}) \circ f
              =
              \interp{C_2 \circ C_1} \circ f
          \end{equation*}\par}\noindent%
    \end{itemize}
    Then by~\cref{Prop:CircInd}, $\interp{F_v( C )} = \interp{C} \circ f$ for all $C \in \Circ( \mathcal{G}, \mathcal{H} )$.
    Therefore, $F_v$ is a syntactic reparameterization with respect to $f$.
\end{proof}

\begin{lemma}
    \label{Lemma:AlphaTx}
    Let $C \in \Circ( \mathcal{C}, \mathcal{H} )$ and $v \in ( \mathbb{Q}^{\times } )^k$.
    Then $|A( C ) | = |A( F_v( C )|$.
    Moreover, if $n = |A( C )|$, $j \in [n]$ and $\ell \in [k]$, then $( A( F_v( C ) )_j )_\ell = v_\ell ( A ( C )_j )_\ell$.
\end{lemma}

\begin{proof}
    Let $\zpred( - )$ denote the predicate on $\Circ( \mathcal{G}, \mathcal{H} )$ such that $\zpred( C )$ if and only if the following properties hold.
    \begin{itemize}
        \item (R1). $|A( C )| = |A( F_v( C )|$.
        \item (R2). If $n = |A( C )|$, $j \in [n]$ and $\ell \in [k]$, then $( A( F_v( C ) )_j )_\ell = v_\ell ( A ( C )_j )_\ell$.
    \end{itemize}
    The proof follows by induction on the structure of $\Circ( \mathcal{G}, \mathcal{H} )$.
    First, the claim is proven for singleton circuits using~\cref{Prop:GateInd}.
    \begin{itemize}
    \item \textbf{Base Case (1)}.
          Let $G \in \mathcal{G}$.
          Since $F_v( G ) = G$, then $|A( F_v( G ) )| = |A( G )|$.
          Then $G$ satisfies (R1).
          Since $|A( G )| = 0$, then (R2) is vacuously true for $G$.
          Then $\zpred( G )$ holds.
    \item \textbf{Base Case (2)}.
          Let $M \in \mathcal{H}$ and $p \in \mathcal{F}$.
          Then there exists coefficients $a \in \mathbb{Q}^k$ and $r \in \mathbb{Q}$ such that $p( \theta ) = a_1 \theta_1 + a_2 \theta_2 + \cdots + a_k \theta_k + r$.
          Then $F_v( R_M( p ) ) = R_M( q )$ where $q( \theta ) = (v_1 a_1) \theta_1 + (v_2 a_2) \theta_2 + \cdots + (v_k a_k) \theta_k + r$.
          Then $|A( R_M( p ) )| = 1 = |A( R_M( q ) )|$ and $R_M( p )$ satisfies (R1).
          Next, let $n = |A( G )|$, $j \in [n]$, and $\ell \in [k]$.
          Since $n = 1$, then $j = 1$ and the following equation holds.
          {\small\begin{equation*}
            ( A( F_v( R_M( p ) ) )_j )_\ell
            =
            ( A( R_M( p ) )_1 )_\ell
            =
            v_\ell a_\ell
            =
            v_\ell ( A( R_M( p ) )_1 )_\ell
            =
            v_\ell ( A( R_M( p ) )_j )_\ell
          \end{equation*}\par}\noindent%
          Since $n = 1$, $j$ and $\ell$ were arbitrary, then $R_M( p )$ satisfies (R2).
          Then $\zpred( R_M( p ) )$ holds.
    \item \textbf{Control Induction}.
          Let $G \in \Sigma( \mathcal{G}, \mathcal{H} )$.
          Assume that $G$ is a unitary gate and that $\zpred( G )$ holds.
          Since $G$ satisfies (R1), then the following equation holds.
          {\small\begin{equation*}
            |A( F_v( C( G ) ) )|
            =
            |A( C( F_v( G ) ) )|
            =
            |A( F_v( G ) )|
            =
            |A( G )|
            =
            |A( C ( G ) )|
          \end{equation*}\par}\noindent%
          Then $C( G )$ satisfies (R1).
          Next, let $n = |A( G )|$, $j \in [n]$, and $\ell \in [k]$.
          Since $C( G )$ satisfies R(2), then the following equation holds.
          {\small\begin{equation*}
            (A(F_v(C(G)))_j)_\ell
            =
            (A(F_v(G))_j)_\ell
            =
            v_\ell ( A( G )_j )_\ell
            =
            v_\ell ( A( C( G ) )_j )_\ell.
          \end{equation*}\par}\noindent%
          Since $j$ and $k$ were arbitrary, then $C( G )$ satisfies (R2) as well.
          Then $\zpred( C( G ) )$ holds.
    \end{itemize}
    Then by~\cref{Prop:GateInd}, $\zpred( G )$, for all $G \in \Sigma( \mathcal{G}, \mathcal{H} )$.
    Next, \cref{Prop:CircInd} is used to prove the claim for all circuits.
    \begin{itemize}
    \item \textbf{Base Case (1)}.
          Since $F_v( \epsilon ) = \epsilon$, then $|A( F_v( \epsilon ) )| = |A( C )|$.
          Then $\epsilon$ satisfies (R1).
          Since $|A ( \epsilon )| = 0$, then (R2) is vacuously true for $\epsilon$.
          Then $\zpred( \epsilon )$ holds.
    \item \textbf{Base Case (2)}.
          If $C \in \Sigma( \mathcal{G}, \mathcal{H} )$, then $\zpred( C )$ holds by the first sub-proof.
    \item \textbf{Parallel Induction}.
          Let $C_1 \in \Circ( \mathcal{G}, \mathcal{H} )$ and $C_2 \in \Circ( \mathcal{G}, \mathcal{H} )$.
          Assume that $\zpred( C_1 )$ and $\zpred( C_2 )$ hold.
          Then by (R1), $|A( C_1 )| = |A( F_v( C_1 ) )|$ and $|A( C_2 )| = |A( F_v( C_2 ) )|$.
          Starting from $|A( C_1 // C_2 )|$,
          {\small\begin{equation*}
            |A( C_1 // C_2 )|
            =
            |A( C_1 ) \cdot A( C_2 )|
            =
            |A( C_1 )| + |A( C_2 )|
            =
            |A( F_v( C_1 ) )| + |A( F_v ( C_2 ) )|.
          \end{equation*}\par}\noindent%
          Likewise, starting from $|A( F_v( C_1 // C_2 ))| = |A( F_v( C_1 ) // F_v( C_2 ))|$,
          {\small\begin{equation*}
            |A( F_v( C_1 ) // F_v( C_2 ))|
            =
            |A( F_v( C_1 ) ) \cdot A( F_v( C_2 ) )|
            =
            |A( F_v( C_1 ) )| + |A( F_v ( C_2 ) )|.
          \end{equation*}\par}\noindent%
          Then by transitivity, $ |A( F_v( C_1 // C_2 ))| = |A( C_1 // C_2 )|$ and $C_1 // C_2$ satisfies (R1).
          Next, let $|A( C_1 // C_2 )| = n$, $j \in [n]$, and $\ell \in [k]$.
          There are two cases to consider.
          \begin{enumerate}
          \item Assume $j \le |A( C_1 )|$.
                Then by indexing, $A( C_1 // C_2 )_j = ( A( C_1 ) \cdot A( C_2 ) )_j = A( C_1 )_j$ and $A( F_v( C_1 // C_2 ) )_j = ( A( F_v( C_1 ) ) \cdot A( F_v( C_2 ) ) )_j = A( F_v( C_1 ) )_j$.
                Since $C_1$ satisfies (R2), then $A( F_v( C_1 ) )_j = v_\ell ( A( C_1 )_j )_\ell$.
                By equality, $( A( F_v( C_1 // C_2 ) )_j )_\ell = v_\ell ( A( C_1 // C_2 )_j )_\ell$.
          \item Assume that $j > |A( C_1 )|$ and define $i = j - |A( C_1 )|$.
                Then by indexing, $A( C_1 // C_2 )_j = A( C_2 )_i$ and $A( F_v( C_1 // C_2 ) )_j = A( C_2 )_i$.
                Since $|A( C_1 // C_2 )| = |A( C_1 )| + |A( C_2 )|$, then $1 \le i \le |A( C_2 )|$.
                Since $C_2$ satisfies (R2), then $A( F_v( C_2 ) )_j = v_\ell ( A( C_2 )_j )_\ell$.
                It follows by equality that $( A( F_v( C_1 // C_2 ) )_j )_\ell = v_\ell ( A( C_1 // C_2 )_j )_\ell$.
          \end{enumerate}
          In either case, $( A( F_v( C_1 // C_2 ) )_j )_\ell = v_\ell ( A( C1 // C_2 )_j )_\ell$.
          Since $j$ and $\ell$ were arbitrary, then $C_1 // C_2$ satisfies (R2).
          Then $\zpred( C_1 // C_2 )$ holds.
    \item \textbf{Sequential Induction}.
          This case follows by the same argument, with all occurrences of the $( // )$ connective replaced by $( \circ )$.
    \end{itemize}
    Then by~\cref{Prop:CircInd}, $\zpred( C )$ for all $C \in \Circ( \mathcal{G}, \mathcal{H} )$.
\end{proof}

\fracexact*
\begin{proof}
    Assume for the intent of contradiction that $F_v( C_1 ) \not \in \ZCirc( \mathcal{G}, \mathcal{H} )$.
    Then by the definition of $\ZCirc( \mathcal{G}, \mathcal{H} )$, $A( F_v( C_1 ) ) \not \in ( \mathbb{Z}^k )^*$.
    Let $n = |A( F_v( C_1 ) )|$.
    Then there exists a $j \in [n]$ such that $A( F_v( C_1 ) )_j \not \in \mathbb{Z}^k$.
    Let $\alpha = A( F_v( C_1 ) )_j$.
    Then there exists an $\ell \in [k]$ such that $\alpha_\ell \not \in \mathbb{Z}$.
    Then $\denom( \alpha_\ell ) \ne 1$.
    Let $d = \denom( \alpha_\ell )$ and $x$ be the numerator such that $\alpha_\ell = x / d$.
    Let $\beta = A( C_1 )_j$.
    Then $x / d = v_\ell \beta_\ell$ by \cref{Lemma:AlphaTx}.
    Define the set,
    {\small\begin{equation*}
        X_\ell = \{ \denom( \alpha_\ell ) : \alpha \in A( C_1 ) \cdot A( C_2 ) \}
    \end{equation*}\par}\noindent%
    Then by definition, $\denom( \beta_\ell ) \in X_\ell$ and $v_\ell = \lcm( X_\ell )$.
    Let $d' = \denom( \beta_\ell )$ and $y$ be the numerator such that $\beta_\ell = y / d'$.
    Since $v_\ell$ is the least common multiple of the elements in $X_\ell$ and $d' \in X_\ell$, then $d' \mid v_\ell$.
    Then there exists a quotient $q \in \mathbb{Z}$ such that $q d' = v_\ell$.
    Then $x / d = v_\ell \beta_\ell = v_\ell ( y / d' ) = ( q d' ) ( y / d' ) = qy \in \mathbb{Z}$.
    In other words, $d = 1$.
    However, $d \ne 1$ by assumption.
    Then by contradiction $F_v( C_1 ) \in \ZCirc( \mathcal{G}, \mathcal{H} )$.
    By a symmetric argument, $F_v( C_2 ) \in \ZCirc( \mathcal{G}, \mathcal{H} )$.
    In remains to be shown that $\interp{C_1} = \interp{C_2}$ if and only if $\interp{F_v( C_1 )} = \interp{F_v( C_2 )}$.
    By \cref{Thm:FracExact}, there exists a bijection $f$ such that $\interp{F_v( C_1 )} = \interp{F_v( C_2 )}$ if and only if $\interp{C_1} \circ f = \interp{C_2} \circ f$.
    By \cref{Lemma:Biject}, $\interp{C_1} \circ f = \interp{C_2} \circ f$ if and only if $\interp{C_1} = \interp{C_2}$.
    This completes the proof.
\end{proof}

\section{Global Affine Linear Phase Inference}
\label{Appendix:Phase}

This section considers the problem of determining parameterized equivalence up to affine rational linear global phase, under the assumptions of~\cref{Sect:Implementation:Param}.
These assumptions subsume circuits over the Clifford+$T$ gate set, with arbitrary Pauli-rotations and state preparation.
This generalizes the gate sets considered in prior work, and moreover, is the first result of decidability for this problem.

Assume that $C_1 \in \ZCirc( \mathcal{G}, \mathcal{H} )$ and $C_2 \in \ZCirc( \mathcal{G}, \mathcal{H} )$ are two circuits which differ by an affine rational linear global phase.
This means that there exists an $\alpha \in \mathbb{Q}^k$ and a $\beta \in \mathbb{R}$ such that $\interp{C_2}( \theta ) = e^{i(\alpha_1\theta_1 + \cdots + \alpha_k\theta_k + \beta)} \interp{C_1}( \theta )$ for each $\theta \in \mathbb{R}^k$.
The goal of this section is to find an algorithm which, given $C_1$ and $C_2$, can solve for $\alpha_1$ through to $\alpha_k$.
In the case where $C_1$ and $C_2$ are not equivalent up to affine linear global phase, the algorithm should still terminate, but may return anything, since affine linear global phase can be validated easily through circuit instrumentation (we will see that $\alpha_1$ through to $\alpha_k$ are integral).

One approach to this problem is to note that $x \to e^{ix\pi}$ is injective on any not strictly closed interval of length $2\pi$.
This means that if $e^{i\alpha_j\pi}$ could be isolated, then it would be possible to compute its inverse.
However, $\alpha_j$ can be arbitrarily large, so it is unclear on which interval this inverse should be computed.
An alternative approach is to find some $b_j \in \mathbb{N}$ satisfying $b_j > \alpha_j$, so that $\alpha_j / b_j \in (-\pi, \pi)$.
Then $(\alpha_j / b_j) \pi \in (-1, 1)$, and consequently, it would be possible to isolate $e^{i(\alpha_j/b_j)\pi}$.

This approach can be decomposed into three sub-problems.
First, it must be shown how to compute an $b_j \in \mathbb{N}$ such that $\alpha_j / b_j \in (-1, 1)$.
It will follow from this analysis that $\alpha_j \in \mathbb{Z}$.
Second, it must be shown how to isolate $z = e^{i(\alpha_j/b_j)\pi}$.
This can be done via arithmetic operations on the matrix entries of $\interp{C_1}( \theta )$ and $\interp{C_2}( \theta )$, where $\theta $ is some rational multiple of $\pi$.
Since parameterized circuits over the universal cyclotomic field stay within the cyclotomic field when evaluated at rational points, then $z$ must live within the universal cyclotomic field as well.
Moreover, this is computable.
Motivated by this observation we then show how to compute the inverse to $x \mapsto e^{ix\pi}$ for $x \in (-1, 1) \cap \mathbb{Q}$.

\subsection{Normalizing the Coefficients}

The goal of this section is to find some $b_j \in \mathbb{N}$ such that $\alpha_j / b_j \in (-1, 1)$.
As in the proof of~\cref{Thm:QuantElim}, the idea is to inspect the polynomial abstraction $\pinterp{-}$ from \cref{Sect:Equiv:Poly}, to restrict the set of possible values for $\alpha_j$.
It will first be shown that $\alpha_j$ must be an integer.
Using this information, the degree bound on $\pinterp{-}$ from~\cref{Thm:PolyMat} will be used to show that $|\alpha_j| < 2\lambda_j$.
Intuitively, the frequency of $\interp{-}$ in the $j$-th component bounds $\alpha_j$.

To this end, $\interp{-}$ must be characterized when restricted to change in the $j$-th component.
Let $e_\ell$ denote the $\ell$-th standard basis vector for $\mathbb{R}^k$ and $c_\ell$ denote the $\ell$-th standard basis vector for $\mathbb{C}^k$.
Fix some starting point $\widehat{\theta} \in \mathbb{R}^k$.
If the $j$-th parameter changes by $x \in \mathbb{R}$, then the new parameter will be $\widehat{\theta} + x e_j$.
It turns out that regardless of the starting point, the function $x \mapsto \interp{C}(\widehat{\theta} + x e_j)$ is periodic with period at most $4\pi$.
As a convenience of notation, define the following two maps.
\begin{itemize}
\item $\omega_1: \mathbb{R} \to \mathbb{R}^k$ such that $\omega_1( x ) = \widehat{\theta} + x e_j$.
\item $\omega_2: \mathbb{C} \to \mathbb{R}^k$ such that $\omega_2( z ) = \sum_{\ell \in [k] \setminus \{ j \}} \left( \exp( -i\widehat{\theta}_\ell/2 ) c_\ell \right) + z c_j$
\end{itemize}
Then $\interp{C}(\widehat{\theta} + x e_j) = \interp{C}( \omega_1( x ) )$.
Moreover, $\interp{C}( \omega_1( x ) ) = \pinterp{C}( \omega_2( \exp( -ix/2 ) ) )$ by \cref{Thm:PolyAbstract}.
It is easy to see that $\pinterp{C} \circ \omega_2$ is a matrix of single variable Laurent polynomials, since $\omega_2$ fixes all other components.
Using $\pinterp{C} \circ \omega_2$, it is straight-forward to prove that $\interp{C} \circ \omega_1$ is periodic.

\begin{lemma}
    \label{Lem:Periodicity}
    If $C \in \ZCirc( \mathcal{G}, \mathcal{H} )$, then for each $\widehat{\theta} \in \mathbb{R}^k$ and $j \in [j]$, the function $\interp{C} \circ \omega_1$ is periodic with period at most $4\pi$.
\end{lemma}

\begin{proof}
    Let $j \in [k]$.
    Then let $s \in [2^n]$ and $t \in [2^m]$ where $n = \win( C )$ and $m \in \wout( C )$.
    Define $F = \interp{C}_{s,t}$ and $f = \left( \pinterp{C} \right)_{s,t}$.
    Then by \cref{Thm:PolyAbstract}, $F( \omega_1( x ) ) = f( \omega_2( \exp( -x/2 ) ) )$.
    Let $x \in \mathbb{R}$ and $\ell \in \mathbb{N}$.
    Then the following equation holds.
    {\small\begin{equation*}
        F( x )
        =
        f( \exp( -ix/2 ) )
        =
        f( \exp( -ix/2 - 2\pi\ell ) )
        =
        f( \exp( -i(x+4\pi\ell)/2 ) )
        =
        F( x + 4\pi\ell )
    \end{equation*}\par}\noindent%
    Since $x$ and $\ell$ were arbitrary, then by definition, $F \circ \omega_1$ is periodic with period at most $4\pi$.
    Since $s$ and $t$ were arbitrary as well, then $\interp{C} \circ \omega_1$ is periodic with period at most $4\pi$.
    Since $j$ was arbitrary as well, then the proof is complete.
\end{proof}

In the case where $C_1$ and $C_2$ do differ by an affine linear global phase, it then follows that $e^{-ip(\theta)/2} \interp{C_1}( \theta )$ is a periodic function, since $\interp{C_2}( \theta )$ is a already known to be periodic.
Since $\interp{C_1}( \theta )$ is already known to be periodic, then this claim holds if and only if $\exp( -i\theta_j/2 )$ is periodic as well.
It is well known that $\exp( -i\theta_j/2 )$ is periodic if and only if $\alpha_j$ is rational.
Moreover, the period of $4\pi$ enforces that $\alpha_j$ is in fact an integer.

\begin{theorem}
    \label{Thm:AlphaBnd}
    If $C_1 \in \ZCirc( \mathcal{G}, \mathcal{H} )$ and $C_2 \in \ZCirc( \mathcal{G}, \mathcal{H} )$ satisfy $\interp{C_2}( \theta ) = e^{-ip(\theta)/2} \interp{C_1}( \theta )$ where $p(\theta) = \alpha_1 \theta_1 + \cdots + \alpha_k \theta_k + \beta$ for some $\alpha \in \mathbb{R}^k$ and $\beta \in \mathbb{R}$, then $\alpha \in \mathbb{Z}^k$.
\end{theorem}

\begin{proof}
    Since $\interp{C_1}$ and $\interp{C_2}$ are equal up to a scalar function, then $\interp{C_1}$ and $\interp{C_2}$ have the same dimension.
    If $\interp{C_1}$ is identically zero, then $\interp{C_1}(\theta) = \mathbf{0}$ and $\interp{C_2}(\theta) = e^{ip(\theta)/2} \interp{C_1} = \mathbf{0}$.
    Then, $\alpha_j = 0$ and $\beta = 0$ is a valid solution.
    Assume instead that $\interp{C_1}( \theta )$ is not identically zero.
    Then there exists some component $( s, t )$ such that $\interp{C_1}_{s,t}$ is not identically zero.
    Define the functions $F = \interp{C_1}_{s,t}$ and $G = \interp{C_2}_{s,t}$.
    Since $\interp{C_1}_{s,t}$ is not identically zero, then there exists some $\widehat{\theta} \in \mathbb{R}^k$ such that $F( \widehat{\theta} ) \ne 0$.
    Define $\gamma = \exp( -ip( \widehat{\theta} ) / 2 )$.
    By \cref{Lem:Periodicity}, both $F$ and $G$ have period at most $4\pi$.
    Then the following equation holds.
    {\small\begin{equation*}
        \gamma F( \omega_1( 0 ) )
        =
        G( \omega_1( 0 ) )
        =
        G( \omega_1( 4\pi ) )
        =
        \gamma \exp( -i(4\pi)\alpha_j/2 ) F( \omega_1( 4\pi ) )
        =
        \gamma \exp( -i(2\pi)\alpha_j ) F( \omega_1( 0 ) )
    \end{equation*}\par}\noindent%
    Since $\gamma$ is an exponential, then $\gamma \ne 0$.
    Moreover, $F( \omega_1( 0 ) ) = F( \widehat{\theta} ) \ne 0$ by assumption.
    Then $1 = \exp( -i(2\pi)\alpha_j )$.
    Then there exists some $\ell \in \mathbb{Z}$ such that $-(2\pi)\alpha_j = 2\pi \ell$.
    Then $\theta_j = -\ell \in \mathbb{Z}$.
    Since $j$ was arbitrary, then $\alpha \in \mathbb{Z}$.
\end{proof}

Since $\alpha_j \in \mathbb{Z}$, then scaling $\pinterp{C_1} \circ \omega_2$ by $e^{i\beta} ( z_j )^{\alpha_j}$ yields a new Laurent polynomial which agrees with $\interp{C_2} \circ \omega_1$ on the unit circle.
Since the number of points on the unit circle is infinite, then it follows from~\cref{Thm:Nullstellensatz} that these two polynomials are identically equal.
This means that the degree bounds for $\pinterp{C_1}$ are also bounds on $|\alpha_j|$.
It follows that $|\alpha_j| \le 2 \lambda_j$.
Consequently, $\alpha_j / b_j \in (-1, 1)$ when $b_j = 2\lambda_j$.

\begin{corollary}
    \label{Cor:AlphaBnd}
    Assume that $C_1 \in \ZCirc( \mathcal{G}, \mathcal{H} )$ and $C_1 \in \ZCirc( \mathcal{G}, \mathcal{H} )$ satisfy the equation $\interp{C_2}( \theta ) = e^{-ip(\theta)/2} \interp{C_1}( \theta )$ where $p(\theta) = \alpha_1 \theta_1 + \cdots + \alpha_k \theta_k + \beta$ for some $\alpha \in \mathbb{R}^k$ and $\beta \in \mathbb{R}$.
    Then $\alpha_j \in \mathbb{Z}$ and $|\alpha_j| \le 2\lambda_j$ for each $j \in [k]$.
\end{corollary}

\newcommand{\todo}[1]{{\color{red} #1}}
\begin{proof}
    Let $j \in [k]$.
    It follows from \cref{Thm:AlphaBnd} that $\alpha_j \in \mathbb{Z}$.
    To prove that $|\alpha_j| \le 2\lambda_j$, first define the functions $F = \interp{C_1}_{0,0}$, $f = \left( \pinterp{C_1} \right)_{0,0}$, $G = \interp{C_2}_{0,0}$, and $g = \left( \pinterp{C_1} \right)_{0,0}$.
    Then define the Laurent polynomial $h( x ) = \gamma x^{\alpha_j} f( \omega_2( x ) )$ where $\gamma = e^{-i\beta/2}$.
    Pick any $z \in \mathbb{C}$ such that $|z| = 1$.
    Since $z$ is on the complex unit circle, then there exists some $\rho \in \mathbb{R}$ such that $z = \exp( - i\rho / 2 )$.
    Then by \cref{Thm:PolyAbstract}, $F( \omega_1( \rho ) ) = f( \omega_2( \exp( -i\rho / 2 ) ) ) = f( \omega_2( z ) )$ where $\widehat{\theta} = 0$.
    Then the following equation also holds.
    {\small\begin{equation*}
        h( z )
        =
        \gamma z^{\alpha_j} f( \omega_2( z ) )
        =
        \gamma e^{-i\alpha_j \rho / 2} f( \omega_2( z ) )
        =
        e^{-ip( \omega_1( \rho ) ) / 2} f( \omega_2( z ) )
        =
        e^{-ip( \omega_1( \rho ) ) / 2} F( \omega_1( \rho ) )
    \end{equation*}\par}\noindent%
    Then by \cref{Thm:PolyAbstract},  $G( \omega_2( x ) ) = g( \omega_2( \exp( -i\rho / 2 ) ) ) = g( \omega_2( z ) )$ where $\widehat{\theta} = 0$.
    This combines with the assumption that $\interp{C_2}( \theta ) = e^{ip(\theta)/2} \interp{C_1}( \theta )$ to obtain the following equation.
    {\small\begin{equation*}
        h( z )
        =
        e^{-ip( \omega_1( \rho ) ) } F( \omega_1( \rho ) )
        =
        G( \omega_1( \rho ) )
        =
        g( \omega_2( z ) )
    \end{equation*}\par}\noindent%
    Since $z$ was arbitrary, then $h$ and $g \circ \omega_2$ agree on the complex unit circle.
    Recall from \cref{Cor:PolyDiff} that $\lambda_j = \max \{ \sum_{\alpha \in A(C)} |\alpha_j| : C \in \{ C_1, C_2 \} \}$.
    Then by \cref{Thm:PolyMat}, $\deg_{z_j}^{\pm}( f ) \le \lambda_j$ and $\deg_{z_j}^{\pm}( g ) \le \lambda_j$.
    Since substituting variables for constants can only decrease the degree of a polynomial, then $\deg^{\pm}( f \circ \omega_2 ) \le \lambda_j$ and $\deg^{\pm}( g \circ \omega_2 ) \le \lambda_j$ as well.
    Then the polynomials $x^{\lambda_j} h( x )$ and $x^{\lambda_j} g( \omega_1( x ) )$ have strictly positive degrees and agree on the complex unit circle.
    In other words, every point on the complex unit circle is a root of $x^{\lambda_j}( h( x ) - g( \omega_1( x ) ) )$.
    Since the number of points on the complex unit circle is uncountable, then it follows trivially by~\cref{Thm:Nullstellensatz} that $x^{\lambda_j}( h( x ) - g( \gamma_1( x ) ) )$ is identically zero.
    Since $x^{\lambda_j}$ is not identically zero, then $h( x ) = g( \gamma_1( x ) )$.
    There are for cases to consider.
    \begin{itemize}
    \item Assume that $\alpha_j \ge 0$ and $\deg^{+}( f \circ \omega_2 ) = 0$.
          Then $\lambda_j \ge \deg^{+}( g \circ \omega_1 ) = \alpha_j - \deg^{-}(f \circ \omega_2 ) \ge \alpha_j - \lambda_j$.
          Then $2 \lambda_j \ge \alpha_j = |\alpha_j|$.
    \item Assume that $\alpha_j \ge 0$ and $\deg^{+}( f \circ \omega_2 ) > 0$.
          Then $\lambda_j \ge \deg^{+}( g \circ \omega_1 ) = \alpha_j + \deg^{+}(f \circ \omega_2 ) \ge \alpha_j$.
          Then $\lambda_j \ge \alpha_j = |\alpha_j|$.
    \item Assume that $\alpha_j < 0$ and $\deg^{-}( f \circ \omega_2 ) = 0$.
          Then $\lambda_j \ge \deg^{-}( g \circ \omega_1 ) = -\alpha_j - \deg^{+}(f \circ \omega_2 ) \ge -\alpha_j - \lambda_j$.
          Then $2 \lambda_j \ge -\alpha_j = |\alpha_j|$.
    \item Assume that $\alpha_j < 0$ and $\deg^{-}( f \circ \omega_2 ) > 0$.
          Then $\lambda_j \ge \deg^{-}( g \circ \omega_1 ) = -\alpha_j + \deg^{-}(f \circ \omega_2 ) \ge -\alpha_j$.
          Then $\lambda_j \ge -\alpha_j = |\alpha_j|$.
    \end{itemize}
     In each case, $|\alpha_j| \le 2 \lambda$.
     Since $j$ was arbitrary, then this completes the proof.
\end{proof}

\subsection{Isolating the Linear Phase Terms}

The goal of this section is to isolate the terms $e^{i\beta}$ and each $e^{i(\alpha_j/b_j)\pi}$.
This is relatively easy, since every matrix in $\mathcal{G}$ is injective.
It follows that $\interp{C_1}( \theta )$ is injective for any choice of $\theta$.
In particular, this means that $\interp{C_1}( 0 )$ will always have a non-zero component.

Let $\theta_0 = 0$ and $\theta_j = e_j (\pi/b_j)$ where $e_j$ is the $j$-th standard basis vector for $\mathbb{R}^k$.
Then there exists some $(s, t)$ such that $(\interp{C_1}( \theta_0 ))_{s,t} \ne 0$ and there exists some $(u, v)$ such that $\interp{C_1}( \theta_1 ))_{u,v} \ne 0$.
It then follows by direct computation that the following equations hold.
\begin{align*}
    e^{i\beta}
    &=
    \frac{(\interp{C_2}(\theta_0))_{s,t}}{(\interp{C_1}(\theta_0))_{s,t}}
    &
    e^{i(\alpha_j/b_j)\pi}
    &=
    \frac{(\interp{C_2}(\theta_j))_{u,v}}{e^{i\beta} (\interp{C_1}(\theta_j))_{u,v}}
\end{align*}
Since both $\theta_0$ and $\theta_j$ are rational multiples of $\pi$, then this can be computed exactly in the universal cyclotomic field.

\begin{theorem}
    \label{Thm:InjGate}
    If $\mathcal{G}$ consists of injective matrices and $C_1 \in \Circ( \mathcal{G}, \mathcal{H} )$, then $\interp{C_1}( \theta )$ is injective for each $\theta \in \mathbb{R}^k$.
\end{theorem}

\begin{proof}
    Let $\zpred( - )$ denote the predicate on $\ZCirc( \mathcal{G}, \mathcal{H} )$ such that $\zpred( C )$ if and only if $\interp{C}( \theta )$ is injective for all $\theta \in \mathcal{R}^k$.
    First, the claim is proven for singleton circuits using~\cref{Prop:GateInd}.
    \begin{itemize}
    \item \textbf{Base Case (1)}.
          Let $G \in \mathcal{G}$.
          Let $\theta \in \mathcal{R}^k$
          Then $\interp{G}( \theta ) = G$ with $G$ injective by assumption.
          Since $\theta$ was arbitrary, then $\zpred( G )$.
    \item \textbf{Base Case (2)}.
          Let $M \in \mathcal{H}$ and $p \in \mathcal{F}$.
          Let $\theta \in \mathbb{R}^k$.
          As explained in~\cref{Sect:Circuits}, $\interp{R_M( p )}( \theta ) = \cos( p ( \theta ) ) I  + i \sin( p( \theta ) ) M$ is unitary.
          Since unitary matrices are invertible, then they are injective.
          Then $\interp{G}( \theta )$ is injective.
          Since $\theta$ was arbitrary, then $\zpred( R_M( p ) )$.
    \item \textbf{Control Induction}.
          Let $G \in \Sigma( \mathcal{G}, \mathcal{H} )$.
          Assume that $G$ is a unitary gate on $n$ wires and that $\zpred( G )$ holds.
          Let $\theta \in \mathcal{R}^k$.
          Since $\zpred( G )$ holds, then $\interp{G}( \theta )$ is injective.
          Since $I_{2^n}$ is injective, and the direct sum of injective matrices must be injective, then $\interp{C( G )} = I_{2^n} \oplus \interp{G}$ is injective.
          Since $\theta$ was arbitrary, then $\zpred( C( p ) )$.
    \end{itemize}
    Then by \cref{Prop:GateInd}, $\zpred( G )$ for all $G \in \Sigma( \mathcal{G}, \mathcal{H} )$.
    Next, \cref{Prop:CircInd} is used to prove the claim for all circuits.
    \begin{itemize}
    \item \textbf{Base Case (1)}.
          Let $\theta \in \mathbb{R}^k$.
          Since the identity matrix is injective, then $\interp{\epsilon}( \theta ) = I_2$ is injective.
          Since $\theta$ was arbitrary, then $G( \epsilon )$.
    \item \textbf{Base Case (2)}.
          If $C \in \Sigma( \mathcal{G}, \mathcal{H} )$, then $\zpred( C )$ holds by the first sub-proof.
    \item \textbf{Parallel Induction}.
          Let $G \in \Sigma( \mathcal{G}, \mathcal{H} )$ and $H \in \Sigma( \mathcal{G}, \mathcal{H} )$.
          Assume that $\zpred( C_1 )$ and $\zpred( C_2 )$ holds.
          Let $\theta \in \mathcal{R}^k$.
          Since $\zpred( C_1 )$ holds, then $\interp{C_1}( \theta )$ is injective.
          Since $\zpred( C_2 )$ holds, then $\interp{C_2}( \theta )$ is injective.
          Since the tensor produce of injective matrices must be injective, then $\interp{C_1 // C_2}( \theta ) = \interp{C_1}( \theta ) \otimes \interp{C_2}( \theta )$ is injective.
          Since $\theta$ was arbitrary, then $\zpred( C_1 // C_2 )$ holds.
    \item \textbf{Sequential Induction}.
          Let $C_1 \in \Sigma( \mathcal{G}, \mathcal{H} )$ and $C_2 \in \Sigma( \mathcal{G}, \mathcal{H} )$ with $C_1$ and $C_2$ composable.
          Assume that $\zpred( C_1 )$ and $\zpred( C_2 )$ holds.
          Let $\theta \in \mathcal{R}^k$.
          Since $\zpred( C_1 )$ holds, then $\interp{C_1}( \theta )$ is injective.
          Since $\zpred( C_2 )$ holds, then $\interp{C_2}( \theta )$ is injective.
          Since the tensor produce of injective matrices must be injective, then $\interp{C_1 \circ C_2}( \theta ) = \interp{C_1}( \theta )  \interp{C_2}( \theta )$ is injective.
          Since $\theta$ was arbitrary, then $\zpred( C_1 \circ C_2 )$ holds.
    \end{itemize}
    Then by \cref{Prop:CircInd}, $\zpred( C )$ for all $C \in \Sigma( \mathcal{G}, \mathcal{H} )$.
\end{proof}

\subsection{Computing the Coefficients}

The goal of this section is to recover $\alpha_j$ from $z = e^{i(\alpha_j/b_j)\pi}$.
In general, global phase recovery can be hard, since the global phase need not be a root of unity.
Thanks to~\cref{Cor:AlphaBnd}, this is much easier for integral circuits.
If $d = 2 b_j$, then $z = e^{i(\alpha_j/b_j)\pi} = e^{i(\alpha_j/d)2\pi} = (\zeta_d)^{\alpha_j}$ where $\zeta_d$ is the primitive $d$-th root of unity.
This means that $z \in \mathbb{Q}[ \zeta_d ]$.
Moreover, since $d$ is even, then all roots of unity in $\mathbb{Q}[ \zeta_d ]$ are of degree at most $d$.
Since $\alpha_j / d \in [ -1/2, 1/2 ]$ with $\alpha_j \in \mathbb{Z}$ and $d = 4 \lambda$, then there must exist some $\ell \in \{ -2\lambda, -2\lambda + 1, \ldots, 2\lambda -1, 2\lambda \}$ such that $\alpha_j = \ell$.
To find such an $\ell$, it suffices to search all of the integers from $-2\lambda$ to $2\lambda$ until $( \zeta_d )^\ell = z$.
If no such integer exists, then $C_1$ and $C_2$ do not differ by a global phase.

\subsection{An Algorithmic Summary}

The previous analysis is summarized by the following algorithm, denoted $\FindPhase( C_1, C_2 )$.
\begin{enumerate}
\item Compute $M_1 = \interp{C_1}( \theta_0 )$ and $M_2 = \interp{C_2}( \theta_0 )$.
\item If there exists indices $( s, t )$ such that $( M_1 )_{s,t} \ne 0$ and $( M_2 )_{s,t} \ne 0$, then define the variable $z_0 = ( M_2 )_{s,t} / ( M_1 )_{s,t}$, else return the $( 1, 0 )$.
\item If $|z_0| \ne 1$, then return $( 1, 0 )$.
\item For each $j \in [k]$, compute $\lambda_j = \max\{ \sum_{a \in A( C )} |a_j| : C \in \{ C_1, C_2 \} \}$.
\item For each $j \in [k]$, compute $M_1^j = \interp{C_1}( \theta_j )$ and $M_2^j = \interp{C_2}( \theta_j )$, where $\theta_j = e_j \pi / ( \lambda_j + 1)$.
\item For each $j \in [k]$, if there exists indices $( u, v )$ such that $( M_1^j )_{u,v} \ne 0$ and $( M_2^j )_{u,v} \ne 0$, then define the variable $z_j = ( M_2^j )_{u,v} / ( z_0 ( M_1^j )_{u,v} )$, else return the $( 1, 0 )$.
\item For each $j \in [k]$, if there exists an $\ell \in \{ -2\lambda_j, -2\lambda_j + 1, \ldots, 2\lambda_j -1, 2\lambda_j \}$ such that $\zeta_{4\lambda_j + 4}^\ell = z_j$, then define the variable $\alpha_j = \ell$, else return the $( 1, 0 )$.
\item Return $( 1 / z_0, f )$ where $f( \theta ) = \alpha_1 \theta_1 + \cdots + \alpha_k \theta_k$.
\end{enumerate}
Note that this does not solve for $\beta$ explicitly.
In principle, $\beta$ could be any value.
However, $1 / z_0$ can be used in place of the $e^{i\beta}$ in all further analysis of the programs.

\begin{lemma}
    \label{Lem:FindPhase}
    If $( z, f ) = \FindPhase( C_1, C_2 )$, then there exists a $\beta \in \mathbb{R}$ such that $z = e^{i\beta}$.
\end{lemma}

\begin{proof}
    Note that there exists a $\beta \in \mathbb{R}$ such that $z = e^{i\beta}$ if and only if $|z| = 1$.
    Then, it suffices to show that $|z| = 1$.
    Note that the algorithm $\FindPhase( -, - )$ will return in one of five possible cases (lines 2, 3, 6, 7, and 8).
    In the first four cases, the value of $z$ is set to $1$, which implies that $|z| = 1$.
    In the fifth case, the value of $z$ is set to $z_0$.
    To reach the fifth return case at line 8, the return statement at line 3 must be bypassed.
    To bypass the return statement at line 3, it must be the case that $|z_0| = 1$, which implies that $|z| = 1$.
    Therefore, $|z| = 1$ in each of the five possible return cases.
\end{proof}

\begin{theorem}
    \label{Thm:FindPhase}
    Assume $\mathcal{G}$ and $\mathcal{H}$ consist of matrices over the universal cyclotomic field, with all gates in $\mathcal{G}$ injective.
    If $C_1 \in \ZCirc( \mathcal{G}, \mathcal{H} )$ is equivalent to $C_2 \in \ZCirc( \mathcal{G}, \mathcal{H} )$ modulo affine rational linear global phase and $( z, f ) = \FindPhase( C_1, C_2 )$, then $\interp{C_1} = z \left( e^{-if( \theta ) / 2} \interp{C_2} \right)$.
\end{theorem}

\begin{proof}
    Since $C_1$ is equivalent to $C_2$ modulo affine linear global phase, then there exists an affine rational linear function $p( \theta ) = x_1 \theta_1 + \cdots + x_k \theta_k + \beta$ such that $\interp{C_1}( \theta ) = e^{ip(\theta)/2} \interp{C_2}( \theta )$ for all $\theta \in \mathbb{R}^k$.
    Since all gates in $\mathcal{G}$ are injective, then by \cref{Thm:InjGate}, the matrix $\interp{C_2}( \theta )$ is also injective for all angles $\theta \in \mathbb{R}^k$.
    In particular, $M_2 = \interp{C_2}( \theta_0 )$ is injective and the following equation holds.
    {\small\begin{equation*}
        M_1
        =
        \interp{C_1}( \theta_0 )
        =
        e^{ip(\theta_0) / 2} \interp{C_2}( \theta_0 )
        =
        e^{i\beta / 2} \interp{C_2}( \theta_0 )
        =
        e^{i\beta / 2} M_2
    \end{equation*}\par}\noindent%
    Since $M_2$ is injective, then there exists indices $( s, t )$ such that $( M_2 )_{s,t} \ne 0 $.
    Since $e^{i\beta / 2} \ne 0$, then $( M_1 )_{s,t} = e^{i\beta/2} ( M_2 )_{s,t} \ne 0$.
    It follows that $z_0 = ( M_2 )_{s,t} / ( M_1 )_{s,t} = e^{-i\beta / 2}$.
    Next, let $j \in [k]$.
    Since $\interp{C_2}( \theta )$ is also injective for all angles $\theta \in \mathbb{R}^k$, then in particular, $M_2^j = \interp{C_2}( \theta_j )$ is injective.
    Moreover, since $p( \theta_j ) = x_j \pi / ( \lambda_j + 1 ) + \beta$, then the following equation holds.
    {\small\begin{equation*}
        M_1^j
        =
        \interp{C_1}( \theta_j )
        =
        e^{ip(\theta_j) / 2} \interp{C_2}( \theta_j )
        =
        \left( e^{i x_j \pi / ( 2\lambda_j + 2 )} \interp{C_2}( \theta_j ) \right) / z_0
        =
        \left( e^{i x_j \pi / ( 2\lambda_j + 2 )} M_2 \right) / z_0
    \end{equation*}\par}\noindent%
    Since $M_2^j$ is injective, then there exists some indices $( u, v )$ such that $( M_2^j )_{u,v} \ne 0$.
    Since $z_0 \ne 0$ and $e^{i x_j \pi / ( 2\lambda_j + 2 )} \ne 0$, then $( M_1^j )_{u,v} = z_0 e^{i x_j \pi / ( 2\lambda_j + 2 )} ( M_2^j )_{u,v} \ne 0$.
    It follows that $z_j = ( M_2^j )_{u,v} / ( z_0 M_1^j )_{u,v} = e^{-ix_j \pi / ( 2\lambda_j + 2 )}$.
    By~\cref{Cor:AlphaBnd}, $x_j \in \mathbb{Z}$ and $|x_j| \le 2\lambda_j$.
    It follows that $x_j \in \{ -2 \lambda_j, -2 \lambda_j + 1, \ldots, 2 \lambda_j - 1, 2 \lambda_j \}$.
    It must now be shown that there exists a unique value $\ell \in \{ -2 \lambda_j, -2 \lambda_j + 1, \ldots, 2 \lambda_j - 1, 2 \lambda_j \}$ such that $\zeta_{(4\lambda_j + 4)}^\ell = \left( e^{i\pi/(2\lambda_j + 2)} \right)^\ell = z_j$.
    If such an $\ell$ did exist, then it must be unique, because $\ell \mapsto \zeta_{(4\lambda_j + 4)}^\ell$ is bijective on $[ -2\lambda_j, 2\lambda_j ]$.
    It remains to be shown that $\ell$ exists.
    If $\ell = -x_j$, then $\zeta_{(4\lambda_j + 4)}^\ell = \left( e^{i\pi/(2\lambda_j + 2)} \right)^{-x_j} = z_j$.
    Consequently, $\ell$ exists and $\alpha_j = -x_j$.
    Since $j$ was arbitrary, then $\FindPhase( C_1, C_2 ) = ( 1 / z_0, f )$ where $f( \theta ) = \alpha_1 \theta_1 + \cdots + \alpha_k \theta_k$.
    Since $z e^{-if(\theta) / 2} = e^{-i(\alpha_1 \theta_1 + \cdots + \alpha_k \theta_k - \beta) / 2} = e^{ip(\theta)/2}$, then in conclusion $\interp{C_1} = z \left( e^{-if( \theta ) / 2} \interp{C_2} \right)$.
\end{proof}

\subsection{Proof of Theorem~\ref{Thm:GPhase}}

\gphase*
\begin{proof}
    By definition, $\interp{R_I( f )}( \theta ) = \cos( -f( \theta ) / 2 ) I + i\sin( -f( \theta ) / 2 ) I = e^{-if( \theta ) / 2} I$.
    There are two cases to consider.
    \begin{itemize}
    \item Assume that $C_1$ is equivalent to $C_2$ modulo affine rational linear global phase.
          Then by \cref{Thm:FindPhase}, $\interp{C_1} = z \left( e^{-if( \theta ) / 2 } \interp{C_2} \right) = \interp{z I} \interp{R_I( f )} \interp{C_2} = \interp{z I \circ R_I( f ) \circ C_2}$.
    \item Assume that $C_1$ is not equivalent to $C_2$ modulo affine rational linear global phase.
          Then by definition, $\interp{z I \circ R_I( f ) \circ C_2} = \interp{z I} \interp{R_I( f )} \interp{C_2} = z e^{-if( \theta ) / 2} \interp{C_2}$.
          By \cref{Lem:FindPhase}, there exists some $\beta \in \mathbb{R}$ such that $z = e^{i\beta}$.
          Then $\interp{z I \circ R_I( f ) \circ C_2} = e^{-i( f(\theta) - 2\beta ) / 2}$
          Since $C_1$ is not equivalent to $C_2$ modulo affine rational linear global phase, then in particular, $\interp{C_1} \ne e^{-i( f( \theta ) - 2\beta ) / 2} \interp{C_2} = \interp{z I \circ R_I( f ) \circ C_2}$.
    \end{itemize}
    This completes the proof.
\end{proof}

\section{Proof of Theorem~\ref{Thm:MinDenom}}

\newcommand{\numerator}{\operatorname{Numerator}}
The section establishes~\cref{Thm:MinDenom}.
The idea of the proof is fairly simple.
First, let $d = \lcm\{ \denom( s ) : s \in S \}$.
Then for each element $s \in S$, there will be a unique numerator $x_s \in \mathbb{Z}$ such that $x_s / d = s$.
Then the size of $S$ will be bounded by the number of unique numerators within the given interval.
To this end, define $\numerator( s, d )$ to be the unique integer $x_s$ such that $x_s / d = s$.

\begin{lemma}
    \label{Lem:MinDenom}
    If $k \in \mathbb{K}$, $S \subseteq [0,k) \cap \mathbb{Q}$, and $d = \lcm\{ \denom( s ) : s \in S \}$, then $|S| = |X|$ where $X = \{ \numerator( s, d ) : s \in S \}$.
\end{lemma}

\begin{proof}
    Let $f: S \to \mathbb{Z}$ such that $f( s ) = \numerator( s, d )$.
    Then $X = f( S )$ and $|X| \le |S|$.
    Then it suffices to show that $f$ is injective.
    Let $s \in S$ and $t \in S$.
    Assume that $f( s ) = f( s )$.
    By definition of $f$, $s = f( s ) / d$ and $t = f( t )$.
    Then $s = t$.
    Since $s$ and $t$ were arbitrary, then $f$ is injective.
    Then $|X| \ge |S|$.
    In conclusion, $|S| = |X|$.
\end{proof}

\mindenom*
\begin{proof}
    Define $d = \lcm\{ \denom( s ) : s \in S \}$ and $X = \{ \numerator( s, d ) : s \in S \}$.
    Let $x \in X$.
    Assume for the intent of contradiction that $x \not \in \{ 0, 1, \ldots, dk - 1 \}$.
    Since $x$ is an integer, then there are two cases to consider.
    \begin{enumerate}
    \item Assume that $x$ is negative.
          Then there exists an $s \in S$ such that $s = x / d$.
          Since $d$ is positive, then $s$ is negative.
          However, $S \subseteq [0, 4)$, so $s$ is positive.
          By contradiction, $s$ is not negative.
    \item Assume $x \ge dk$.
          Then there exists an $s \in S$ such that $s = x / d$.
          Since $x \ge dk$, then $s \ge k$.
          However, $S \subseteq [0, 4)$, so $s < 4$.
          By contradiction, $s < dk$.
    \end{enumerate}
    Since $x$ was arbitrary, then $X \subseteq \{ 0, 1, \ldots, dk - 1 \}$.
    Then $|X| < dk$.
    Then by \cref{Lem:MinDenom}, $|S| = |X| < dk$.
    Then $b / k \le d$.
    Since $d$ is an integer, then $\ceil{b / k} \le d$.
\end{proof}

\end{document}